\pdfoutput=1
\documentclass{article}
\pdfoutput=1
\newif\iflncs
\lncsfalse

\usepackage{amsmath,amsfonts,amssymb,complexity}

\iflncs
\usepackage{breakcites}
\usepackage[hyperfootnotes=false,hypertexnames=false,colorlinks=true,linkcolor=blue,citecolor=blue,urlcolor=blue]{hyperref}
\else
\usepackage{amsthm}
\usepackage{fullpage}
\usepackage[hyperfootnotes=false,hypertexnames=false,colorlinks=true,linkcolor=Maroon,citecolor=Maroon,urlcolor=Maroon]{hyperref}
\fi

\usepackage{enumitem}
\usepackage{booktabs}
\usepackage{multirow}
\usepackage{mdframed}
\usepackage[dvipsnames]{xcolor}              
\usepackage{graphicx}
\usepackage[font=small,labelfont=bf]{caption}
\usepackage{thm-restate}
\usepackage{orcidlink}                     
\usepackage{bbm}
\usepackage[normalem]{ulem}
\usepackage{tikz}
\usepackage{mathrsfs}
\usepackage{cancel}
\usepackage{float}
\usepackage[capitalize]{cleveref}            

\newcommand{\remove}[1]{}

\iffalse
\newcommand{\snote}[1]{[{\footnotesize \color{blue}{\bf Sagnik:} { {#1}}}]}
\newcommand{\nnote}[1]{[{\footnotesize \color{teal}{\bf Nikolaj:} { {#1}}}]}
\newcommand{\pnote}[1]{[{\footnotesize \color{red}{\bf Prashant:} { {#1}}}]}
\newcommand{\prashant}[1]{{\small \authornote{Prashant}{#1}{red}}}
\newcommand{\shweta}[1]{{\small \authornote{Shweta}{#1}{blue}}}
\newcommand{\sagnik}[1]{{\small \authornote{Sagnik}{#1}{brown}}}
\newcommand{\akhil}[1]{{\small \authornote{Akhil}{#1}{teal}}}
\newcommand{\nikolaj}[1]{{\small \authornote{Nikolaj}{#1}{ForestGreen}}}
\newcommand{\authornote}[3]{\textcolor{#3}{[\textsc{#1:} {#2}]}}
\else
\newcommand{\snote}[1]{}
\newcommand{\nnote}[1]{}
\newcommand{\pnote}[1]{}
\newcommand{\prashant}[1]{}
\newcommand{\shweta}[1]{}
\newcommand{\sagnik}[1]{}
\newcommand{\akhil}[1]{}
\newcommand{\nikolaj}[1]{}
\fi

\iflncs
\input{lncs_issues}
\else
\newtheorem{theorem}{Theorem}[section]
\newtheorem{lemma}[theorem]{Lemma}
\newtheorem{claim}{Claim}[theorem]
\newtheorem{conjecture}[theorem]{Conjecture}
\newtheorem{definition}[theorem]{Definition}

\newtheorem{corollary}[theorem]{Corollary}
\newtheorem{remark}[theorem]{Remark}

\Crefname{myclaim}{Claim}{Claims}
\fi

\crefname{conjecture}{Conjecture}{Conjectures}

\iflncs
\renewcommand{\paragraph}[1]{\;\newline \noindent \textbf{#1}}
\fi

\newtheoremstyle{infthm}
{}                
{}                
{\itshape}        
{}                
{}                
{.}               
{ }               
{}                
\theoremstyle{infthm}

\newenvironment{proofof}[1]{\begin{proof}[\textit{Proof of #1}]}{\end{proof}}

\crefname{claim}{Claim}{Claims}
\crefname{proposition}{Proposition}{Propositions}
\crefname{assumption}{Assumption}{Assumptions}

\def\ddefloop#1{\ifx\ddefloop#1\else\ddef{#1}\expandafter\ddefloop\fi}

\def\ddef#1{\expandafter\def\csname bb#1\endcsname{\ensuremath{\mathbb{#1}}}}
\ddefloop ABCDEFGHIJKLMNOPQRSTUVWXYZ\ddefloop

\newcommand{\N}{\bbN}

\def\ddef#1{\expandafter\def\csname c#1\endcsname{\ensuremath{\mathcal{#1}}}}
\ddefloop ABCDEFGHIJKLMNOPQRSTUVWXYZ\ddefloop

\def\ddef#1{\expandafter\def\csname v#1\endcsname{\ensuremath{\boldsymbol{#1}}}}
\ddefloop ABCDEFGHIJKLMNOPQRSTUVWXYZabcdefghijklmnopqrstuvwxyz\ddefloop

\newcommand{\algfont}[1]{\mathsf{#1}}
\def\ddef#1{\expandafter\def\csname alg#1\endcsname{\ensuremath{\algfont{#1}}}}
\ddefloop ABCDEFGHIJKLMNOPQRSTUVWXYZ\ddefloop

\newcommand{\langfont}[1]{\mathnormal{#1}}
\def\ddef#1{\expandafter\def\csname lang#1\endcsname{\ensuremath{\langfont{#1}}}}
\ddefloop ABCDEFGHIJKLMNOPQRSTUVWXYZ\ddefloop

\def\ddef#1{\expandafter\def\csname v#1\endcsname{\ensuremath{\boldsymbol{\csname #1\endcsname}}}}
\ddefloop {alpha}{beta}{gamma}{delta}{epsilon}{varepsilon}{zeta}{eta}{theta}{vartheta}{iota}{kappa}{lambda}{mu}{nu}{xi}{pi}{varpi}{rho}{varrho}{sigma}{varsigma}{tau}{upsilon}{phi}{varphi}{chi}{psi}{omega}{Gamma}{Delta}{Theta}{Lambda}{Xi}{Pi}{Sigma}{varSigma}{Upsilon}{Phi}{Psi}{Omega}{ell}\ddefloop

\newcommand{\paren}[1]      {\left( #1 \right)}
\newcommand{\eps}           {\epsilon}

\DeclareMathOperator{\negl}{negl}

\DeclareMathOperator{\supp}{Supp}

\newcommand{\bitset}{\{0,1\}}

\newcommand{\Nat}           {\mathbb{N}}
\newcommand{\Int}           {\mathbb{Z}}

\newcommand{\F}             {\GF}

\DeclareMathOperator*{\Exp} {\bbE}

\newcommand{\norm}[1]       {\left\| #1\right\|}
\newcommand{\set}[1]        {\left\{ #1 \right\}}
\newcommand{\abs}[1]        {\left| #1\right|}
\newcommand{\size}[1]       {\left| #1\right|}
\newcommand{\floor}[1]      {\left\lfloor #1 \right\rfloor}
\newcommand{\ceil}[1]       {\left\lceil #1 \right\rceil}

\newcommand{\pr}[1]         {\Pr\left[ #1 \right]}
\newcommand{\prob}[2]       {\Pr_{#1}\left[ #2 \right]}

\renewcommand{\exp}[1]      {\Exp\left[ #1 \right]}

\renewcommand{\choose}[2]   {{\dbinom{#1}{#2}}}

\newcommand{\pfrac}[2]      {\paren{\frac{#1}{#2}}}

\newcommand{\secp}{r} 

\newcommand{\ksum}{k\text{-}\mathrm{SUM}}

\newcommand{\keygen}{\mathsf{KeyGen}}

\newcommand{\enc}{\mathsf{Enc}}
\newcommand{\dec}{\mathsf{Dec}}

\usetikzlibrary{decorations.pathreplacing}
\newcommand{\one}{\mathbbm{1}}
\newcommand{\Otilde}{\widetilde{{\O}}}

\newcommand{\sample}{\overset{\$}{\gets}}
\renewcommand{\A}{\mathcal{A}}
\newcommand{\B}{\mathcal{B}}
\renewcommand{\E}{\mathbb{E}}
\renewcommand{\F}{\mathbb{F}}

\renewcommand{\O}{\mathcal{O}}

\newcommand{\Var}{\mathrm{Var}}
\newcommand{\Cov}{\mathrm{Cov}}
\newcommand{\Det}{\mathrm{Det}}
\renewcommand{\th}{\mathrm{th}}
\renewcommand{\det}{\mathrm{det}}
\newcommand{\obf}{\mathrm{obf}}
\newcommand{\rec}{\mathrm{rec}}
\newcommand{\amp}{\mathrm{amp}}
\newcommand{\Ber}{\mathrm{Ber}}
\newcommand{\Std}{\mathrm{Std}}

\renewcommand{\polylog}{\mathrm{polylog}}
\renewcommand{\poly}{\mathrm{poly}}
\renewcommand{\ceil}[1]{\left\lceil#1\right\rceil}
\renewcommand{\supp}{\mathrm{supp}}

\newcommand{\KeyGen}{\mathsf{KeyGen}}
\newcommand{\Enc}{\mathsf{Enc}}
\newcommand{\Dec}{\mathsf{Dec}}

\def\ddef#1{\expandafter\def\csname bb#1\endcsname{\ensuremath{\mathbb{#1}}}}
\ddefloop ABCDEFGHIJKLMNOPQRSTUVWXYZ\ddefloop

\def\ddef#1{\expandafter\def\csname sf#1\endcsname{\ensuremath{\mathsf{#1}}}}
\ddefloop ABCDEFGHIJKLMNOPQRSTUVWXYZ\ddefloop

\def\ddef#1{\expandafter\def\csname c#1\endcsname{\ensuremath{\mathcal{#1}}}}
\ddefloop ABCDEFGHIJKLMNOPQRSTUVWXYZ\ddefloop

\def\ddef#1{\expandafter\def\csname vec#1\endcsname{\ensuremath{\mathbf{#1}}}}
\ddefloop abcdefghijklmnopqrstuvwxyz\ddefloop

\def\ddef#1{\expandafter\def\csname mat#1\endcsname{\ensuremath{\mathbf{#1}}}}
\ddefloop ABCDEFGHIJKLMNOPQRSTUVWXYZ\ddefloop

\newcommand{\sampleU}{\overset{\$}{\gets}} 

\newcommand{\pk}{\mathsf{pk}}
\newcommand{\sk}{\mathsf{sk}}
\newcommand{\ct}{\mathsf{ct}}

\newcommand{\pke}{\mathsf{PKE}}

\newcommand{\lwe}{{{\sf LWE}}}
\newcommand{\LWE}{{\lwe }}
\newcommand{\kSUM}{{k\text{-SUM}}}
\newcommand{\tvkSUM}{{targeted vector $k$-SUM}}
\newcommand{\vkSUM}{{vector $k$-SUM }}

\newcommand{\kXOR}{{k\text{-XOR}}}

\newcommand{\indcpa}{\mathsf{IND}\mbox{-}\mathsf{CPA}}

\newcommand{\pkeGame}{\mathsf{PKEGame}}

\newcommand{\Z}{\mathbb{Z}}
\renewcommand{\G}{\mathsf{G}}
\renewcommand{\R}{\mathbb{R}}

\def\ddefloop#1{\ifx\ddefloop#1\else\ddef{#1}\expandafter\ddefloop\fi}

\def\ddef#1{\expandafter\def\csname bb#1\endcsname{\ensuremath{\mathbb{#1}}}}
\ddefloop ABCDEFGHIJKLMNOPQRSTUVWXYZ\ddefloop

\def\ddef#1{\expandafter\def\csname sf#1\endcsname{\ensuremath{\mathsf{#1}}}}
\ddefloop ABCDEFGHIJKLMNOPQRSTUVWXYZ\ddefloop

\def\ddef#1{\expandafter\def\csname c#1\endcsname{\ensuremath{\mathcal{#1}}}}
\ddefloop ABCDEFGHIJKLMNOPQRSTUVWXYZ\ddefloop

\def\ddef#1{\expandafter\def\csname vec#1\endcsname{\ensuremath{\mathbf{#1}}}}
\ddefloop abcdefghijklmnopqrstuvwxyz\ddefloop

\def\ddef#1{\expandafter\def\csname mat#1\endcsname{\ensuremath{\mathbf{#1}}}}
\ddefloop ABCDEFGHIJKLMNOPQRSTUVWXYZ\ddefloop

\def\ddef#1{\expandafter\def\csname vec#1\endcsname{\ensuremath{\boldsymbol{\csname #1\endcsname}}}}
\ddefloop {alpha}{beta}{gamma}{delta}{epsilon}{varepsilon}{zeta}{eta}{theta}{vartheta}{iota}{kappa}{lambda}{mu}{nu}{xi}{pi}{varpi}{rho}{varrho}{sigma}{varsigma}{tau}{upsilon}{phi}{varphi}{chi}{psi}{omega}{Gamma}{Delta}{Theta}{Lambda}{Xi}{Pi}{Sigma}{varSigma}{Upsilon}{Phi}{Psi}{Omega}{ell}\ddefloop

\renewcommand{\exp}[1]      {\Exp\left[ #1 \right]}
\renewcommand{\choose}[2]   {{\dbinom{#1}{#2}}}

\begin{document}

\title{$k$-SUM in the Sparse Regime}
\iftrue
\iflncs
\else
\author{Shweta Agrawal\thanks{IIT Madras. This work was supported in part by the DST “Swarnajayanti” fellowship, Cybersecurity Center of Excellence, IIT Madras, National Blockchain Project and the Algorand Centres of Excellence programme managed by Algorand Foundation.} \and Sagnik Saha\thanks{Computer Science Department, Carnegie Mellon University.} \and Nikolaj I. Schwartzbach\thanks{
CIFRA Institute, Bocconi University. Supported by the European Research Council (ERC) under the European Union’s Horizon 2020 research and innovation programme (Grant agreement No. 101019547). Part of this work was done while Nikolaj was a Ph.D. student at Aarhus University and was funded by VILLUM FONDEN under the Villum Kann Rasmussen Annual Award in Science and Technology under grant agreement no 17911.
} \and Akhil Vanukuri${}^{*}$ \and Prashant Nalini Vasudevan\thanks{
Department of Computer Science, National University of Singapore. 
Supported by the National Research Foundation, Singapore, under its NRF Fellowship programme, award no. NRF-NRFF14-2022-0010. Part of the work on this project was done when Sagnik and Nikolaj were visiting the National University of Singapore, and were also supported by this award.}}

\fi
\else
\author{}
\date{}
\fi

\pagenumbering{roman}

\maketitle
\thispagestyle{empty}
\begin{abstract}
  
    
      In the average-case $k$-SUM problem, given $r$ integers chosen uniformly at random from $\set{0,\dots,M-1}$, the objective is to find a ``solution'' set of $k$ numbers that sum to $0$ modulo $M$. In the \emph{dense} regime of $M \leq r^k$, where solutions exist with high probability, the complexity of these problems is well understood. Much less is known in the \emph{sparse} regime of $M\gg r^k$, where solutions are unlikely to exist. 
      
      In this work, we initiate the study of the \emph{sparse} regime for $k$-SUM and its variant $k$-XOR, especially their \emph{planted} versions, where a random solution is planted in a randomly generated instance and has to be recovered. We provide evidence for the hardness of these problems and suggest new applications to cryptography. Our contributions are summarized below. \smallskip

\begin{itemize}[align=left,leftmargin=*]
    \item[\textbf{Complexity.}] First we study the complexity of these problems in the sparse regime and show: 
        \begin{itemize} 
        \item[$\bullet$] \textit{Conditional Lower Bounds.} Assuming established conjectures about the hardness of average-case (non-planted) $k$-SUM/$k$-XOR when $M = r^k$, we provide non-trivial lower bounds on the running time of algorithms for planted $k$-SUM when $r^k\leq M\leq r^{2k}$. \smallskip 
         \item[$\bullet$] \textit{Hardness Amplification.} We show that for any $M \geq r^k$, if an algorithm running in time $T$ solves planted $k$-SUM/$k$-XOR with success probability $\Omega(1/\polylog(r))$, then there is an algorithm running in time $\Otilde(T)$ that solves it with probability $(1-o(1))$.  This in particular implies hardness amplification for 3-SUM over the integers, which was not previously known. \smallskip  
         
         Technically, our approach departs significantly from existing approaches to hardness amplification, and relies on the locality of the solution together with the group structure inherent in the problem. 
        \item[$\bullet$] \textit{New Reductions and Algorithms.} We provide reductions for $k$-SUM/$k$-XOR from search to decision, as well as worst-case and average-case reductions to the Subset Sum problem from $k$-SUM. Additionally, we present a new algorithm for average-case $k$-XOR that is faster than known worst-case algorithms at low densities. 
        
    \end{itemize}
     
   \item[\textbf{Cryptography.}] 
   We show that by additionally assuming mild hardness of $k$-XOR, we can construct Public Key Encryption (PKE) from a weaker variant of the Learning Parity with Noise (LPN) problem than was known before. In particular, such LPN hardness does not appear to imply PKE on its own --  this suggests that $k$-XOR/$k$-SUM can be used to bridge ``minicrypt'' and ``cryptomania'' in some cases, and may be applicable in other settings in cryptography. 
   
  \end{itemize}

  
\end{abstract}


\newpage
\tableofcontents
\thispagestyle{empty}
\newpage

\pagenumbering{arabic}
\newcommand { \romannumeralcaps }[1]{ \MakeUppercase { \romannumeral #1}}
\section{Introduction}
\label{sec:intro}

In the $k$-SUM problem, given a set of $r$ numbers from $\set{0,\dots,M-1}$, the task is to find a set of $k$ of them that sum to $0$ (modulo $M$), if such a set exists.\footnote{The $k$-SUM problem is usually defined with the sum being over the integers. These variants are equivalent in complexity in the worst-case, and also in the average-case in certain regimes of parameters (see~\cite{BSV21,dinur_keller_klein}).} This problem has been central to studying the complexity of important problems in a variety of domains such as computational geometry, data structures and graph theory ~\cite{abboud_william_lower_bounds,3sum_dynamic_bounds,3sum_comp_geom,3sum_comp_geom_2,3sum_comp_geom_3,3sum_problems_lower_bound}. It also has several applications in cryptanalysis~\cite{wagner,BCJ11}. The naïve algorithm of iterating through all $k$-sets takes time $\mathcal{O}(r^k)$. The ``meet-in-the-middle'' algorithm that computes the sums of all $\lceil k/2 \rceil$-sets and looks for collisions runs in time $\Otilde(r^{\ceil{k/2}})$~\cite{HS74}. Better algorithms are known that are faster than $r^{\ceil{k/2}}$ by a few polylog factors~\cite{3sum_subquadratic,love_triangles,Chan20}. There are also FFT-based algorithms that run in time $\Otilde(M+r)$~\cite{Bringmann17,JW19}, which is faster if $M \ll r^{\ceil{k/2}}$.

The $3$-SUM hypothesis of Gajentaan and Overmars~\cite{3sum_comp_geom} states that it is not possible to do much better than the above -- that any algorithm for the $3$-SUM problem in general takes time at least $r^{2-o(1)}$. This hypothesis has been instrumental in establishing conditional lower bounds on the complexities of a variety of interesting problems. More generally, it is conjectured that any algorithm for $k$-SUM takes time at least $r^{\ceil{k/2}-o(1)}$~\cite{AL13}.

\paragraph{Average-Case $k$-SUM.} In the \emph{average-case} $k$-SUM problem, the $r$ numbers in the input are each chosen uniformly at random from $\set{0,\dots,M-1}$. The characteristics of the problem now change depending on the relative sizes of $M$, $r$, and $k$. In analogy to subset sum~\cite{LO85}, we define the \emph{density} of average-case $k$-SUM as the following ratio:
\begin{align*}
  \Delta = \displaystyle\frac{\log{\choose{r}{k}}}{\log{M}}.
\end{align*}
When this density $\Delta$ is $1$, the expected number of $k$-SUM solutions in an instance is also $1$. In general, the expected number of solutions is approximately $r^{k\,\left(1-1/{\Delta}\right)}$. For values of $\Delta$ larger than $1$ (the \emph{dense} regime), this number is polynomial in $r$, and for $\Delta$ smaller than 1 (the \emph{sparse} regime), the number of solutions is vanishing with $r$. When $k$ is small, the density can be approximated by the ratio $(k\log{r}/\log{M})$ and we will use this simplification in the remainder of this paper\footnote{Please see \cref{rem:density-approx} for a discussion of the accuracy of this approximation.}. 

In the dense regime, several non-trivial algorithms are known for average-case $k$-SUM that are more efficient than the worst-case algorithms. For instance, the ``birthday'' algorithm that computes the sum of random sets of $k/2$ numbers and looks for collisions runs in expected time $\mathcal{O}(\sqrt{M}) = \mathcal{O}(r^{k/(2\Delta)})$ which is $r^{O(1)}$ for densities larger than $k$. For densities larger than $\approx k/(1+\log_2{k})$, Wagner's $k$-tree algorithm~\cite{wagner} solves the problem in time $\Otilde(r)$. 

However, at density $\Delta=1$, when there is only a single solution in expectation, no algorithms are known that outperform the best worst-case algorithms. The average-case $k$-SUM problem at this density is believed to be as hard as the worst-case variant, although no worst-case to average-case reduction is currently known~\cite{Pet15,LLW19,dinur_keller_klein}.

\begin{conjecture}[Average-Case $k$-SUM Conjecture]
\label{conj:ksum-intro}
  Any algorithm for average-case $k$-SUM at density $1$ with constant probability of success has running time at least $\Omega\left(r^{\ceil{k/2}-o(1)}\right)$.
\end{conjecture}

Building on this conjecture, Dinur, Keller and Klein~\cite{dinur_keller_klein} showed lower bounds on the complexity of $k$-SUM at densities in the range $(1,2)$. In particular, their results implied that Wagner's algorithm is optimal for $k=3, 4$, and $5$, for certain densities in this range. LaVigne, Lincoln and Williams~\cite{LLW19} used a decision version of a weaker form of this conjecture to construct fine-grained One-Way Functions.

\paragraph{The Sparse Regime.} In this work, we investigate the average-case complexity of $k$-SUM at densities at most $1$, where random instances are unlikely to have solutions. Unlike the dense regime, i.e. with densities $1$ and higher~\cite{wagner,LLW19,BSV21,dinur_keller_klein}, very little is known about the complexity of average-case $k$-SUM in the sparse regime. Addressing this basic gap in our understanding of $k$-SUM is the goal of the present work. 


In the sparse regime where solutions are unlikely, a meaningful variant of the $k$-SUM problem which we introduce is the \emph{planted} $k$-SUM problem -- here a randomly chosen set of $k$ numbers that sum to $0$ is planted at random locations in a random $k$-SUM instance. There are two problems that arise naturally in this setting:
\begin{itemize}
  \item The \emph{planted search} $k$-SUM problem is to recover a $k$-SUM solution given such an instance (at low densities, with high probability, the planted solution is the only one).
  \item The \emph{planted decision} $k$-SUM problem is to distinguish a random instance with a planted solution from a random instance without a planted solution.
\end{itemize}


Aside from being interesting problems that warrant study in their own right, we are also interested in these problems from the standpoint of applications to cryptography. 



\subsection{Our Results}
\label{sec:results}
In this work, we initiate the study of $k$-SUM and its variants, namely the $k$-XOR and vector $k$-SUM problems, in the sparse regime. In vector $k$-SUM, the elements in the input are vectors from $\bbZ_q^m$ for some $m$, and addition is done over this vector space; $k$-XOR is the special case of $q = 2$. Our results are described below. Please also see \cref{fig:density} for a summary. 


\paragraph{Complexity.}

To begin with, we provide conditional lower bounds on the complexity of the planted $k$-SUM problem in the sparse regime via the following theorem, assuming the hardness of the regular (non-planted) $k$-SUM problem at density $1$.
\begin{theorem}[\cref{cor:lb-1}, \cref{sec:lower_bound}]
  \label{infthm:lb}
  Assuming the average-case $k$-SUM conjecture (\cref{conj:ksum-intro}), any algorithm that solves planted search $k$-SUM at density $\Delta \in \left(\frac12, 1\right]$ with constant success probability has to take time $\Omega\left(r^{k \,\left(1-\frac{1}{2\Delta}\right)-o(1)}\right)$. This generalizes to the $k$-SUM problem defined over any abelian group.
\end{theorem}

To establish the above, we first show a reduction from non-planted search $k$-SUM to planted search $k$-SUM at density $1$. Then, for any $\Delta\in (1/2,1)$,  we reduce planted $k$-SUM at density $1$ to planted $k$-SUM at density $\Delta$. We demonstrate two such reductions, one of which additionally lets us show lower bounds for $k$-SUM assuming the hardness of $k'$-SUM for a different $k'$. Please see \cref{sec:planted_nonplanted_equivalence,sec:lower_bound} for details. Following these lower bounds, our current understanding of the complexity of average-case $k$-SUM at various densities is depicted in \cref{fig:graph}. 

We then connect the complexity of the search $k$-SUM problem to related problems such as the subset sum problem and the decision variant of the $k$-SUM problem. We show an average-case reduction from very sparse planted $k$-SUM over integers to the average-case subset sum problem at low densities, as well as a search-to-decision reduction for planted $k$-SUM that, in particular, carries over the above conditional lower bounds to the decision $k$-SUM problem. Please see \cref{sec:search_to_decision,sec:subset-sum-reduction} for details.

Finally, we show an algorithm for the $k$-XOR problem that, at densities less than $O(1/\sqrt{r})$, is faster than the best known worst-case algorithms. More precisely, we show that the planted search $k$-XOR problem can be solved in time $\Otilde\left(r^k\, \left(\frac{1}{\Delta}\right)^{3-k}\right)$ for any $\Delta \geq 1/r$. Please see \cref{sec:algo} for details.




\definecolor{viridis1}{HTML}{440154}
\definecolor{viridis2}{HTML}{414487}
\definecolor{viridis3}{HTML}{2A788E}
\definecolor{viridis4}{HTML}{22A884}
\definecolor{viridis5}{HTML}{7AD151}
\definecolor{viridis6}{HTML}{FDE725}

\usepgflibrary{plotmarks}
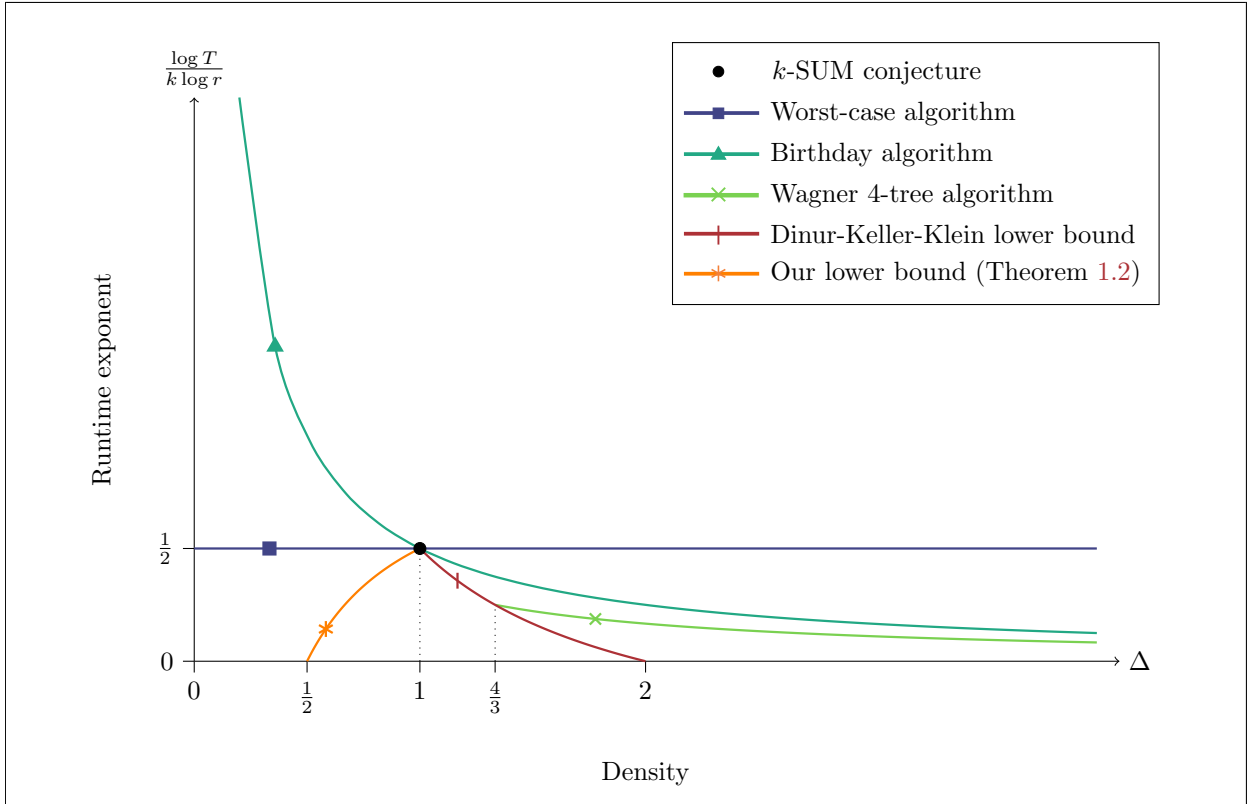
\begin{figure}[h!]
\begin{mdframed}
    \centering
    \vspace{1em}
    \begin{tikzpicture}[scale=3]
      \draw[->] (0, 0) -- (4.1, 0) node[right] {$\Delta$};
      \draw[->] (0, 0) -- (0, 2.5)   node[above] {$\frac{\log T}{k \log r}$};
      \node[align=center] at (2,-0.5) {Density};
      \node[align=center,rotate=90] at (-0.4,1.25) {Runtime exponent};
      
      \draw[-] (0,0)   -- (0,-0.05)    node[below] {$0$};
      \draw[-] (0.5,0) -- (0.5,-0.05)  node[below] {$\frac12$};
      \draw[-] (1,0)   -- (1,  -0.05)  node[below] {$1$};
      \draw[-] (4/3,0) -- (4/3,-0.05)  node[below] {$\frac43$};
      \draw[dotted] (4/3,0)   -- (4/3,0.25);
      \draw[-] (2,0)   -- (2,-0.05)    node[below] {$2$};
      
      \draw[dotted] (1,0)   -- (1,0.5);
      
      \draw[-] (0,0.5)   -- (-0.05,0.5)    node[left]  {$\frac12$};
      \draw[-] (0,0)     -- (-0.05,0)      node[left]  {$0$};
      
      \draw[domain=0:4,line width=0.3mm, smooth, variable=\x, viridis2 , mark=square*, mark size=0.75pt, mark repeat=100, mark phase=3] plot ({\x}, {0.5}); 
      \draw[domain=0.2:4,line width=0.3mm, smooth, variable=\x, viridis4 , mark=triangle*, mark size=1pt, mark repeat=100, mark phase=2] plot ({\x}, {1 / (2*\x)});  
      \draw[domain=4/3:4,line width=0.3mm, smooth, mark=x, mark size=1pt, mark repeat=100, mark phase=5, variable=\x, viridis5] plot ({\x}, {1 / (3*\x)});  
      
      \draw[domain=1:2,line width=0.3mm, smooth, variable=\x, Maroon , mark=|, mark size=1pt, mark repeat=100, mark phase=5] plot ({\x}, {1/\x - 0.5});  
      
      
      \draw[domain=0.5:1,line width=0.3mm, smooth, variable=\x, orange , mark=asterisk, mark size=1pt, mark repeat=100, mark phase=5] plot ({\x}, {1 - 1/(2*\x)});  
      
      \draw[fill](1,0.5) circle[radius=0.75pt];
      
      \tikzstyle{line} = [rectangle,inner sep=0pt, minimum width=1cm,minimum height=0pt,draw]
      \matrix [draw,below left] at (current bounding box.north east) {
          \node [circle,minimum size=4pt,inner sep=0pt,white,fill,label=right:$k$-SUM conjecture] {}; \\
          \node [line,line width=0.5mm,viridis2,label=right:Worst-case algorithm] {}; \\
          \node [line,line width=0.5mm,viridis4,label=right:Birthday algorithm] {}; \\
          \node [line,line width=0.6mm,viridis5,label=right:Wagner 4-tree algorithm] {}; \\
          \node [line,line width=0.5mm,Maroon,label=right:Dinur-Keller-Klein lower bound] {}; \\
          \node [line,line width=0.5mm,orange,label=right:Our lower bound (\cref{infthm:lb})] {}; \\
        };
    \node[black] at (2.323,2.615) {\pgfuseplotmark{*}};
    \node[viridis2] at (2.323,2.43) {\pgfuseplotmark{square*}};
    \node[viridis4,scale=1.5] at (2.323,2.245) {\pgfuseplotmark{triangle*}};
    \node[viridis5,scale=1.8] at (2.323,2.07) {\pgfuseplotmark{x}};
    \node[Maroon,scale=1.8] at (2.323,1.89) {\pgfuseplotmark{|}};
    \node[Orange,scale=1.5] at (2.323,1.72) {\pgfuseplotmark{asterisk}};
    \end{tikzpicture}
    \vspace{0.5em}
    \hrule
    \vspace{1em}
    \caption{Landscape of the known bounds on the complexity of average-case (planted) $k$-SUM problems as $k\rightarrow \infty$ (except for the Wagner algorithm, which is depicted for $k=4$). The $x$-axis represents the density $\Delta = \frac{k \log r}{m}$ of the instances, and the $y$-axis represents the runtime $T$, with $y=\frac{\log{T}}{k\log{r}}$. More specifically, the $y$-axis is the exponent of the runtime, such that if the runtime of an algorithm at density $\Delta$ is $r^{k\alpha}$, we plot the point $(\Delta,\alpha)$.  We have omitted the  algorithm that works at density $\O(\log(r)/r^2)$ presented in \cref{sec:subset-sum-reduction} since its runtime is independent of $k$. Similarly, we have omitted the algorithm for $k$-XOR that works at density $\O(1/r^{0.5+\epsilon})$ presented in \cref{sec:algo}.}
    \vspace{-0.5em}
    \label{fig:graph}
\end{mdframed}
\end{figure}


\paragraph{Hardness Amplification.} For the $k$-SUM and vector $k$-SUM problems (including $k$-XOR) at density $1$ or less, we show that the success probability of an algorithm for the planted search problem can be amplified. We do this using a random walk over instances that preserves the planted solution (with non-trivial probability) and is also rapidly mixing. This amplification also extends to $k$-SUM over general groups, albeit for density slightly smaller than $1$.


\begin{theorem}[\cref{thm: k-Msum success amplification}, \cref{cor:kxor-amp-1}]\label{thm:informal_rec_det_equiv}
  At any density $\Omega\left(\frac{1}{\polylog(r)}\right) \leq \Delta \leq 1$, for any constant $k\geq 3$, suppose there is an algorithm that runs in time $T$ and solves planted search $k$-SUM (resp. vector $k$-SUM) with success probability $\Omega(1/\polylog(r))$. Then, there is an algorithm that runs in time $\Otilde\left(T\right)$ and solves planted search $k$-SUM (resp. vector $k$-SUM) at the same density with success probability $\left(1-o\left(\frac{1}{\log r}\right)\right)$.
\end{theorem}

The above also extends to super-constant values of $k$, with some additional loss in the running time. This hardness amplification, together with a search-to-decision reduction, enables us to use relatively mild hardness of $k$-SUM (or its variants) in applications -- for instance in the public key encryption scheme we construct, it is sufficient for us to assume hardness of solving search $k$-XOR with success probability $(1-o(1/\log{r}))$. Without the hardness amplification, we would have had to assume the hardness of solving it with some $\Omega(1)$ success probability.

\paragraph{Cryptography.}

Next, we show that somewhat mild hardness of planted search $k$-XOR at sufficiently low densities can be used to construct Public-Key Encryption (PKE) assuming weaker hardness of the Learning Parity with Noise (LPN) problem than was known before. Previously, it has been shown how to construct PKE assuming either that LPN with $m$-bit secrets at noise rate $\mathcal{O}(1/\sqrt{m})$ is hard for $\poly(m)$-time algorithms~\cite{Alekhnovich03}, or that LPN with constant noise rate is hard for $2^{m^{0.5}}$-time algorithms~\cite{lpn_pke}. In contrast, adding $k$-XOR enables us to use just $2^{m^{c}}$ hardness (for any constant $c>0$) of LPN with constant noise rate. 

Intriguingly, the level of hardness needed from LPN in our construction does not appear to imply public-key encryption by itself. This suggests the possibility of the $k$-SUM family of problems serving as a bridge for problems from the world of ``Minicrypt'' (where one way functions exist) to the world of ``Cryptomania'' (where public-key encryption exists) --- see also \cite{impagliazzo1995personal}. Qualitatively, our technique allows to interpret the $k$-SUM family of problems as a computational variant of the famous Leftover Hash Lemma \cite{HILL}, which provides statistical guarantees and is used ubiquitously in cryptography \cite{leftover}\footnote{Technically, we are using $k$-XOR as a substitute for a specific strong extractor -- the family of all linear functions $\matA \vecx + \vecb$. Indeed, LHL is more general -- it says any pairwise independent hash family is a strong extractor, but we only replace this specific family with $k$-SUM. However, this family suffices for most applications in cryptography.}. Looking ahead, this can help to not only weaken the required hardness from the ``core'' problem being used in the cryptographic construction, but may also improve overall efficiency of the construction. We demonstrate this phenomenon in two PKE schemes, one based on LPN and another (with lesser improvement) based on its large-field analog Learning With Errors (LWE) \cite{regev05}. Please see \cref{sec:pke} and \cref{sec:pke-lwe} for details. We are optimistic that this technique will find other applications in cryptography.  


\begin{theorem}[\cref{thm:pke_from_lpn}, \cref{sec:pke}]
\label{thm:informal-pke}
  Suppose that constant-noise LPN with an $m$-bit secret is $2^{m^c}$-hard for some constant $c$, and that any algorithm for planted-search $k$-XOR at densities $\Delta = 1/\polylog(r)$ with success probability $(1-o(1/\log{r}))$ has running time at least $r^{\ceil{k/2}-o(1)}$. Then, there is a PKE scheme which is secure against adversaries running in time $r^{\ceil{k/2}-\Omega(1)}$. 
\end{theorem}

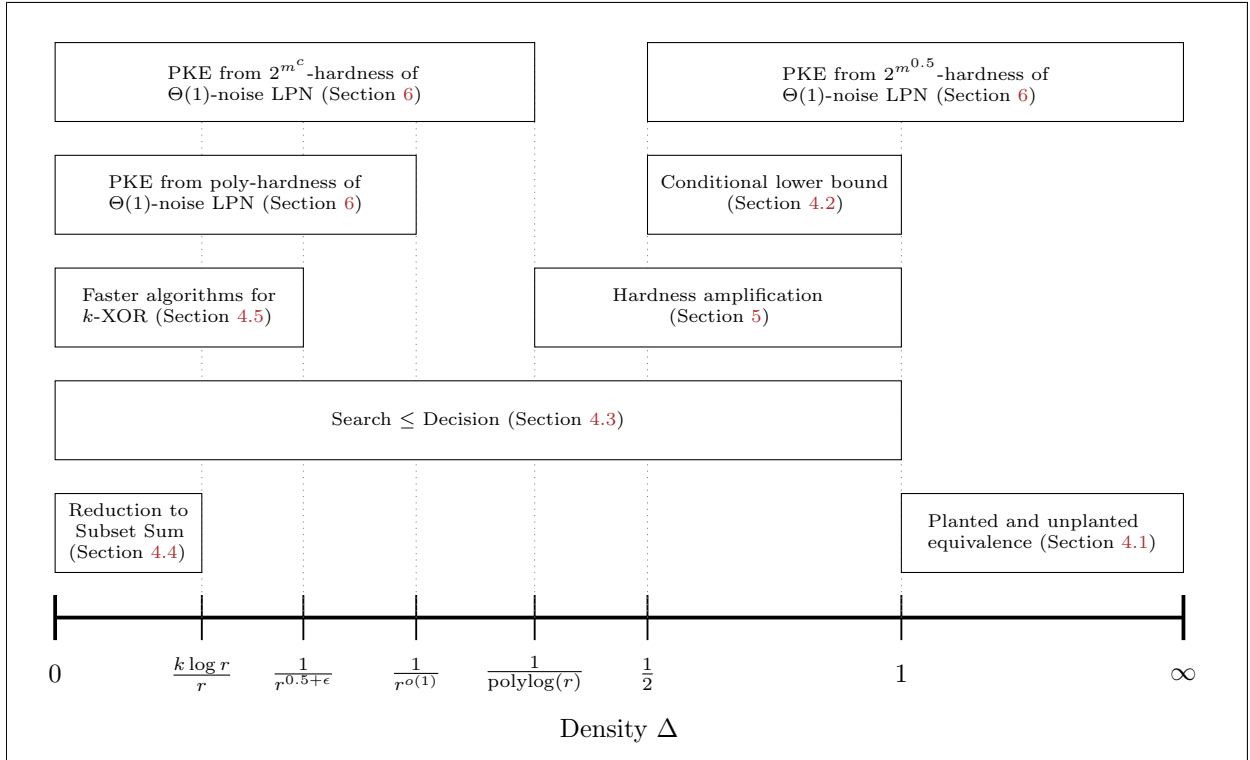
\begin{figure}
\begin{mdframed}
    \centering
    \vspace{1em}
    \begin{tikzpicture}[scale=1.5]
        \node at (5,-1) {Density $\Delta$};
    
        \draw[line width=0.5mm] (0,0) -- (10,0);
        \draw[line width=0.5mm] (0,-0.2) -- (0,0.2);
        \draw[line width=0.5mm] (10,-0.2) -- (10,0.2);
        
        \draw[line width=0.25mm] (7.5,-0.2) -- (7.5,0.2);
        \draw[dotted,color=gray] (7.5,0) -- (7.5,5);
        \node at (7.5,-0.5) {1};
        
        \draw[line width=0.25mm] (5.25,-0.2) -- (5.25,0.2);
        \draw[dotted,color=gray] (5.25,0) -- (5.25,4.4);
        \node at (5.25,-0.5) {$\frac12$};
        
        \draw[line width=0.25mm] (4.25,-0.2) -- (4.25,0.2);
        \draw[dotted,color=gray] (4.25,0) -- (4.25,5);
        \node at (4.25,-0.5) {$\frac1{\polylog(r)}$};
        
        \draw[line width=0.25mm] (3.2,-0.2) -- (3.2,0.2);
        \draw[dotted,color=gray] (3.2,0) -- (3.2,5);
        \node at (3.2,-0.5) {$\frac1{r^{o(1)}}$};

        \draw[line width=0.25mm] (2.2,-0.2) -- (2.2,0.2);
        \draw[dotted,color=gray] (2.2,0) -- (2.2,5);
        \node at (2.2,-0.5) {$\frac1{r^{0.5+\epsilon}}$};
        
        \draw[line width=0.25mm] (1.3,-0.2) -- (1.3,0.2);
        \draw[dotted,color=gray] (1.3,0) -- (1.3,5);
        \node at (1.3,-0.5) {$\frac{k \log r}{r}$};

        \node at (0,-0.5) {0};
        \node at (10,-0.5) {$\infty$};
        
        \draw[draw,fill=white] (7.5,1.1) rectangle node[align=left,font=\scriptsize] {Planted and unplanted \\equivalence (\cref{sec:planted_nonplanted_equivalence})} (10,0.4);
        
        \draw[draw,fill=white] (0,1.1) rectangle node[align=left,font=\scriptsize] {Reduction to \\ \ Subset Sum\\ (\cref{sec:subset-sum-reduction})} (1.3,0.4);
        
        \draw[draw,fill=white] (0,2.1) rectangle node[align=left,font=\scriptsize] {Search $\leq$ Decision (\cref{sec:search_to_decision})} (7.5,1.4);
        
        \draw[draw,fill=white] (4.25,3.1) rectangle node[align=left,font=\scriptsize] {Hardness amplification\\ \quad\quad\ \ (\cref{sec:success amplification})} (7.5,2.4);
        
        \draw[draw,fill=white] (5.25,4.1) rectangle node[align=left,font=\scriptsize] {Conditional lower bound \\\hspace{2.8em} (\cref{sec:lower_bound})} (7.5,3.4);
        
        \draw[draw,fill=white] (5.25,5.1) rectangle node[align=left,font=\scriptsize] {PKE from $2^{m^{0.5}}$-hardness of\\ $\Theta(1)$-noise LPN (\cref{sec:pke})} (10,4.4);
        
        \draw[draw,fill=white] (0,5.1) rectangle node[align=left,font=\scriptsize] {PKE from $2^{m^{c}}$-hardness of\\ $\Theta(1)$-noise LPN (\cref{sec:pke})} (4.25,4.4);
        
        \draw[draw,fill=white] (0,4.1) rectangle node[align=left,font=\scriptsize] {PKE from poly-hardness of\\ $\Theta(1)$-noise LPN (\cref{sec:pke})} (3.2,3.4);

        \draw[draw,fill=white] (0,3.1) rectangle node[align=left,font=\scriptsize] {Faster algorithms for\\ $k$-XOR (\cref{sec:algo})} (2.2,2.4);
    \end{tikzpicture}
    \vspace{0.5em}
    \hrule
    \vspace{1em}
    \caption{Overview of our results in terms of the density of the generated instances (excluding $\Delta \in \{0,\infty\}$), where $r$ is size (number of elements) of the instance. The $x$-axis is not to scale. The PKE schemes for $\Delta \leq 1$ can be constructed from the specified hardness of LPN under the \emph{assumption} that $k$-XOR is hard at the densities depicted on the plot. By contrast, the PKE scheme for $\Delta>1$ works under no assumptions on the hardness of $k$-XOR, and for $\Delta \in (\frac12,1]$ under the standard average-case $k$-SUM conjecture.}
    \label{fig:density}
    \vspace{-0.5em}
\end{mdframed}
\end{figure}

\subsection{Technical Overview}
In this section, we give a high-level overview of some of our results and the techniques we use to show them. 

\paragraph{Relating Planted and Non-Planted $k$-SUM.} We show that any algorithm for planted search $k$-SUM at density $1$ also works for non-planted search $k$-SUM with a small loss in success probability. This follows from showing that the planted and non-planted distributions at densities $1$ and higher are statistically close. At densities somewhat larger than $1$, it is not hard to show that these are, in fact, very close in total variation distance. At density $1$, however, their total variation distance is some constant. Nevertheless, we still show that any algorithm that has any constant success probability over planted instances also has some constant success probability over non-planted instances. We sketch our approach below.

Let $D_0$ denote the distribution of uniformly random $k$-SUM instances, and $D_1$ the distribution of such instances with random planted $k$-SUM solutions. We construct a class of hybrid distributions $D^{(\ell)}$ that ``interpolates'' between these. For $\ell\in[4,r^k]$, the distribution $D^{(\ell)}$ has the following properties:
\begin{itemize}
    \item $D^{(\ell)}$ has total variation distance at most $2/(\ell-3)^2$ from $D_1$
    \item $D^{(\ell)}$ has Rényi divergence\footnote{To be precise, this bound is on the maximum ratio of probability values between the distributions. Technically, the Rényi divergence of order $\infty$ is actually the log of this quantity, but for simplicity, throughout the paper we use the term to refer to this maximum ratio itself instead.} (of order $\infty$) at most $c_k\,(\ell+1)$ relative to $D_0$, where $c_k$ is a constant that depends only on $k$.
\end{itemize}
So if an algorithm succeeds with probability $\eps$ on $D_1$, then it succeeds with probability at least $\eps' = (\eps-2/(\ell-3)^2)$ on $D^{(\ell)}$, and thus with probability at least $\frac{\eps'}{c_k(\ell+1)}$ on $D_0$. Picking an appropriate $\ell=\Theta(1/\sqrt{\eps})$ then gives us what we want. 

 It remains to construct the distribution $D^{(\ell)}$, which is defined as follows. It samples an instance from $X$ from $D_1$, checks if the instance at most $\ell$ solutions: if so, it outputs $X$, and otherwise it outputs a fresh sample from $D_0$. The above bounds on the distances are then shown by expressing the probability mass functions of $D^{(\ell)}$ and $D_1$ in terms of that of $D_0$, and using bounds on the probability of an instance from $D_0$ having more than $\ell$ solutions. 
 
 We briefly mention the relevance of this reduction to cryptography. Previously, it was shown by LaVigne, Lincoln and Williams \cite{LLW19} that the hardness of (a variant of) the planted $k$-SUM problem yields a fine-grained one way function. Our aforementioned reduction (\cref{thm:planted_nonplanted_equivalence}) shows that this can rely on the hardness of the standard {\it non-planted} $k$-SUM problem instead. For more details, please refer to \cref{sec:planted_nonplanted_equivalence}.

\paragraph{Conditional Lower Bounds.} Following the above reduction, the average-case $k$-SUM conjecture implies that planted search $k$-SUM at density $1$ needs at least $r^{\ceil{k/2}-o(1)}$ time. Assuming this, we show lower bounds for lower densities. In more detail, suppose there is an algorithm $\A$ that runs in time $T(r)$ and solves planted search $k$-SUM at some density $\Delta<1$ with constant probability. The idea is, given an instance $X$ at density $1$, to convert $X$ into an instance of density $\Delta$ that still contains the planted solution, and then use $\A$ to recover the solution. We implement this in two different ways.

In the first approach, given an instance $X$ consisting of $r$ elements sampled at density $1$, we choose a random subset $X'$ of $X$ of size $r^\Delta$ (thus reducing the density to $\Delta$) and run $\A$ on $X'$. If we condition on all elements of the solution planted in $X$ being copied to $X'$ in this process, then $X'$ is distributed identically to a planted $k$-SUM instance of size $r^\Delta$ sampled at density $\Delta$. In this case, $\A$ will find this solution with constant probability (over $X$). The event we conditioned on happens with probability at least $\Omega(1 / r^{k\,(1-\Delta)})$, and so if we repeat this process $\mathcal{O}(r^{k\,(1-\Delta)})$ times, it happens at least once with constant probability and we can find a solution in $X$. The $k$-SUM conjecture now implies that $r^{k\,(1-\Delta)}\cdot T(r^\Delta) \geq r^{k/2}$. This, in turn, implies that $T(r) \geq r^{k\,\left(1-\frac{1}{2\Delta}\right)}$, which is the bound we show.  

The second approach is to reduce the density by combining elements in the input and reducing the $k$ in the $k$-SUM problem being considered, not unlike the Wagner $k$-tree algorithm. For example, given instance $X$ for planted $k$-SUM at density $1$, randomly choose $r/4$ disjoint pairs of elements to remove from the instance, compute their sum and put the result back in to get instance $X'$. If it happened that, out of the $k$ elements in the solution in $X$, two were picked as one of these pairs to be combined and the remaining were left untouched, then this leads to a set of $(k-1)$ elements in $X'$ that sum to $0$. Seen as an instance of $(k-1)$-SUM, $X'$ has density $\approx(1-1/k)$, and an algorithm solving it can be used to solve $X$. Computing the probability of this happening then leads to a similar lower bound of $r^{k\,\left(1-\frac{1}{2\Delta}\right)}$ for density $\Delta$, with two important differences. First, it translates between different values of $k$, inferring lower bounds on $k'$-SUM from the hardness of $k$-SUM for $k' \neq k$. This allows us to e.g. establish that solving 4-SUM at density $\Delta=\frac45$ requires $r^{2-o(1)}$ time, assuming that 5-SUM is hard to solve at density $\Delta=1$ (see \cref{cor:lb-2}). Second, the lower bound only works for a discrete set of densities for a given value of $k$, whereas the first lower bound works ``continuously'' as depicted in \cref{fig:graph}. Please see \cref{sec:lower_bound} for details.



\paragraph{Hardness Amplification.} We show that an algorithm that solves planted search $k$-SUM (resp. $k$-XOR) at density in the range $\left({1}/{\polylog(r)},1\right]$ with probability $\Omega(1/\polylog(r))$ in time $T$ implies an algorithm that solves it at the same density with probability $(1-o(1))$ in time $\Otilde(T)$. Our procedure also works for $k$-SUM over general abelian groups, though in this case it only works for densities slightly less than $1$. We briefly describe our approach here, using the specific case of $k$-XOR for illustration; all the steps described below except the final reduction from density $1$ can be applied with any abelian group. 

For simplicity, we will start with the stronger assumption that there is an algorithm $\A$ that solves planted search $k$-XOR with probability $\Omega(1)$ and, further, is deterministic. Let $T_\A\subseteq(\F_2^m)^r$ be the set of $k$-XOR instances for which $\A$ correctly finds a solution; note that $T_\A$ consists of an $\Omega(1)$ fraction of \emph{planted} instances. Our approach, given an instance $X\in(\F_2^m)^r$ with a planted solution, is to find an $X'\in T_\A$ such that a solution for $X$ can be recovered from $\A(X')$. If there is an efficient procedure that finds such an $X'$ given $X$ and fails for at most a $o(1)$-fraction of $X$'s, then this would prove the required amplification. We do this using the following process.

\vspace{1em}
\noindent\underline{$\text{Walk}(X,t)$:}
\begin{itemize}
    \item[1.] Set $X^0\gets X$
    \item[2.] For $i$ from $1$ to $t$:
    \begin{itemize}
        \item[2.1.] Sample $j\gets[r]$
        \item[2.2.] Replace the $j^{\text{th}}$ element $X^{i-1}[j]$ with a random element $x\gets\F_2^m$ such that $x\neq X^{i-1}[j]$
        \item[2.3.] Set $X^i$ to be the resulting instance
    \end{itemize}
    \item[3.] Output $X^t$
\end{itemize}

Consider a graph where each vertex corresponds to an instance in $(\F_2^m)^r$, with an edge between two vertices iff they differ in exactly one column. This is a well-studied graph known as the Hamming graph, here defined over length-$r$ strings and alphabet of size $2^m$. The above process is a $t$-step random walk on this graph starting from the vertex corresponding to $X$. The expansion properties of the Hamming graph imply that random walks of length $\omega(r)$ mix quite well. In other words, with $t=\omega(r)$, for any sets $S$ and $T$ that each contain an $\Omega(1)$ fraction of instances, at least an $\Omega(1)$ fraction of $t$-step random walks that start from $S$ end in $T$. 

With $T = T_\A$, this is reminiscent of what we want -- by the above property, if $T_\A$ contains an $\Omega(1)$-fraction of instances, then the set of instances $X$ from which a constant fraction of walks \emph{do not} lead to $T_\A$ has to be of relative size $o(1)$. There are some issue here, though -- first, this random walk does not preserve solutions so it is not clear how to use it to solve $X$; and second, this graph mostly consists of \emph{non-planted} instances, and $T_\A$ does not actually contain a constant fraction of these. We deal with both of these by considering a conditioning of this random walk.

For simplicity, we restrict our attention to planted instances $X$ that have a unique solution. Denote by $\text{span}(X)$ the set of all instances $X'$ such that the solution in $X'$ appears at the same locations and consists of exactly the same elements as that in $X$. Now, conditioning on all the $X^i$'s being contained in $\text{span}(X)$, the process $\text{Walk}(X,t)$ is again a $t$-step random walk over a Hamming graph, this time defined over $\text{span}(X)$. This conditioned random walk does preserve the solution of $X$, as $X^t$ is now in $\text{span}(X)$. Further, for densities less than $1$, with high probability no additional solutions are introduced during this walk. 

Suppose the fraction of instances in $\text{span}(X)$ that are contained in $T_\A$ is at least $\Omega(1)$. Then the set of $X'\in \text{span}(X) = \text{span}(X')$ for which an $\Omega(1)$-fraction of conditioned $t$-step random walks starting from $X'$ \emph{do} end in $T_\A$ is of relative size at least $(1-o(1))$. For each such $X'$, the event we conditioned on happens with probability at least $(1-k/r)^t$. So for all but a $o(1)$ fraction of $X'\in \text{span}(X)$, the unconditioned $t$-step random walk starting from $X'$ ends in $T_\A$ with probability at least $\Omega\left((1-k/r)^t\right)$. We can set $t= r\cdot \log\log{r}$ so that it is large enough for the walk to mix, and also $(1-k/r)^t = 1/\polylog(r)$ is not too small so that success can then be amplified by repetition. 

It remains to show that the fraction of instances in $\text{span}(X)$ contained in $T_\A$ is at least $\Omega(1)$. We show that this property can be achieved for all but a $o(1)$-fraction of planted instances $X$ by obfuscating the planted solution. Our obfuscation works by sampling a random set of $k$ vectors $E$ that sum to $0$, and then adding a random element from this set to each column of the given instance $X$. With probability at least $1/k^k$, a distinct element from $E$ is added to each element of the solution planted in $X$, and thus the existence and location of the solution are preserved, while the set of vectors that form the solution is fully randomized. The columns of the modified $X$ are then randomly permuted. This ensures that the location of the solution is also randomized.

By repeating the above obfuscation process (and the entire reduction) $\O(k^k)$ times, we can ensure that a solution in $X$ is preserved in at least one of the iterations with high probability. This process hides most properties of the solution and ensures that for most instances $X$, the fraction of $\text{span}(X)$ that is solved correctly by the algorithm is the same, and hence is at least $\Omega(1)$. This entire argument works at every density $\leq 1-\frac{\log \log r}{\log r}$, and in fact works for $k$-SUM over any abelian group. For the cases of $k$-SUM over integers and vector $k$-SUM, we can further extend the result to density $1$ using a couple of other reductions. Please see \cref{sec:success amplification} for details.

\paragraph{Public Key Encryption.} 
Finally, we demonstrate an application of the planted search $k$-XOR problem to cryptography. We construct a Public-Key Encryption (PKE) scheme whose security is based on the hardness of the planted search $k$-XOR problem at low densities together with the hardness of the Learning Parity with Noise (LPN) problem. The hardness required from LPN here is weaker than what was previously known to imply PKE. At a high level, this is possible because the hardness of (decision) $k$-XOR serves as a computational analogue of the leftover hash lemma -- this allows us to set the LPN parameters to result in public keys that are only computationally close to random, rather than statistically close to random, allowing us to weaken the hardness needed from LPN. 

In our construction, we simply generate an instance of planted $k$-XOR and use the result as the public key, with the location of the planted solution being used as the secret key. Note that the secret key can be interpreted as a $k$-sparse vector. The security parameter is the number $r$ of vectors which must be generated. Such an instance can be interpreted as a matrix $X \in \F_2^{m \times r}$ where $m = {k \lg r}/{\Delta}$. We then encrypt a bit as follows. To encrypt zero, we sample a uniform random vector of length $r$. To encrypt one, we take a random linear combination of the rows of the public key, i.e. we sample a random vector $s \gets \F_2^{m}$ and output the ciphertext $s^\top X$. Our hope is that only a recipient who knows the location of planted vector can distinguish $s^\top X$ from a random vector. Unfortunately, this transformation preserves the kernel of $X$ which makes distinguishing between an encryption of zero and one easy. To circumvent this issue, we add i.i.d. noise to each entry of the ciphertext, i.e. we sample $e \gets \Ber_\eta^r$  where $\eta \in (0,1)$ is some {noise parameter}. Distinguishing such a noisy linear combination from a random vector is now hard by LPN, implying indistinguishability of ciphertexts. Decryption follows by using the location of the planted solution to annihilate the large term $s^\top X$ in the encryption of one. Sparsity of the secret key vector ensures that the added noise does not blow up too much.  

The reason the hardness of $k$-XOR helps weaken the assumption on LPN is as follows. Suppose we wish to work with LPN with some constant noise rate $\eta$. In order to be able to decrypt correctly in the above construction, we would need to plant a set of fewer than $k = (1/\eta)$ vectors in the public-key matrix that sum to $0$. Doing so might alter the distribution of the public matrix, whereas the hardness of LPN is only with respect to a public matrix that is uniformly random. If we want the distribution of the planted matrix to be close to uniform, then it needs to at least have enough rows so that sets of $k$ vectors that sum to $0$ occur naturally in the uniform distribution. This ends up requiring around $2^{m^{0.5}}$ rows, and so LPN had to be hard for algorithms running in this time. If decision $k$-XOR was hard, we would not need to rely on this statistical closeness to uniform, and the number of rows in the public matrix can be much smaller while keeping it computationally indistinguishable from uniform. This lets us weaken the hardness required from constant-noise LPN. Please see \cref{sec:pke} for details.

\subsection{Related Work}
\label{sec:related}

The worst-case complexity of the $k$-SUM problem has been studied extensively in the field of fine-grained complexity due to its reductions to a large number of other interesting problems~\cite[\dots]{3sum_comp_geom,3sum_comp_geom_2,3sum_comp_geom_3,3sum_subquadratic,3sum_dynamic_bounds,abboud_william_lower_bounds,3sum_problems_lower_bound,DSW18,Chan20}. We refer the reader to the survey by Williams~\cite{Wil18} for details. The complexity of the $k$-SUM problem in other computational models has also been studied, and it is known to have non-trivial decision trees~\cite{ksum_decision_tree,love_triangles}, non-deterministic algorithms~\cite{nseth}, and lower bounds in some of these models~\cite{Erickson95,AC05}. Questions regarding data structures for it have also been studied~\cite{KP19,GGHPV20,CL23}.

Some conditional bounds for worst-case $k$-SUM are known in certain settings. For super-constant $k$, an algorithm that runs in time $r^{o(k)}$ would contradict the Exponential Time Hypothesis (ETH)~\cite{PW10}. Additionally, an algorithm for $k$-SUM with numbers in the range $\set{0,\dots,M-1}$ that runs in time $M^{1-\Omega(1)}$ would contradict the Strong Exponential Time Hypothesis (SETH)~\cite{ABHS19}.

\paragraph{Average-Case $k$-SUM.} Average-case $k$-SUM and $k$-XOR in the dense regime have several applications in cryptanalysis and has been the subject of substantial work in the area, most involving better algorithms and applications~\cite{wagner,minder_sinclair,NS15,Nandi15,Dinur19,LS19,BDJ21}. 

More recently, different conditional lower bounds have been shown in this regime. Brakerski, Stephens-Davidowitz and Vaikuntanathan~\cite{BSV21} show that Wagner's algorithm is near-optimal for $k$-SUM at large densities as $k$ tends to infinity, using reductions from worst-case lattice problems. Dinur, Keller and Klein~\cite{dinur_keller_klein}, as discussed above, show lower bounds at densities in $(1,2)$ assuming the $k$-SUM conjecture at density $1$. Dalirrooyfard, Lincoln and Williams~\cite{DLW20} show the average-case hardness of counting solutions in a ``factored'' version of $k$-SUM assuming SETH. They also show search-to-decision reductions for the average-case Zero-$k$-Clique problem. 

The study of average-case fine-grained complexity in general has proliferated in the past few years~\cite{average_case_fine_grained_hardness,DLW20}. Of particular relevance here is line of work on worst-case to average-case reductions for counting $k$-cliques, which focuses on reducing from and to the same problem~\cite{GR18,BBB19}. The general paradigm of looking for small hidden solutions in random instances is common in problems studied in statistical inference, such as planted clique, Sparse PCA, etc.~\cite{Jerrum92,sparse_pca1,sparse_pca2,GZ19}. Worst-case versions of these problems have also been subjects of interest in fine-grained complexity~\cite{Wil18,GV21}. 

\paragraph{Hardness Amplification.} Approaches similar to ours for hardness amplification have been used to prove direct product theorems in the past~\cite{IJKW10}, but its use in amplifying the hardness of a fixed natural problem is new. In concurrent independent work, Hirahara and Shimizu~\cite{HS23} use a similar framework to show hardness amplification for the planted clique problem, triangle counting, matrix multiplication, and online matrix-vector multiplication. We briefly describe below the high-level similarities and differences in our approaches. 

Our approach to amplifying the hardness of planted search $k$-SUM/$k$-XOR is as follows. Given an instance, we perform a random walk over instances of the same size where each step consists of adding some noise to the instance and then randomizing it in a way that preserves solutions. We show that the graph defined over the instances by these steps has sufficient expansion properties for the random walk to mix well before the noise added destroys the initial solution. Then, for most instances as starting point, with a large enough probability, the random walk leads to an instance that still has the original solution and at which the weak average-case algorithm is correct.

Hirahara and Shimizu's approach, roughly, is to embed the given instance in a randomized instance of larger size -- note that this never destroys the original solution. They then show, in each of their reductions, that the bipartite graph that captures this random embedding has sufficient expansion properties that again, with most instances as starting point, with a large enough probability, taking a random edge on the bipartite graph leads to a larger instance at which the weak average-case algorithm is correct. This approach is closer to that of Impagliazzo, Jaiswal, Kabanets and Wigderson~\cite{IJKW10}, who also relied similarly on bipartite graphs with expansion properties.

\paragraph{Fine-Grained Cryptography.} The question of constructing cryptographic primitives with fine-grained security guarantees assuming fine-grained hardness conjectures has been studied alongside average-case fine-grained complexity~\cite{BRSV18,LLW19,BC22}. LaVigne, Lincoln and Williams ~\cite{LLW19} use an assumption about the hardness of decision $k$-SUM to construct a fine-grained One-Way Function. They also construct fine-grained Public-Key Encryption (with quadratic security) assuming the average-case hardness of the Zero-$k$-Clique problem. Juels and Peinado~\cite{cliques_owf} similarly constructed One-Way Functions from the conjectured hardness of planted clique for certain parameters. 

Structured problems where a hidden solution can be planted have also been used to construct cryptography in~\cite{ABW10, LLW19}. An immediately relevant illustration of this may be seen in the case of the subset sum problem, which is the unparametrized version of $k$-SUM where the size of the solution is not restricted. The average-case hardness of the planted subset sum problem at very low densities has been used to construct Public Key Encryption by Lyubashevsky, Palacio and Segev~\cite{LPS10}. We stress that our PKE construction based on LPN and $k$-XOR is not a simple modification of this construction. In fact, the appropriate adaptation of their construction to $k$-XOR would be insecure\footnote{Briefly, the construction by \cite{LPS10} relies on the hardness of subset sum (or possibly $k$-SUM) at a density where the number of bits in each element is roughly equal to the number of elements in an instance. At this very low density, $k$-XOR (unlike subset sum or $k$-SUM) can be easily solved using Gaussian elimination (see \cref{sec:algo}).}.

\subsection{Open Problems}

Our work raises multiple interesting questions, some of which we state below. 
\begin{enumerate}
  \item Are there algorithms for planted $k$-XOR at densities in $\left(1 / r^{0.5},1\right)$ that are better than the worst-case algorithms?
  \item Can our conditional lower bounds be improved? In particular, could similar bounds be shown for densities smaller than $1/2$?
  \item Similarly, can conditional lower bounds for search $k$-SUM be shown for densities larger than $2$?
  \item Is there a fine-grained reduction from worst-case $k$-SUM to average-case $k$-SUM at any density?
  \item Can our approach to hardness amplification be applied to other problems in fine-grained complexity?
  \item Can the hardness of $k$-SUM or $k$-XOR help to weaken assumptions made for other cryptographic constructions?
\end{enumerate}

\setcounter{theorem}{0}

\section{Preliminaries}
We denote by $\log x$ the base-2 logarithm of $x$. We denote by $[n] = \{1,2,\ldots,n\}$ the set containing the first $n$ positive integers. We use the notation $X \sample G$ to denote that $X$ is sampled uniformly from $G$ when $G$ is finite. We let $\one[\cdot]$ be the indicator variable for the validity of the statement in the brackets, with 1 denoting true and 0 denoting false. If $A,B$ are two sets, we denote by $A \Delta B = (A \cup B) \setminus (A \cap B)$ the symmetric difference between $A,B$.

We use standard notation for asymptotics, $\mathcal{O}(\cdot), o(\cdot), \Omega(\cdot), \omega(\cdot)$, and use a subscript $\mathcal{O}_k(\cdot)$ to hide factors that only depend on $k$. Similarly, we use the tilde $\Tilde{\mathcal{O}}(\cdot)$ to hide polylogaritmic factors in the main parameter (usually $r$). We say a function $f(\cdot)$ is \emph{negligible} if it grows slower than the inverse of any polynomial, i.e. if for any constant $c$, it holds that $f(x)=o(x^c)$. We denote by $\negl$ a generic negligible function. 

\paragraph{Probability Theory.} If $D$ is a probability distribution on a countable set $\Omega$, and $X \in \Omega$, we denote by $D(X)$ the probability mass of $D$ on $X$. We use the notation $X \sim D$ to denote that $X$ is sampled according to $D$. If $D, D'$ are two probability distributions, we denote by $SD(D,D')$ the total variation distance, defined as,
$$
    SD(D,D') = \frac12 \sum_{X \in \Omega} |D(X) - D(X')|.
$$ 
The total variation distance gives an upper bound on the advantage of any algorithm in distinguishing between the two probability distributions $D, D'$.

\begin{lemma}[Rényi Divergence, \cite{renyi}]\label{lemma:renyi}
    Let $P,Q$ be two probability distributions, with $\supp(P) \subseteq \supp(Q)$, and let $E \subseteq \supp(Q)$ be an event. Then,
    $$
        Q(E) \geq P(E)/R(P\lVert Q),
    $$
    where $R(P\lVert Q)$ is the Rényi divergence (of order $\infty$), defined as, 
    $$
        R(P \lVert Q) = \underset{x \in \supp(P)}{\max} \frac{P(x)}{Q(x)}.
    $$
\end{lemma}

\noindent The Rényi divergence between two distributions can be used to obtain multiplicative bounds on the success probabilities of average-case algorithms whose inputs are sampled from those distributions.

We denote by $\Ber_\eta$ the Bernoulli distribution on support $\F_2$ with parameter $\eta$. In a similar vein, we let $\Ber_\eta^r$ be the distribution of $r$ i.i.d. Bernoulli distributions with support $\F_2^r$ where $X \sim \Ber_\eta^r$ means that $X_i \sim \Ber_\eta$ and that $X_i$ and $X_j$ are independent for $i\neq j$, and likewise for $\Ber_\eta^{m\times r}$ with support $\F_2^{m \times r}$.

\paragraph{Concentration Bounds.}
We will make use of a variety of concentration bounds that we include here for the purpose of self-containment. Markov's inequality gives concentration of a non-negative random variable in terms of its first moment.
\begin{lemma}[Markov's Inequality, \cite{real_analysis}]\label{lemma:markov}
    Let $X$ be a non-negative random variable. Then for every $\varepsilon>0$, 
    $$
        \Pr[X > \varepsilon\,\E[X]] < \frac{1}{\varepsilon}.
    $$
\end{lemma}
\noindent Chebyshev's inequality bounds it in terms of its second moment.
\begin{lemma}[Chebyshev's Inequality, \cite{chebyshev}]\label{lemma:chebyshev}
    Let $X$ be a random variable with finite variance. Then for every $\varepsilon>0$, 
    $$
        \Pr[|X - \E[X]| > \varepsilon\,\Std[X]] < \frac{1}{\varepsilon^2}.
    $$
    where $\Std[X]=\sqrt{\Var[X]}$ is the standard deviation of $X$.
\end{lemma}
\noindent The Paley-Zygmund inequality gives an anti-concentration bound it in terms of its first two moments.
\begin{lemma}[Paley-Zygmund Inequality, \cite{paley_zygmund_1932}]\label{lemma:paley_zygmund}
    Let $X$ be a non-zero random variable with finite variance. Then for every $\varepsilon\in[0,1]$,
    $$
        \Pr[X > \varepsilon\,\E[X]] \geq (1-\varepsilon)^2\, \frac{\E[X]^2}{\E[X^2]}.
    $$
    A slightly stronger (and rewritten) version of the inequality is as follows.
    $$
        \Pr[X > \varepsilon\,\E[X]] \geq \frac{(1-\varepsilon)^2\, \E[X]^2}{\Var[X] + (1-\varepsilon)^2 \,\E[X]^2}
    $$
\end{lemma}
\noindent The Chernoff bounds gives strong concentration for the mean of $n$ i.i.d. 0-1 random variables.
\begin{lemma}[Chernoff Bound, \cite{chernoff}]\label{lemma:chernoff}
    Let $X_1,X_2, \ldots, X_n$ be i.i.d random variables on $\{0,1\}$, and let $X = \sum_{i=1}^n X_i$. Then for every $\varepsilon>0$,
    \begin{align*}
        &\Pr[X > (1+\varepsilon)\,\E[X]] < \mathsf{exp}\left({-\frac{\varepsilon^2\,\E[X]}2}\right),\\
        \mathrm{and,} \quad &\Pr[X < (1-\varepsilon)\, \E[X]] < \mathsf{exp}\left({-\frac{\varepsilon^2\,\E[X]}2}\right).
    \end{align*}
    Similarly, the following also holds,
    \begin{align*}
        &\Pr\left[\frac{1}n X > \frac1n\E[X] + \varepsilon\right] < \mathsf{exp}\left(-2\varepsilon^2\,n\right),\\
        \mathrm{and,} \quad &\Pr\left[\frac{1}n X < \frac1n\E[X] - \varepsilon\right] < \mathsf{exp}\left(-2\varepsilon^2\,n\right).
    \end{align*}
\end{lemma}
\noindent Finally, the Hoeffding also bounds the probability with which a sum exceeds a certain threshold.
\begin{lemma}[Hoeffding's Inequality, \cite{hoeffding}]\label{lemma:hoeffding}
    Let $X_1,X_2,\ldots,X_n$ be independent random variables on $\{0,1\}$, and let $X = \sum_{i=1}^n X$. Then for every $\varepsilon>0$,
    $$
        \Pr\left[X > \E[X] + \varepsilon\right] < \mathsf{exp}\left(-\frac{2 \varepsilon^2}{n}\right).
    $$
\end{lemma}

\paragraph{Spectral Graph Theory.}

We will analyze our construction for the hardness amplification by representing it as a graph and obtain bounds on its edge expansion to argue correctness (see \cref{sec:success amplification}. Formally, an undirected graph $G=(V,E)$ consists of a set of $n$ vertices $V$, with $|V|=n$, and a set of $m$ edges $E \subseteq V \times V$, such that $(u,v) \in E$ iff $(v,u) \in E$. Let $n$ denote the number of nodes, and $m$ the number of edges. If $S,T \subseteq V$, we denote by $E(S,T)$ the set of edges connecting $S$ and $T$, i.e. $(u,v) \in E(S,T)$ iff $u \in S, v \in T$ and $(u,v) \in E$. The degree of a node is the number of edges that includes it. A graph is said to be $d$-regular if all nodes have degree $d$. The graph may also be represented using its adjacency matrix $A \in \F_2^{n \times n}$. Fix any ordering of the vertices and let $(i,j)$ denote the edge between the $i^{th}$ and the $j^{th}$ node. With slight overload of notation, we let $G$ refer also to the $n \times n$ matrix defined as $G_{ij}=\mathbbm{1}[(i,j) \in E]$. A \emph{multigraph} is a graph that is allowed to have multiple edges between the same nodes. We may represent such graphs using matrices of the form $A \in \Nat^{n \times n}$, where the value $A_{ij}$ represents the number of edges from $i$ to $j$. A graph remains a special case of a multigraph where $A_{ij} \in \{0,1\}$ for every $i,j \in [n]$ \cite{bondy1982graph}.

Let $G$ be a multigraph with adjacency matrix $A$. We associate to $G$ the eigenvalues of $A$. Now, let $\lambda_1\geq\lambda_2\geq\cdots\geq\lambda_n$ be the eigenvalues of $G$. We then define the \emph{algebraic expansion} as $\lambda(G) = \max_{i=2\ldots n} |\lambda_i| = \max(|\lambda_2|, |\lambda_n|)$. The algebraic expansion measures the extent to which nodes are connected, with smaller values of $\lambda(G)$ meaning a graph that is more connected. In particular, the following lemma allows us to lower bound the number of edges between any two sets of vertices in terms of $\lambda(G)$.
\begin{lemma}[Expander Mixing Lemma, \cite{expander_mixing_lemma}]\label{lemma:expander_mixing}
    Let $G=(V,E)$ be a $d$-regular graph, and let $S,T \subseteq V$. Then,
    $$
        \left|E(S,T)-\frac{d\cdot|S|\cdot|T|}{|V|}\right| \leq \lambda(G)\sqrt{|S|\cdot|T|}.
    $$
    where $\lambda(G)$ is the algebraic expansion of the graph.
\end{lemma}

\begin{definition}[Hamming Graphs]\label{def:hamming graph} Fix a set $Q$ with $|Q|=q$. The \emph{Hamming graph} $H(d,q)$ is defined as the graph $(V,E)$ whose vertex set $V=Q^d = Q \times Q \times \cdots \times Q$ is the Cartesian product of $Q$ with itself $d$ times, where $(u,v) \in E$ if $u$ and $v$ differ in precisely one coordinate, i.e. if there is an index $j \in [d]$ such that $u_i = v_i$ if and only if $i=j$. The graph $H(d,q)$ is a regular graph of diameter $d$, whose eigenvalues can be characterized as follows.
\end{definition}
\begin{lemma}[Hamming Graph Eigenvalues, \cite{distance_regular_graphs}]\label{lemma:hamming}
    The $i^\th$ eigenvalue of the adjacency matrix of $H(d,q)$ satisfies,
    $$
        \lambda_i(H(d,q)) = \left(q\,(d-i) - d\right)^{\binom{d}{i} (q-1)^i}.
    $$
\end{lemma}
\section{The \texorpdfstring{$k$}{}-SUM Problem}
\label{sec:planted_k_sum}

We now formally define the average-case $k$-SUM problem over general abelian groups, and present existing hardness conjectures for certain interesting groups. We start with a framework for discussing the general groups in this setting, and some descriptive quantities we will use for them.

\paragraph{Group Ensembles.} We fix some underlying countably infinite sequence of finite abelian groups, 
$$
    \G=\{G^{(r)}\}_{r\in \mathbb{N}^+},
$$
that we refer to as a \emph{group ensemble}. Informally, an instance of ``size'' $r$ of the $k$-SUM problem over $\G$ will consist of $r$ elements chosen uniformly at random from the group $G^{(r)}$. With slight abuse of notation, we denote the group operation in all of these groups by $+$ (thus removing its dependency on $r$), and trust that it is clear from the context to which group it belongs. For brevity, we may simply refer to the ensemble as $\left\{G^{(r)}\right\}$ (omitting the subscript).

\begin{definition}[Density]
  For any $k\in\Nat$ and group ensemble $\G=\left\{G^{(r)}\right\}$, we define the \emph{$k$-SUM density} of the $r^\th$ group as,
  \begin{equation}
    \Delta_k^{(r)}(\G) = \frac{k\log{r}}{\log |G^{(r)}|}.
  \end{equation}
  We then define the \emph{$k$-SUM density} of the ensemble $\G$ as the limit of $\Delta^{(r)}$ as $r$ tends to infinity, i.e.,
  \begin{equation}
    \Delta_k(\G) = \lim_{r \rightarrow \infty} \Delta_k^{(r)}(\G).
  \end{equation}
   When $k$ is clear from the context, we will simply refer to the above quantity as the \emph{density of $\G$}, and denote it by $\Delta(\G)$ or even $\Delta$.
\end{definition}

\begin{remark}
    \label{rem:density-approx}
    A more natural definition for density, as described in \cref{sec:intro}, is $\Delta(G) = \log{\binom{r}{k}}/{\log\size{G}}$, which corresponds more closely to the expected number of solutions. The above definition, however, is much more convenient to use in analysis, and is still a good approximation to this quantity -- the difference between them is roughly $\mathcal{O}\left(\frac{\Delta \log{k}}{\log{r}}\right)$. So we use this instead, as Dinur, Keller and Klein~\cite{dinur_keller_klein} also implicitly do.
\end{remark}

We only work with group ensembles that have well-defined density, though many of our techniques can be applied to specific groups (rather than all groups in an ensemble) if needed. We will also need the groups in the ensembles to be efficiently sampleable and have group operations that can be efficiently performed. This is both so that hard problems defined over them can be used, and because our reductions sometimes need to sample random group elements.\footnote{Note that $\log\size{G^{(r)}}$ is the number of bits required to represent elements of $G^{(r)}$, and we ask that random elements be sampleable in time quasilinear in this. This asks for a \emph{uniform} algorithm that samples elements for any $G^{(r)}$. All the theorems stated in the paper are for uniform algorithms. All of our reductions are uniform except where they use this group sampler and compute group operations. So if the group ensembles in consideration only have non-uniform samplers and non-uniform algorithms for group operations, the non-uniform versions of our theorems are still true for them.} 

\begin{definition}[Admissibility]
  \label{def:admissible}
  For $k\in\Nat$, a group ensemble $\G = \set{G^{(r)}}$ is \emph{admissible for $k$-SUM} if it satisfies the following properties:
  \begin{itemize}
    \item \emph{Efficient sampling:} There exists an algorithm that, on input $r\in\Nat$, samples a uniformly random group element from $G^{(r)}$ and runs in time $\mathcal{O}\left(\log\size{G^{(r)}} \polylog\log\size{G^{(r)}}\right)$.
    \item \emph{Efficient operations:} There exists an algorithm that, on input $r\in\Nat$ and group elements $g,h\in G^{(r)}$, outputs the result of the corresponding group operation on $g$ and $h$, and runs in time $\mathcal{O}\left(\polylog\size{G^{(r)}}\right)$.
    \item \emph{Convergent density:} $\Delta_k(\G)$ exists and is finite. Further, we have:
    \begin{equation} \label{eq:delta relates to size of G-admissibility}
    \frac{\abs{\Delta_k^{(r)}(\G) - \Delta_k(\G)}}{\Delta_k(\G)} \leq \frac{1}{\log\size{G^{(r)}}}.
  \end{equation}
  \end{itemize}
\end{definition}

If the density $\Delta_k^{(r)}(\G)$ were to be equal to the limit $\Delta_k(\G) = \Delta$, then the size of the group $G^{(r)}$ would have to be equal to $r^{k/\Delta}$. As these quantities are discrete, this exact equality cannot be achieved for arbitrary values of $\Delta$.  Instead, we have the above convergence condition, which ensures that the size of $G^{(r)}$ is always within a factor of $2$ of its ideal value $r^{k/\Delta}$.

All our statements are to be taken to be made only for ensembles that are admissible for $k$-SUM for $k$ that will be clear from the context, and we leave out this specification in the rest of the paper. For most of the paper, we will also ignore the convergence error, and assume that $\Delta_k^{(r)}(\G) = \Delta_k(\G)$ in our analysis. This error is only of size $(\Delta_k(\G)/\log{\size{G^{(r)}}})$. As we almost always work with small values of density, this will not affect our results substantially.
\label{def:special-ensembles}
\paragraph{Special Ensembles.} We now define two classes of group ensembles that will be of particular interest to us. Each class is parameterized by the density $\Delta$ of the ensemble. The first is the class of modular $k$-SUM ensembles, i.e., ensembles associated with the $k$-SUM problem modulo some integer. The ensemble corresponding to density $\Delta > 0$ is defined as follows.
\begin{align}\label{eq: k-MSUM def}
  \G_{\text{$k$-SUM}}^{(\Delta)}=\left\{\Int_{2^{m(r)}}\right\}, \quad \text{where } m(r) = \left\lceil \frac{k \log r}{\Delta} \right\rceil.
\end{align}
Another class that we will pay special attention to is the class of $k$-XOR group ensembles, i.e. $k$-SUM defined over $GF_{2^m}$ for an appropriately chosen $m$. We define it as follows.
\begin{align}
  \G_{\text{$k$-XOR}}^{(\Delta)}=\left\{GF_{2^{m(r)}}\right\}, \quad \text{where } m(r) = \left\lceil \frac{k \log r}{\Delta} \right\rceil.
\end{align}
 We will refer to the $k$-SUM problem over $\G_{\text{$k$-XOR}}^{(\Delta)}$ simply as \emph{the $k$-XOR problem}. We introduce a natural generalisation of $k$-XOR that we call \emph{vector $k$-SUM} which is defined as follows.
\begin{align}
  \G_{\textsc{vector-$(q,k)$-SUM}}^{(\Delta)}=\left\{(\F_q)^{m(r)}\right\}, \quad \text{where } m(r) = \left\lceil \frac{k \log r}{\Delta \lg q} \right\rceil.
\end{align}
The $k$-XOR problem remains a special case with $q=2$. If $q$ is given by the context, we may refer to this problem as simply \emph{vector $k$-SUM}. It may be verified that all three of these ensembles are admissible.

\subsection{The Non-Planted $k$-SUM Problem} 
Fix some $k\in\Nat$, a group ensemble $\G=\{G^{(r)}\}_{r\in\mathbb{N}}$, and define the related ensemble of ``null distributions'' as follows.
\paragraph{Distribution $D_0^{(r)}$}
\begin{enumerate}
  \item Sample $r$ group elements $X_1, X_2, \ldots, X_r$ i.i.d. uniformly at random from $G^{(r)}$
  \item Return $X = (X_1,\dots,X_r)$
\end{enumerate}

\noindent  In the (non-planted) search $k$-SUM problem, given such an $X$, the task is to find a set of $k$ elements in $X$ that sum to zero (the identity element of the group). If $r$ is clear from the context, we may refer to the distribution simply as $D_0$ (omitting the superscript $r$).

\begin{definition}[Non-Planted Search $k$-SUM]
  \label{def:non-planted-search-ksum}
  For $k\in\Nat$ and an ensemble $\G$, an algorithm $\A$ is said to solve the (non-planted) \emph{search $k$-SUM problem over $\G$} with success probability $\eps$ if, on input an instance $X$ of size $r$ it outputs an $S=\A(X)\subseteq[r]$ with $|S|=k$ such that,
\begin{align*}
  \underset{\substack{X \sim D_0^{(r)}\\S\gets \A(X)}}{\Pr}\left[\sum_{i\in S} X_{i} = 0\right] \geq \eps.
\end{align*}
Where the randomness is taken over the instance and the random coins used by $\A$. If $\eps = \Omega(1)$, we simply say that $\A$ solves the search $k$-SUM problem over $\G$.
\end{definition}

We will refer to any set $S$ of size $k$ that satisfies $\sum_{i\in S} X_{i} = 0$ as a \emph{$k$-SUM solution} for $X$ (or a \emph{$k$-XOR solution} for the $k$-XOR problem). Not all instances $X$ necessarily have a $k$-SUM solution. However, if $\G$ has density $1$, there is (asymptotically) at least a constant probability that $X$ drawn from $D_0$ has at least one $k$-SUM solution. Such a solution, can be found in time $\Otilde(r^{\ceil{k/2}})$ using a simple meet-in-the-middle algorithm. So for any ensemble $\G$ of density $1$, there is an algorithm that runs in time $\Otilde(r^{\ceil{k/2}})$ and solves search $k$-SUM over $\G$ with success probability $\Omega(1)$. 

For certain ensembles of density $1$, it is conjectured that it is not possible to do much better than this. That is, that there is no algorithm that is significantly faster that still solves the search $k$-SUM problem with constant success probability. The following conjectures were formalized by Dinur, Keller and Klein~\cite{dinur_keller_klein}, where they are stated to be folklore.\footnote{\label{footnote:integers}To be accurate, Dinur, Keller and Klein state their conjecture for the $k$-SUM problem where the sum is performed over integers (rather than modulo some number as in $\G_{\text{$k$-SUM}}$). They show, however, that $k$-SUM over integers is equivalent to $k$-SUM over $\G_{\text{$k$-SUM}}$ at approximately the same density, roughly implying the conjecture above.} Weaker versions appear in~\cite{LLW19} and \cite{Pet15}.

\begin{conjecture}[Average-Case $k$-SUM Conjecture]
  \label{conj:ksum}
  For any $k\in\Nat$, any algorithm that solves (non-planted) search $k$-SUM over $\G^{(1)}_{\text{$k$-SUM}}$ with constant success probability has expected running time at least $\Omega(r^{\ceil{k/2} - o(1)})$.
\end{conjecture}

\begin{conjecture}[Average-Case $k$-XOR Conjecture]
  \label{conj:kxor}
  For any $k\in\Nat$, any algorithm that solves (non-planted) search $k$-SUM over $\G^{(1)}_{\text{$k$-XOR}}$ with constant success probability has expected running time at least $\Omega(r^{\ceil{k/2} - o(1)})$.
\end{conjecture}

\subsection{The Planted $k$-SUM Problem}
We now define a different distribution --- the planted distribution --- where, again we sample a random instance, but now we additionally plant a solution at random before outputting it. We may define this process formally as follows.

\paragraph{Distribution $D_1^{(r)}$}
\begin{enumerate}
  \item Sample $r$ group elements $X_1, X_2, \ldots, X_r$ i.i.d. uniformly at random from $G^{(r)}$
  \item Choose a random set $S \subseteq [r]$ with $|S|=k$
  \item Let $i \in S$ be the smallest index and let $X_i \gets - \sum_{\underset{j \neq i}{j \in S}} X_j$
  \item Return $X$
\end{enumerate}

\noindent If $r$ is clear from the context, we may refer to the distribution simply as $D_1$ (omitting the superscript $r$).

\begin{remark}
    Another natural distribution to study in the low-density regime is the uniform distribution conditioned on there being at least one solution. However, there is no simple way to sample from this distribution. It is worth noting that, at density $(1-\eps)$ for any constant $\eps>0$, this distribution is statistically close to, but not the same as, the planted distribution. 
\end{remark}

Once again, we may define a (planted) search $k$-SUM  problem for the planted distribution in the same way as we did in \cref{def:non-planted-search-ksum}. 

\begin{definition}[Planted Search $k$-SUM]
  \label{def:planted-search-ksum}
  For $k\in\Nat$ and an ensemble $\G$, an algorithm $\A$ is said to solve the \emph{planted search $k$-SUM problem} over $\G$ with success probability $\eps$ if, on input an instance $X$ of size $r$ it outputs an $S=\A(X)\subseteq[r]$ with $|S|=k$ such that,
\begin{align*}
  \underset{\substack{X \sim D_1^{(r)}\\S\gets \A(X)}}{\Pr}\left[\sum_{i\in S} X_{i} = 0\right] \geq \eps.
\end{align*}
Where the randomness is taken over the distribution $D_1^{(r)}$ and the random coins used by $\mathcal{A}$. If $\eps = \Omega(1)$, we simply say that $\A$ solves the planted search $k$-SUM problem over $\G$.
\end{definition}

Note that the $\Otilde\left(r^{\ceil{k/2}}\right)$-time algorithm mentioned above can solve the planted search $k$-SUM problem with probability $(1-o(1))$. For certain group ensembles with additional structure, the $k$-SUM problem becomes easy to solve at very low densities. For instance, in the $k$-XOR problem, each element in the instance is a vector. If the length of these vectors is larger than $r$, then with high probability the planted solution will be the only linear dependence among these vectors, and can be found by Gaussian elimination (see \cref{sec:algo} for details).




In addition, we also define a decision version of the $k$-SUM problem, which is to distinguish between these above two distributions. Now the algorithm is given a sample from either $D_0^{(r)}$ or $D_1^{(r)}$, and has to guess from which distribution its input was sampled. 

\begin{definition}[Decision $k$-SUM]
  \label{def:decision-ksum}
  For $k\in\Nat$ and an ensemble $\G$, an algorithm $\A$ is said to solve the \emph{decision $k$-SUM problem} over $\G$ with success probability $\eps$ if, for both $b\in\{0,1\}$,
\begin{align*}
  \underset{X \sim D_b^{(r)}}{\Pr}\left[\A(X) = b\right] \geq \eps.
\end{align*}
Where the randomness is taken over the random coins chosen by $\A$. If $\eps = 1-o(1)$, we simply say that $\A$ solves the decision $k$-SUM problem over $\G$.
\end{definition}

An algorithm that randomly guesses can solve the decision $k$-SUM problem with success probability $1/2$, so the interesting task is doing better than this. It follows from our proofs in \cref{sec:planted_nonplanted_equivalence} that at density $1$, the best success probability any algorithm can have is some constant $\eps < 1$. At lower densities, it is possible to achieve success probability that is $(1-o(1))$ by checking whether any solution exists.

In later sections, we will occasionally be `sloppy' with our use of these formal definitions. It will often be the case that the choice of $r$ is fixed and unambiguous, and hence we will sometimes refer to the group ensemble simply as $G$, thus removing the dependence on the superscript. Similarly, we may denote the null distribution as simply $D_0$, or the planted distribution as $D_1$. Similarly, the underlying group ensemble may be implicitly given in terms of the density; when we talk about `sampling at density $\Delta_0$', we refer to a group ensemble $G^{(\Delta_0)}$ that satisfies $\Delta(\G)=\Delta_0$. These group ensembles will often be $G^{(\Delta_0)}_{\text{$k$-SUM}}$ and $G^{(\Delta_0)}_{\text{$k$-XOR}}$, though we may omit formally specifying this and trust it is clear from the context what we mean.

\subsection{Statistics on the Number of Solutions}
We will use the following notation for ease of discussion of the number of solutions in $k$-SUM instances.

\begin{definition}[Number of Solutions]
  \label{def:num-solns}
  For $k\in\Nat$, an ensemble $\G$, and an instance $X = (X_1,\dots,X_r)$ where each $X_i\in G^{(r)}$, we denote by $c^{(k,\G)}(X)$ the number of sets $S\subseteq [r]$ with $\size{S} = k$ such that $\sum_{i \in S} X_i = 0$. When $k$ and $\G$ are clear from context, we simply denote this by $c(X)$. 
\end{definition}
We will now prove certain properties about $c(X)$ for the uniform as well as the planted distribution that will be useful for proving several different results about the $k$-SUM problem.

\begin{lemma}\label{lemma:c-X stats}
    For a vector $X$ sampled from the uniform distribution $D_0$,
    \begin{equation}\label{eq:D-0 expectation}
        \exp{c(X)} = \binom{r}{k}/|G|.
    \end{equation}
    \begin{equation}\label{eq:D-0 variance}
        \Var({c(X)}) = \frac{\binom{r}{k}}{|G|}\left(1-\frac{1}{|G|}\right).
    \end{equation}
    If $X$ is instead sampled from the planted distribution $D_1$,
    \begin{equation}\label{eq:D-1 expectation}
        \exp{c(X)} = 1+\frac{\binom{r}{k}-1}{|G|}.
    \end{equation}
    \begin{equation}\label{eq:D-1 variance}
        \Var({c(X)}) < \frac{\binom{r}{k}}{|G|}\left(1+\frac{2^kk^2}{r}\right).
    \end{equation}
\end{lemma}
\begin{proof}
    For each $S\subset [r]$ with $|S|=k$, let $I_S$ be the indicator random variable for whether $S$ represents a $k$-SUM solution. Formally, 
    $$
        I_S = \one\left[\sum_{i\in S}X_i=0\right].
    $$

    \noindent We will first consider the case where $X$ is sampled uniformly, i.e. $X\sim D_0$. Since $G$ is a finite group, the sum of a set of elements is uniformly random as long as at least one of those elements is chosen randomly. Therefore, each $I_S$ is a Bernoulli random variable with success probability $\frac{1}{|G|}$. Furthermore, since we only consider sets of size $k$, for any two distinct sets $S$ and $T$ we can find some $i$ such that $i\in S$ and $i\notin T$. Since all the group elements are chosen independently, 
    \begin{equation}\label{eq:I-s-t independence}
        \Pr[I_S=0\mid I_T] = \Pr\left[X_i=-\sum_{\substack{j\in S\\j\neq i}}X_j \,\mid\, I_T\right] = \Pr\left[X_i=-\sum_{\substack{j\in S\\j\neq i}}X_j\right] = \Pr[I_S].
    \end{equation}
    This shows that the variables $I_S$ are i.i.d. Bernoulli variables. Note that we can write $c(X)=\sum I_S$. We therefore have
    \begin{equation}
        \exp{c(X)} = \exp{\sum_{\substack{S\subset[r]\\|S|=k}}I_S} = \sum_{{\substack{S\subset[r]\\|S|=k}}}\exp{I_S} = \sum_{{\substack{S\subset[r]\\|S|=k}}}\frac{1}{|G|} = \frac{\binom{r}{k}}{|G|}.
    \end{equation}
    \begin{equation}
        \Var({c(X)}) = \Var\left({\sum_{\substack{S\subset[r]\\|S|=k}}I_S}\right) = \sum_{{\substack{S\subset[r]\\|S|=k}}}\Var({I_S}) = \sum_{{\substack{S\subset[r]\\|S|=k}}}\frac{1}{|G|}\left(1-\frac{1}{|G|}\right)\ = \frac{\binom{r}{k}}{|G|}\left(1-\frac{1}{|G|}\right).
    \end{equation}

    \noindent Now let us consider the sampled distribution; $X\sim D_1$. We denote by $K$ the $k$-tuple where the solution has been planted. By linearity of expectation, we can still calculate the expected number of solutions quite simply.
    \begin{equation}
         \exp{c(X)} = \exp{\sum_{\substack{S\subset[r]\\|S|=k}}I_S} = \sum_{{\substack{S\subset[r]\\|S|=k}}}\exp{I_S} = \exp{I_K}+\sum_{{\substack{S\subset[r]\\|S|=k\\S\neq K}}}\frac{1}{|G|} = 1+\frac{\binom{r}{k}-1}{|G|}.
    \end{equation}
    However, we can no longer calculate the total variance in the same way as before since these variables may not be independent anymore. We will bound $\Var(c(X))$ by arguing that $I_S$ and $I_T$ are independent for \emph{most} $S, T$ pairs. 

    Let $S$ and $T$ be any two distinct subsets of $[r]$, both of which are different from $K$. Since the indicator variables are binary, independence can be shown by proving $\Pr[I_S=1|I_T=1]=\Pr[I_S=1]$. Observe that if $S\nsubseteq T\cup K$, there exists some index $i\in S$ such that $i\notin T\cup K$. In this case, we can just repeat the argument in \cref{eq:I-s-t independence} to establish independence. Similarly, if $T\nsubseteq S\cup K$, we are done as well by symmetry. If $K\nsubseteq S\cup T$, there must be some index $l\in K$ such that $l\notin S\cup T$. Observe that in the definition of the planted distribution, it does not matter which of the $k$ elements in the planted solution is chosen to be replaced. Therefore, without loss of generality, we can assume that $X_l$ was the element replaced during the planting process. However, this implies that all the elements in $S\cup T$ are independent and chosen uniformly at random; this implies the independence of $I_S$ and $I_T$ as before.

    In the following calculations, $S$ and $T$ are always size-$k$ subsets of $[r]$. Note that we can write the variance as follows.
    \begin{align*}
        \Var\left(c(X)\right) &= \Var\left(\sum_S I_S\right) 
        \intertext{Since $I_K$ is always 1, subtracting it from $c(X)$ does not change the variance.}
        &=\Var\left(\sum_{S\neq K} I_S\right)\\
        &=\sum_{S\neq K}\Var(I_S) + \sum_{S\neq K} \sum_{\substack{T\neq K\\ T\neq S}}\Cov(I_S, I_T)
        \intertext{As shown above, the covariance terms are zero unless any two of $S$, $T$ and $K$ contain the third.}
        &= \left(\binom{r}{k}-1\right)\cdot\frac{1}{|G|}\cdot\left(1-\frac{1}{|G|}\right) + \sum_{\substack{\text{$S,T,K$ distinct}\\S\subset T\cup K\\T\subset S\cup K\\K\subset S\cup T}} \Cov(I_S, I_T)
        \intertext{Using the inequalities $\Cov(X,Y)\leq \sqrt{\Var(X)\Var(Y)}$, and $\Var(I_S)=\Var(I_T)<1/|G|$, we can rewrite this as,}
        &<\frac{\binom{r}{k}}{|G|} + \sum_{\substack{\text{$S,T,K$ distinct}\\S\subset T\cup K\\T\subset S\cup K\\K\subset S\cup T}}\frac{1}{|G|}
        \intertext{We now count the number of $S, T$ pairs satisfying these constraints. Since $|S|=|T|=|K|=k$, we must have $|K\setminus S|=|K\setminus T|=|S\setminus K|=l$ for some $l\in[1,k/2]$. For a fixed $l$, we can choose $K\setminus S$ and $K\setminus T$ in $\binom{k}{l}$ ways. We have $\binom{r-k}{l}$ ways of choosing $S\setminus K$.}
        &\leq\frac{\binom{r}{k}}{|G|} + \sum_{l=1}^{k/2}\binom{k}{l}^2\binom{r-k}{l}\cdot\frac{1}{|G|} \\
        &<\frac{\binom{r}{k}}{|G|} + \frac{k}{2|G|}\binom{k}{k/2}^2\binom{r-k}{k/2}\\
        &<\frac{\binom{r}{k}}{|G|} + \frac{k2^k}{2|G|}\binom{r-1}{k-1}\\
        &=\frac{\binom{r}{k}}{|G|}\cdot\left(1+2^kk^2/r\right) && \qedhere
    \end{align*}
\end{proof}
\section{Basic Complexity of $k$-SUM}
\label{sec:complexity}
In this section, we provide various results about the basic complexity of the $k$-SUM problem and its variants in the sparse regime. We first show an equivalence of planted and non-planted $k$-SUM at densities $\Delta \geq 1$. We then show two conditional lower bounds on the runtime of an algorithm that solves planted $k$-SUM at any density $\Delta\in\left(\frac12,1\right]$. We then show a reduction from search to decision at densities $\Delta < 1$. Finally, we show how to solve $k$-XOR efficiently at very low densities.
\subsection{Relating Planted and Non-Planted \texorpdfstring{$k$}{}-SUM}\label{sec:planted_nonplanted_equivalence}

In this section, we prove an equivalence between planted and non-planted $k$-SUM at densities $\geq 1$ for any finite Abelian group. We first show that at $\Delta=1$, any algorithm that solves the planted problem can be used to solve the non-planted problem. The precise theorem we show is the following.

\begin{theorem}[Equivalence at Density 1]\label{thm:planted_nonplanted_equivalence}
For any $k\in\Nat$ and ensemble $\G$ of density $1$, suppose there exists an algorithm that runs in time $T(r)$ and solves \emph{planted} search $k$-SUM over $\G$ with success probability at least $\epsilon(r)$. Then, the same algorithm also solves \emph{non-planted} search $k$-SUM over $\G$ with success probability at least $\epsilon(r)^{3/2}/(21k^k)$.
\end{theorem}
This theorem implies that planting is a fine-grained one-way function assuming the (non-planted) average-case $k$-SUM conjecture over $\G$ holds. To illustrate this, consider the case of $k$-XOR. Let, 
$$
    f : \F_2^{\lceil k \log r \rceil \times r}\times \binom{r}{k} \rightarrow \F_2^{\lceil k \log r \rceil \times r},
$$
be the `planting function' that takes as input a matrix at density 1 -- such that $m = k \log r$ -- and plants a solution in the locations specified by the second input, where these locations are ordered lexicographically among all subsets of $[r]$ of size $k$. For instance, the output $f(1)$ is a random matrix that has a solution in the set $[k]$. Then, assuming the average-case $k$-XOR conjecture, it follows immediately from \cref{thm:planted_nonplanted_equivalence} that $f$ is a fine-grained one-way function that takes $\Otilde(r)$ time to compute, and cannot be inverted by algorithms running in $r^{k/2-\Omega(1)}$ time. This function was also considered by \cite{LLW19} who show that it constitutes a fine-grained one-way function based on a decision version of the $k$-SUM conjecture. In their work, the one-way function relies implicitly on the hardness of planted $k$-SUM, whereas ours can rely on the hardness of non-planted search $k$-SUM.

We observe that in the dense regime, the two distributions are equivalent in a stronger sense. 

\begin{theorem}[Statistical Closeness in Dense Regime]\label{thm:stat_indist}
  Fix some admissible group ensemble $\textsf{G} = \left\{G^{(r)}\right\}_{r \in \bbN}$ and let $D_0^{(r)}$ (resp. $D_1^{(r)}$) be the non-planted (resp. planted) distribution on $r$ group elements. If for some $k>0$, it holds that the density $\Delta=\Delta_k(\textsf{G}\,)>1$, then,
  $$
    SD\left(D_0^{(r)}, D_1^{(r)}\right) = \mathcal{O}\left(\frac{k^k}{r^{k\,\left[1-\frac1\Delta\right]}}\right).
  $$
\end{theorem}

\noindent This means that if we modify $f$ above to have $k = \Theta(\log{r})$ and $\Delta>1$ be some constant, then this planting is an `actual' one-way function against any polynomial-time algorithm assuming the $k$-XOR conjecture is true. This is similar to \cite{cliques_owf} who show that planting a clique of a certain size in an Erdős–Rényi graph also constitutes a one-way function, assuming it is hard to find planted cliques of size $(1+\epsilon) \log n$ for some constant $\epsilon>0$ in an Erdős–Rényi graph of size $n$.

\paragraph{Proof Strategy.} At a high level, we wish to show that at density $\Delta=1$, any algorithm $\A$ that solves the planted $k$-SUM recovery problem with some constant probability $\epsilon>0$ also solves the non-planted $k$-SUM recovery problem with constant probability $\epsilon'>0$ for a possibly different constant $\epsilon'$. To do so, we proceed using a hybrid argument where we define a intermediate distribution, parameterized by some integer $\ell > 0$, whose distance to both $D_0$ and $D_1$ can be bounded. By transitivity, this shows that $D_0$ and $D_1$ are also close and allows us to bound the error probability. In the former case, we are able to bound the Rényi divergence, and in the latter the statistical distance. This allows us to express $\epsilon'$ as an affine function of $\epsilon$, i.e. $\epsilon' = \alpha \epsilon - \beta$ where $\alpha=\alpha(\ell)$ and $\beta=\beta(\ell)$ are functions of $\ell$. We will show that, for each $\epsilon$, there is a choice of $\ell$ such that $\epsilon' > 0$ for sufficiently large $r$, which would conclude the proof. Specifically, we define the following family of probability distributions, 
\paragraph{Distribution $D^{\ell}$}
\begin{enumerate}
    \item Sample $X \sample D_1$.
    \item Let $c(X)$ be the number of solutions.
    \item If $c(X) > \ell$, let $X \sample D_0$.
    \item Output $X$.
\end{enumerate}
Note that this distribution `interpolates' between $D_0$ and $D_1$ - in particular, we have $D^{0} = D_0$ and $D^{\binom{r}{k}} = D_1$.

\begin{lemma}\label{lemma:explicit_pmf}
  For any $X$, if $c(X)$ is the number of solutions in $X$, we have,
    \begin{enumerate}
        \item $D_1(X) = \frac{|G|}{\binom{r}{k}}\,c(X) \,D_0(X)$.
        \item $D^\ell(X) = \left( \one[c(X) \leq \ell]\cdot\frac{|G|}{\binom{r}{k}}\,c(X) + \underset{X \sim D_1}{\Pr}[c(X) > \ell]\right) D_0(X)$.
    \end{enumerate}
\end{lemma}
\begin{proof}
To prove the first statement, we break down the expression for $D_1(X)$ using the definition of the planted distribution as follows.
    \begin{align*}
        D_1(X) &= \underset{\substack{Y\sim D_0\\S\subset[r]\\|S|=k\\i\gets\min(S)}}{\Pr}\left[X_j=Y_j \forall j\neq i \bigwedge X_i=-\sum_{\substack{j\in S\\j\neq i}}Y_j\right]\\
        &= \underset{Y\sim D_0}{\E}\left[\underset{\substack{S\subset[r]\\|S|=k\\i\gets\min(S)}}{\Pr}\left[X_j=Y_j \forall j\neq i \bigwedge X_i=-\sum_{\substack{j\in S\\j\neq i}}Y_j\right]\right].
    \end{align*}
    Observe that the planting process ensures there is at least one solution in the resulting vector. Hence, $c(X)=0\Rightarrow D_1(X)=0=c(X)\cdot D_0(X)$. Now let us assume that $X$ has $c(X)\geq 1$ distinct solutions, and it was obtained by choosing $Y\sim D_0$, $S\subset [r]$ and $i\in S$ in the planting process. Clearly, $S$ can be any of the $c(X)$ solutions of $k$-SUM in $X$, and $i$ is the minimum index in $S$. Since $Y_i$ is completely replaced whereas the other elements in $Y$ remain unchanged, $Y$ can be any vector that agrees with $X$ in all indices other than $i$; there are $|G|$ such vectors corresponding to each possible group element as $Y_i$. Starting from $X$, we can therefore make $|G|\cdot c(X)$ choices for the pair $(Y, S)$. Since all the choices made in the planting process are uniformly random, the probability of any particular pair is $\frac{1}{|G|^r}\frac{1}{\binom{r}{k}}$. Multiplying the two expressions, we get
    \begin{align*}
        D_1(X)
        &=\frac{|G|}{\binom{r}{k}}\cdot c(X)\cdot\frac{1}{|G|^r}
        = \frac{|G|}{\binom{r}{k}}\,c(X)D_0(X)
    \end{align*}
To prove the second statement, observe that $D^\ell(X)$ is the sum of the probability of choosing $X$ in step 1 and that of choosing $X$ in step 3 of the sampling procedure. The first term is clearly 0 if $c(X)>\ell$ (since step 3 would override it in that case) and $D_1(X)$ otherwise. The second term is the product of the probability of re-sampling in step 3 (which is exactly $\underset{X\sim D_1}{\Pr}[c(X)>\ell]$) and the probability of getting $X$ from re-sampling (which is just $D_0(X)$). The statement now follows from adding the two terms and expanding $D_1(X)=\left(|G|/\binom{r}{k}\right)\,c(X)D_0(X)$.
\end{proof}

Next, we will bound the Rényi divergence of $D^\ell$ and $D_0$ using \cref{lemma:renyi} which gives a multiplicative bound on the error. In fact, applying \cref{lemma:explicit_pmf}, it is straight-forward to bound the Rényi divergence for our use-case.
\begin{corollary}\label{cor:renyi}
    $R(D^\ell \lVert D_0) = \left(\frac{|G|}{\binom{r}{k}}\cdot\ell + \underset{X \sim D_1}{\Pr}[c(X) > \ell]\right)$.
\end{corollary}
This establishes that $D^\ell$ is not `too far' from $D_0$ and establishes a multiplicative bound on the error probabilities for an algorithm that solves the hybrid distribution, and the non-planted distribution. Next, we will bound the statistical distance between $D^\ell$ and $D_1$ to get an additive bound.

\begin{lemma}\label{lemma:sd}
For any $\ell>3$, the following two inequalities hold at density $\Delta=1$:
\begin{equation}
    SD(D^\ell, D_1) \leq \underset{X \sim D_1}{\Pr}[c(X) > \ell] \leq  \frac{2}{(\ell-3)^2}
\end{equation}
\end{lemma}

\begin{proof}
    At a high level, our proof strategy is to bound the statistical distance in terms of the probability that a planted instance has at least a certain number of solutions that we can then bound using Chebyshev's inequality by bounding its first two moments. Note that as density $\Delta$ is $1$, we have $\size{G} = r^k > \binom{r}{k}$ if $k \geq 3$.
    \begin{align}\nonumber
        SD&(D^\ell, D_1) \\\nonumber
        &= \frac12 \sum_{X \in G^r} \left\lvert D^\ell(X) - D_1(X) \right\rvert\\\nonumber
        &= \frac{1}{2}\sum_{\underset{c(X) \leq \ell}{X \in G^r}} \underset{X \sim D_1}{\Pr}[c(X) > \ell]\, D_0(X)  + \frac{|G|}{2\binom{r}{k}} \left(\sum_{\underset{c(X) > \ell}{X \in G^r}} \left[c(X) - \underset{X \sim D_1}{\Pr}[c(X) > \ell]\right] D_0(X) \right)\nonumber
        \intertext{Now identify those instances $X$ for which $D^\ell(X) \geq D_1(X)$. This is exactly the probability that $c(X) \leq \ell$ which means the statistical distance is just the difference in probability between $D^\ell(X)$ and $D_1(X)$ which we may also write as follows.}\nonumber
        SD(D^\ell, D_1)&=  \sum_{\underset{c(X) \leq \ell}{X \in G^r}} \underset{X \sim D_1}{\Pr}[c(X) > \ell]\, D_0(X)\\\nonumber
        &= \underset{X \sim D_1}{\Pr}[c(X) > \ell] \cdot\frac{\left\lvert \{X \in G^r \mid c(X) \leq \ell\} \right\rvert}{|G|^r}\\\nonumber
        &= \underset{X \sim D_1}{\Pr}[c(X) > \ell] \cdot \underset{X \sim D_0}{\Pr}[c(X) \leq \ell]\\
        &\leq \underset{X \sim D_1}{\Pr}[c(X) > \ell] \label{eq:SD DL D1 bound}
    \end{align}
    Note that the standard deviation of $c(X)$ when $X\sim D_1$ is less than $\sqrt{2\binom{r}{k}/|G|}$ for large enough $r$ (\cref{lemma:c-X stats}, \cref{eq:D-1 variance}). Furthermore, \cref{eq:delta relates to size of G-admissibility} at density 1 implies $|G^r|\geq r^k/2$. Therefore,
    $$
        \underset{X\sim D_1}{\E}[c(X)]<1+\frac{\binom{r}{k}}{|G|}\leq 1+\frac{2\binom{r}{k}}{r^k}<3,\hspace{10pt} \text{using \cref{lemma:c-X stats}, \cref{eq:D-1 expectation}.}
    $$
    We can now apply Chebyshev's inequality (\cref{lemma:chebyshev}) to get,
    \begin{align*}
        SD&(D^\ell, D_1)\\
        &\leq \Pr_{X\sim D_1}[|c(x)-3|>\ell-3]\\
        &\leq \Pr_{X\sim D_1}\left[|c(X)-\exp{c(X)}|>\frac{\ell-3}{\sqrt{2\binom{r}{k}/|G|}}\Std(c(X))\right]\\
        &\leq \frac{2\binom{r}{k}}{|G|(\ell-3)^2}\\
        &\leq \frac{2}{(\ell-3)^2} \cdot \left(\frac{er}{k}\right)^k \cdot \frac{1}{r^k}\\ 
        &\leq\frac{2}{(\ell-3)^2}&&\qedhere
    \end{align*}

\end{proof}

\noindent This establishes that $D^\ell$ is not `too far' from $D_1$, and implies a bound on the additive error between the success of an algorithm for the planted distribution and its success on $D^\ell$. 
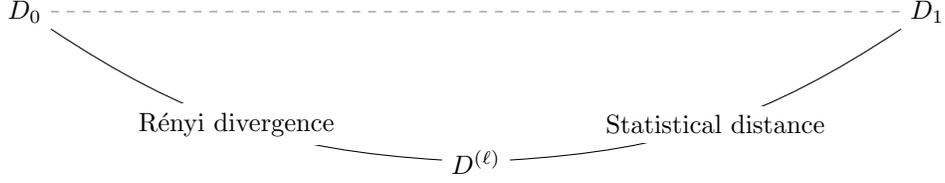
\begin{figure}
    \centering
    \begin{tikzpicture}
        \node (D0) at (0,0) {$D_0$};
        \node (D1) at (12,0) {$D_1$};
        \node (Dl) at (6,-2) {$D^{(\ell)}$};
        \draw[dashed,color=gray] (D0) -- (D1);
        \draw (D0) to[bend right=15] node[fill=white]{Rényi divergence} (Dl);
        \draw (Dl) to[bend right=15] node[fill=white]{Statistical distance} (D1);
    \end{tikzpicture}
    \caption{Depiction of the argument for statistical closeness of $D_0$ and $D_1$. We define a class of hybrid distributions $D^{(\ell)}$ for which we can bound the distance to $D_0$ and $D_1$, specifically we bound the Rényi divergence between $D_0$ and $D^{(\ell)}$ and the statistical distance between $D^{(\ell)}$ and $D_1$. We show that for large enough values of $r$, there is an $\ell$ such that the nonplanted error probability on of any algorithm that solves the planted problem is some constant $>0$.}
\end{figure}
We are now ready to show the main result of this section.

\begin{proofof}{\cref{thm:planted_nonplanted_equivalence}} 
We now prove the main result of this section. Let $\A$ be an algorithm that solves planted $k$-SUM, and let $\epsilon>0$ be a constant that lower bounds its success probability. We shall prove that $\A$ solves non-planted $k$-SUM with probability $\geq \epsilon'$ for some other constant $\epsilon'>0$. Since $\A$ succeeds with probability $\geq\epsilon$ on inputs from $D_1$, it must have a success probability at least $\epsilon-SD(D^\ell,D_1)$ on inputs from $D^\ell$. \cref{lemma:renyi} now implies that the success probability of $\A$ on inputs from $D_0$ must satisfy
\begin{align*}
    \epsilon' &\geq \frac{\epsilon - SD(D^\ell,D_1)}{R(D^\ell \lVert D_0)}
    \intertext{Plugging in values from \cref{lemma:sd,cor:renyi}, we get that,}
    &\geq \frac{\epsilon - \frac{2}{(\ell-3)^2}}{\left(\frac{|G|}{\binom{r}{k}} \cdot \ell + \underset{X \sim D_1}{\Pr}[c(X) > \ell]\right)}\\
    &\geq \frac{\epsilon - \frac{2}{(\ell-3)^2}}{\frac{\ell|G|}{\binom{r}{k}} + \frac{2}{(\ell-3)^2}}
\end{align*}

To ensure that $\epsilon'>0$, solving for $\ell$, we get that $\ell > 3 + \sqrt{\frac{2}{\epsilon}}$. Note that for any $\epsilon$ and $k$, there is a viable such $\ell$ for sufficiently large $r$ (note that $\ell$ is confined to the interval $[0, \binom{r}{k}]$, which concludes the proof. Specifically, we let $\ell = 3 + \frac{2}{\sqrt{\eps}}$ which gives a bound of,
\begin{align*}
    \epsilon' &\geq \frac{\epsilon/2}{\frac{(3\sqrt{\epsilon}+2)|G|}{\binom{r}{k}\sqrt{\epsilon}} + \frac{\epsilon}{2}} 
    \intertext{By \cref{eq:delta relates to size of G-admissibility}, we have $\frac{|G|}{\binom{r}{k}}< \frac{2r^k}{(r/k)^k}=2k^k$. So the above inequality simplifies to}
    &> \frac{\epsilon^{3/2}}{(12\sqrt{\epsilon}+8)k^k+\epsilon^{3/2}} \\
    &>\frac{\epsilon^{3/2}}{21k^k} =\Omega_k\left(\epsilon^{3/2}\right)&&\qedhere
\end{align*}
\end{proofof}

\paragraph{Stronger Equivalence in the Dense Regime.} We now show the second theorem of this section, namely that the two distributions are close in a stronger sense in the dense regime
\begin{proofof}{\cref{thm:stat_indist}}
    Let $M$ be the set of instances without a solution. By \cref{lemma:explicit_pmf}, we know that $D_1$ has the property that $D_1(X) = 0$ for every $X \in M$ and that $D_1(X) \geq D_0(X)$ for every $X \in \supp(D_1)$. Hence, we get that,
\begin{align*}
    SD\left(D_0, D_1\right) = \sum_{X \in M} D_0(X) = \frac{|M|}{|G|^r} &= \underset{X \sim D_0}{\Pr}[c(X) = 0] = 1 - \underset{X \sim D_0}{\Pr}[c(X) > 0]
\end{align*}
We can now use the Paley-Zygmund inequality (\cref{lemma:paley_zygmund}) to get 
\begin{align*}
    \underset{X \sim D_0}{\Pr}[c(X) > 0]
    &\geq \frac{\underset{X \sim D_0}{\E}[c(X)]^2}{\underset{X \sim D_0}{\Var}[c(X)]+\underset{X \sim D_0}{\E}[c(X)]^2}
    \intertext{Substituting the values from \cref{lemma:c-X stats}, \cref{eq:D-0 expectation,eq:D-0 variance}, we get}
    &= \frac{\left(\binom{r}{k}/|G|\right)^2}{\frac{\binom{r}{k}}{|G|}\left(1-\frac{1}{|G|}\right)+\left(\binom{r}{k}/|G|\right)^2}\\
    &=\frac{\binom{r}{k}}{|G|-1+\binom{r}{k}}\\
    &>1-\frac{|G|}{|G|+\binom{r}{k}}
\end{align*}
This means we can upper bound the statistical distance as follows.
\begin{align*}
    SD\left(D_0, D_1\right) < \frac{|G|}{|G|+\binom{r}{k}} <\frac{|G|}{\binom{r}{k}} < \frac{2r^{k/\Delta}}{(r/k)^k} = \mathcal{O}\left({k^kr^{k\left[\frac1\Delta-1\right]}}\right)=\O_k\left({r^{k\left[\frac1\Delta-1\right]}}\right) &\qedhere
\end{align*}

\end{proofof}
\subsection{Conditional Lower Bounds for Sparse \texorpdfstring{$k$}{}-SUM}
\label{sec:lower_bound}
In this section, we establish two different conditional lower bounds for planted $k$-SUM in the sparse regime. In \cref{sec:linear lower bound}, we describe a sparsification procedure on $k$-SUM that reduces the size of the input array to decrease the density of an instance. The resulting reduction establishes a conditional lower bound for recovery and detection that is non-trivial at any density $\Delta\in\left[\frac12,1\right)$. Next in \cref{sec:lower_bound_discrete}, we use a different method to lower density by changing the value of $k$; this gives us non-trivial bounds at some particular densities in $\left[\frac{2}{3},1\right)$.

Before going into further details, let us describe the conditional lower bound by Dinur, Keller and Klein \cite{dinur_keller_klein}. They establish a conditional lower bound for the dense regime $\Delta \in (1,2]$. We describe their reduction at a high level for the case of $G=GF_{2^m}$ with $k$ even for simplicity of exposition.\footnote{See \cref{footnote:integers} in \cref{sec:planted_k_sum} discussing the slightly different definition of the $k$-SUM problem as considered by~\cite{dinur_keller_klein}.} Here, we may interpret the input as a matrix $X \in \F_2^{m \times r}$, with the goal being to find $k$ columns that XOR to the all-zero vector. Their lower bound is established by giving a reduction from an instance of density 1 to a dense instance by removing rows from the instance and giving this instance to a dense oracle. 

Now suppose we wish to convert a density 1 instance to having density $\Delta \in (1,2]$. In order to do this, we need to remove $t = k \log r-m$ rows from the instance. This process introduces $2^m / r^k = r^{k \,\left[\frac1\Delta - 1\right]}$ new solutions in expectation. Hence, ignoring constant factors, assuming that the oracle returns a random solution, we need to invoke the oracle $r^{k\,\left[\frac1\Delta - 1\right]}$ many times to obtain constant success probability. Now suppose the dense oracle takes time $T$, then we can solve a density 1 instance in time $r^{k\,\left[\frac1\Delta - 1\right]}\, T$ which by \cref{conj:ksum}\footnote{Throughout this section, we use a weaker version of \cref{conj:ksum,conj:kxor} that state a lower-bound of $r^{k/2-o(1)}$ rather than $r^{\ceil{k/2}-o(1)}$. This is done for simplicity in our expressions. Note that this relaxation only weakens our lower bounds, which are hence actually stronger than stated for certain values of $k$ and $\Delta$.} must satisfy $r^{k\,\left[\frac1\Delta - 1\right]}\, T \geq r^{k/2-o(1)}$, and hence we must have that $T \geq r^{k\,\left[\frac12 - \frac1\Delta\right]-o(1)}$. This establishes a lower bound for the dense case, and assuming the oracle returns a random solution. However, this is not the case of a malicious oracle as the inputs as described are highly correlated. Thus, the main technical contribution of \cite{dinur_keller_klein} is an obfuscation procedure that ensures the oracle gives (mostly) random responses, whose correctness is analyzed using discrete Fourier analysis. The lower bound they obtain is known to be optimal for $k=3,4,5$.
\begin{theorem}[Dinur, Keller, Klein \cite{dinur_keller_klein}]
  Suppose \cref{conj:ksum} (resp. \cref{conj:kxor}) is true. Then, for $k\in\Nat$ and $\Delta\in(1,2]$, any algorithm that solves search $k$-SUM in $\G_{\text{$k$-SUM}}^{(\Delta)}$ (resp. $\G_{\text{$k$-XOR}}^{(\Delta)}$\,) with constant success probability has to take expected time $\Omega\left(r^{\left[k\,\left(\frac1\Delta-\frac12\right)\right]-o(1)}\right)$.
\end{theorem}

As a first observation, note that this lower bound is easily adaptable to the sparse setting (at least in the case of $k$-XOR). Here, instead of removing rows to increase the density, we will add random rows to lower the density and give the resulting instance to the sparse oracle. Here, we do not need to worry about correlations between instances, as we are not introducing new solutions. In fact, the oracle cannot be malicious as it has to be correct over the randomness of the instance which is distributed exactly according to what it expects. Note that by adding $t$ rows, the original solution is preserved with probability $2^{-t}$ and hence we will have to invoke to oracle $\Omega(2^t)$ times to recover the solution with constant probability. Now suppose we start with a density 1 instance: in order to convert this to a density $\Delta$ instance, we need to add $t = k \log r \left(\frac1\Delta - 1\right)$ such rows. Assuming it takes time $T$ to solve the instance at density $\Delta$, by \cref{conj:ksum} we get a bound of $2^t\, T \geq r^{k/2}$, i.e. $r^{k\,\left[\frac1\Delta-1\right]}\, T \geq r^{k/2}$, and thus $T \geq r^{k\,\left[\frac32 - \frac1\Delta\right]}$ which is non-trivial for $\Delta \geq \frac23$. This reduction establishes the following lower bound.
\begin{theorem}[Follows from techniques in \cite{dinur_keller_klein}]
\label{thm:k-XOR weak lower bound}
    Suppose \cref{conj:kxor} is true. Then, for $k\in\Nat$ and $\Delta\in(\frac23,1]$, any algorithm that solves search $k$-SUM in $G_{\text{$k$-XOR}}^{(\Delta)}$ with constant success probability has to take expected time $\Omega\left(r^{\left[k\,\left(\frac3{2} - \frac1\Delta\right)\right]-o(1)}\right)$.
\end{theorem}

\subsubsection{Lower Bound for Densities \texorpdfstring{$\Delta\in\left(\frac{1}{2},1\right)$}{}}\label{sec:linear lower bound}
In this section, we show how to generalize \cref{thm:k-XOR weak lower bound} to $k$-SUM in arbitrary groups. Specifically, we will prove the following theorem.

\begin{theorem}
\label{thm:lb-1}
  Consider some $k,r\in\Nat$, $\Delta\in\left(\frac12,1\right)$, $\eps\in(0,1]$, and Abelian group $G$. Suppose there is an algorithm that runs in time $T$ and, given an instance of $r^\Delta$ uniformly random group elements from $G$ with a planted $k$-SUM solution, outputs a $k$-SUM solution for it with probability $\eps$. Then, for some constants $c\in\Nat, \eps'\in(0,1]$, there is an algorithm that runs in time $(c\cdot T\cdot r^{k\,(1-\Delta)}\cdot\log(\size{G}))$ that, given an instance of $r$ uniformly random group elements from $G$ with a planted $k$-SUM solution, outputs a $k$-SUM solution for it with probability $\eps'$.
\end{theorem}
\begin{proof}
    Let $\A^{(\Delta)}$ be the algorithm that given $r^\Delta$ random group elements from the group $G$ with a planted solution, outputs a $k$-SUM solution with probability at least $\epsilon$. We then construct the following algorithm $\A^{(1)}$ for recovering $k$-SUM solutions given $r$ random elements from $G$ with a planted solution.
    
    \paragraph{Algorithm $\A^{(1)}(X)$}
    \begin{enumerate}
        \item Repeat $2r^{k\,(1-\Delta)}$ times:
        \begin{enumerate}
            \item[1.1.] Initialize $X'$ to be an empty array.
            \item[1.2.] Randomly choose $r^{\Delta}$ elements from $X$ and copy them to $X'$.
            \item[1.3.] Define $P$ to be the indexing function such that $X'[i]=X[P(i)]$.
            \item[1.4.] Let $K \gets \A^{(\Delta)}(X')$
            \item[1.5.] If $K$ is a solution, return $P(K)$.
        \end{enumerate}
    \end{enumerate}
    
    \noindent By definition of planted $k$-SUM, we know that $X$ has at least one solution $K$. On any given iteration, the probability of all $k$ of those elements being copied to $X'$ is $\frac{r^\Delta}{r}\cdot\frac{r^\Delta-1}{r}\cdots\frac{r^\Delta-k+1}{r} > \frac{r^{k(\Delta-1)}}{2}$. Therefore, the probability that we call $A$ on an array containing all the elements of $\kappa$ at least once is at least, 
    $$
        1-\left(1-\frac{r^{k\,(\Delta-1)}}{2}\right)^{2r^{k\,(1-\Delta)}}\geq 1-\frac{1}{e}=\Omega(1).
    $$ 
    We claim that the probability distribution induced on $X'$ conditioned on the original solution being preserved is just the planted distribution on density $\Delta$. Observe that the elements of $X$ outside $K$ are uniformly i.i.d from $G$, and $K$ independently contains a uniformly random $k$-tuple from $G$ that sums to 0. Therefore, the elements of $X'$ outside $K$ are also uniformly i.i.d from $G$, and $K$ still contains a uniformly random $k$-tuple from $G$ that sums to 0. We can conclude that if $\A^{(\Delta)}$ gets called on an array where the solution is preserved, its input will look like an average-case instance sampled from the planted distribution, and $\A^{(\Delta)}$ will succeed with probability $\epsilon$. The overall success probability of $\A^{(1)}$ is therefore at least $\epsilon \left(1 - \frac1e\right)=\Omega(1)$. The runtime of $\A^{(1)}$ is $\mathcal{O}\left(r^{k\,(1-\Delta)}\,T\right)$, as desired.
\end{proof}
\noindent This reduction immediately gives a lower bound on $k$-SUM in terms of the density.
\begin{corollary}[Conditional Lower Bound]
  \label{cor:lb-1}
  Suppose \cref{conj:ksum} (resp. \cref{conj:kxor}) is true. Then, for $k\in\Nat$ and $\Delta\in\left[\frac12,1\right)$, any algorithm that solves planted search $k$-SUM in $\G_{\text{$k$-SUM}}^{(\Delta)}$ (resp. $\G_{\text{$k$-XOR}}^{(\Delta)}$) with constant success probability has to take expected time $\Omega\left(r^{\left[k\,\left(1-\frac1{2\Delta}\right)\right]-o(1)}\right)$.
\end{corollary}

\subsubsection{Reducing Between \texorpdfstring{$k$}{}-SUM for Different \texorpdfstring{$k$}{}'s}\label{sec:lower_bound_discrete}
In this section, we will present a different sparse conditional lower bound. This bound also applies to any group. Recall that we previously decreased the density by reducing the number of elements in the instance. Instead, now we will reduce the density by compressing elements of the inputs and hope that the resulting instance has a `nice' structure. An interesting feature of this lower bound is that it relates the hardness of $k$-SUM to the hardness of $k'$-SUM at a different density (where $k\neq k'$).

\begin{theorem}\label{thm: discrete lower bound}
  Consider $k_1,k_2\in\Nat$ such that $k_1\geq3$ and $k_2\in [k_1+1, 2k_1-1]$, and an ensemble $\G$ of density $\Delta_1 \leq \frac{k_1}{k_2}$. Suppose there exists an algorithm that runs in time $T(r)$ and solves planted search $k_1$-SUM on $\G$ with constant success probability. Then, there is an algorithm $B$ that runs in time $\mathcal{O}\left(r^{k_2-k_1}T(r)\right)$ and solves planted search $k_2$-SUM on $\G$ with constant success probability.
\end{theorem}


\begin{proof} 
    We start by describing the new algorithm.
    
    \paragraph{Algorithm $B(X)$}
    \begin{enumerate}
        \item Repeat $3^{2k_2}\,r^{k_2-k_1}$ times:
        \begin{enumerate}
            \item[1.1.] Initialize $X'$ to be an empty array.
            \item[1.2.] Randomly choose $\frac{r}{2}$ elements from $X$ and copy them to $X'$.
            \item[1.3.] Randomly split the remaining elements of $X$ into $\frac{r}{4}$ disjoint pairs.
            \item[1.4.] Insert the sums of each of the above $\frac{r}{4}$ pairs into $X'$.
            \item[1.5.] Add $\frac{r}{4}$ random elements of $G$ to $X'$.
            \item[1.6.] Apply a random permutation to $X'$.
            \item[1.7.] Let $S \gets A(X')$.
            \item[1.8.] If $S$ is a solution and it depends on exactly $k_2$ elements of $X$, return those $k_2$ elements.
        \end{enumerate}
    \end{enumerate}

    By definition of planted $k$-SUM, we know that $X$ has at least one solution $S$. Recall that $|S|=k_2$. We are interested in the event where $2k_1-k_2$ of the elements in $S$ were copied directly to $X'$ and the remaining $2k_2-2k_1$ elements of $S$ were paired with each other such that their sums got copied to $X'$. Clearly, this would give rise to a $(2k_1-k_2)+\frac{2k_2-2k_1}{2}=k_1$-SUM solution in $X'$. We call solutions of this type \emph{valid}. The probability of exactly $2k_1-k_2$ elements of $S$ being copied directly to $X'$ in step 1.2. is at least, 
    $$
        \left(\frac{r/2-(2k_1-k_2)}{r}\right)^{2k_2-2k_1}\geq \frac{1}{3^{2k_2}}.
    $$
    The probability that the $2k_2-2k_1$ remaining elements get paired amongst themselves in step 1.3. is at least,
    $
        \left(\frac{1}{r}\right)^{k_2-k_1}
    $.
    Therefore, the probability that we call $A$ on an array containing $k_1$-SUM solution at least once is at least, 
    $$
        1-\left(1-\frac{1}{3^{2k_2}r^{k_2-k_1}}\right)^{3^{2k_2}r^{k_2-k_1}}\geq 1-\frac{1}{e}=\Omega(1)
    $$ 
    We claim that the probability distribution induced on $X'$ conditioned on it having a \emph{valid} $k_1$-SUM solution is just the planted distribution on density $\Delta_1$. Observe that the elements of $X$ outside $S$ are uniformly i.i.d from $G$. Therefore, the $\frac{r}{2}$ elements added in step 1.2. and the $\frac{r}{4}$ elements added in step 1.5. are uniformly i.i.d. from $G$. Since $G$ is a group, the sum of two random elements is also random; this implies that the other $\frac{r}{4}$ elements of $X'$ are uniformly i.i.d. too (excluding the solution). The density of $X'$ is clearly $\frac{k_1\log(r)}{m} = \frac{k_1}{k_2}\cdot\Delta_2=\Delta_1$.

    We can conclude that if $A$ gets called on an array where a valid solution exists, its input will look like an average-case instance sampled from the planted distribution, and it will succeed with constant probability. The overall success probability of $B$ is therefore also a constant. The runtime of $B$ is clearly $\mathcal{O}\left(r^{k_2-k_1}T(r)\right)$, as required.
\end{proof}

\begin{corollary}[Conditional Lower Bound for Different $k$'s]
  \label{cor:lb-2}
  Suppose \cref{conj:ksum} (resp. \cref{conj:kxor}) is true. Then, for $k\in\Nat$ and $\Delta\in\left(\frac12,1\right)$ such that $\Delta = \frac{k}{k'}$ for some $k'\in [k+1,2k-1]$, any algorithm that solves planted search $k$-SUM in $\G_{\text{$k$-SUM}}^{(\Delta)}$ (resp. $\G_{\text{$k$-XOR}}^{(\Delta)}$\,) with constant success probability has to take expected time $\Omega\left(r^{\left[k\,(1-\frac{1}{2\Delta})\right]-o(1)}\right)$.
\end{corollary}

\begin{proof}
  This follows directly from the contrapositive of the previous theorem and the $k$-SUM conjecture. If there is an algorithm at density $\Delta=\frac{k}{k'}$ for $k$-SUM with runtime $T$, \cref{thm: discrete lower bound} implies that we can solve $k'$-SUM at density 1 in time $\mathcal{O}\left(T\,r^{k'-k}\right)$. By assumption, this is at least $\Omega\left(r^{\left[k'/2\right]-o(1)}\right)$. This, along with the definition of $\Delta$, implies $T=\Omega\left(r^{\left[k\,\left(1-\frac{1}{2\Delta}\right)\right]-o(1)}\right)$, as needed.
\end{proof}

\subsection{Search to Decision Reduction}\label{sec:search_to_decision}

A search-to-decision reduction for $k$-SUM is implied by the work of Impagliazzo and Naor~\cite{IN89}. They show a similar reduction for the Subset Sum problem modulo prime numbers or powers of $2$, and their proof can be extended -- using an efficient instantiation of the Goldreich-Levin algorithm~\cite{GL89,Trevisan04} and some minor optimizations -- to obtain the following theorem.

\begin{theorem}[{Search-to-Decision Reduction, implied by \cite{IN89}}]
    \label{thm:search-decision-IN}
    For $k\in\Nat$ and $\Delta < 1$, suppose there is an algorithm that runs in time $T(r)$, and solves the \emph{decision} $k$-SUM problem over $\G_{\text{$k$-SUM}}^{(\Delta)}$ (resp. $\G_{\text{$k$-XOR}}^{(\Delta)}$) with success probability $(1/2+\eps)$. Then, there is an algorithm that runs in time $O(T(r)\cdot r\log{(r)}\cdot (k/\eps)^4)$, and solves the \emph{planted search} $k$-SUM problem over $\G_{\text{$k$-SUM}}^{(\Delta)}$ (resp. $\G_{\text{$k$-XOR}}^{(\Delta)}$) with success probability at least $3/4$. 
\end{theorem}

We show a different search-to-decision reduction for the $k$-SUM problem over general group ensembles that is incomparable to the one above. Whereas the above reduction can work with any decision algorithm that has success probability more than $1/2$, our reduction requires this success probability to be close to $1$. On the other hand, it avoids the factor of $r$ loss in the running time of the above reduction. Our proof is also more elementary, using an algorithm reminiscent of binary search. The precise statement we show is the following. 

\begin{theorem}[Search-to-Decision Reduction]\label{thm:search_to_decision}
For $k\in\Nat$, and ensemble $\G$ of density $\Delta<1$, suppose there is an algorithm that runs in time $T(r) = \Omega(r/\Delta)$, and solves the \emph{decision} $k$-SUM problem over $\G$ with success probability $(1-o(1))$. Then, for any constant $\gamma < 1$, there is an algorithm that runs in time $\Otilde(T(r))$, and solves the \emph{planted search} $k$-SUM problem in $\G$ with success probability at least $\gamma$. 
\end{theorem}

We first describe the reduction at a high level and then prove its correctness. Our algorithm is vaguely related to binary search. We will repeatedly guess a random half of the inputs to replace with fresh random elements, invoke the decision oracle on the resulting instance, and record whether or not the oracle reported there was a solution. Specifically, we will maintain a counter for every element in the original input that we increment whenever an element was found to belong to an unreplaced half of the inputs for which the decision algorithm reported there was a solution. Finally, we output the indices corresponding to the $k$ largest counters. Our hope is that this process is biased in favor of the indices in the solution, and that we do not introduce too many new solutions in the process. 
\paragraph{Algorithm $R(X)$}
\begin{enumerate}
    \item Sample $r/2$ elements from $[r]$ at random (with replacement), and let $S\subseteq [r]$ be the resulting set.
    \item Let $X'_i \gets X_i$ for every $i\not\in S$, and let $X'_i \sample G$ otherwise. 
    \item Return $X'$
\end{enumerate}
We now describe our reduction formally. Let $X \sim D_1$ be some planted instance and let $\A^\det$ be an algorithm that solves the decision problem with probability $1-o(1)$, and let $T$ be its runtime. 

\begin{figure}
    \centering
    \begin{tikzpicture}
        \node (Adet) at (0,-0.25) {$\A^\det$};
        \draw (-1.5,1) rectangle (1.5,-1.5);
        
        \draw (2.5,0.9) rectangle (6.5,0.6);
        \draw[fill=white!60!gray] (2.5,0.9) rectangle (3.5,0.6);
        \draw[fill=white!60!gray] (4,0.9) rectangle (4.5,0.6);
        \draw[fill=white!60!gray] (5,0.9) rectangle (5.5,0.6);
        \draw[->] (2.5,0.75) -- (1.5,0.75);
        \draw[->] (-1.5,0.75) -- (-3,0.75);
        \node at (-3.25,0.75) {1};
        
        \draw (2.5,0.4) rectangle (6.5,0.1);
        \draw[fill=white!60!gray] (2.5,0.4) rectangle (3,0.1);
        \draw[fill=white!60!gray] (4,0.4) rectangle (5.5,0.1);
        \draw[->] (2.5,0.25) -- (1.5,0.25);
        \draw[->] (-1.5,0.25) -- (-3,0.25);
        \node at (-3.25,0.25) {0};
        
        \draw (2.5,-0.1) rectangle (6.5,-0.4);
        \draw[fill=white!60!gray] (3,-0.1) rectangle (4,-0.4);
        \draw[fill=white!60!gray] (4.5,-0.1) rectangle (5,-0.4);
        \draw[fill=white!60!gray] (6,-0.1) rectangle (6.5,-0.4);
        \draw[->] (2.5,-0.25) -- (1.5,-0.25);
        \draw[->] (-1.5,-0.25) -- (-3,-0.25);
        \node at (-3.25,-0.25) {1};
        
        \draw (2.5,-1.1) rectangle (6.5,-1.4);
        \draw[fill=white!60!gray] (3.5,-1.1) rectangle (4,-1.4);
        \draw[fill=white!60!gray] (4.5,-1.1) rectangle (5.5,-1.4);
        \draw[fill=white!60!gray] (6,-1.1) rectangle (6.5,-1.4);
        \draw[->] (2.5,-1.25) -- (1.5,-1.25);
        \draw[->] (-1.5,-1.25) -- (-3,-1.25);
        \node at (-3.25,-1.25) {1};
        
        % ...
        \node at (4.75,-0.65) {$\vdots$};
        \node at (-3.25,-0.65) {$\vdots$};
        
        \draw (-0.5,-2.4) rectangle (3.5,-2.9); 
        \draw(0,-2.4) -- (0,-2.9);
        \draw(0.5,-2.4) -- (0.5,-2.9);
        \draw(1,-2.4) -- (1,-2.9);
        \draw(1.5,-2.4) -- (1.5,-2.9);
        \draw(2,-2.4) -- (2,-2.9);
        \draw(2.5,-2.4) -- (2.5,-2.9);
        \draw(3,-2.4) -- (3,-2.9);
        
        \draw[line width=0.3mm] (-0.5,-2.4) rectangle (0,-2.9);
        \draw[line width=0.3mm] (1,-2.4) rectangle (1.5,-2.9);
        \draw[line width=0.3mm] (2.5,-2.4) rectangle (3,-2.9);
        
        \node at (-1.25,-2.65) {$C:$};
        \node at (-0.25,-2.65) {2};
        \node at (0.25,-2.65) {1};
        \node at (0.75,-2.65)  {1};
        \node at (1.25,-2.65)  {2};
        \node at (1.75,-2.65)  {1};
        \node at (2.25,-2.65)  {1};
        \node at (2.75,-2.65)  {3};
        \node at (3.25,-2.65)  {1};
        
        \draw [decorate,decoration={brace,mirror,amplitude=5pt},xshift=-4pt,yshift=0pt]
(6.75,-1.4) -- (6.75,0.9) node[midway,xshift=1cm] {$p$ times};
    \end{tikzpicture}
    \caption{Illustration of the structure of the algorithm $\A^\rec$ shown for an instance $X$ of size $r=8$ and $k=3$. It chooses a random subset $S$ of the elements of size $r/2=4$ and replaces all elements in $X$ from $S$ with fresh samples from the group and gives the resulting instance $Y$ to the decision oracle $\A^\det$. If $\A^\det$ reports there is a solution, we increment a counter for each of the elements we did not replace. We then repeat this process $p=\polylog(r)$ times and output the $k$ elements whose counters are the highest. In the illustration, we have $p=4$ and have marked each set $S$ chosen. At the end, the three elements with the highest counters are $X_1, X_4, X_7$, so the search algorithm outputs $\{X_1,X_4,X_7\}$ as the solution.}
\end{figure}
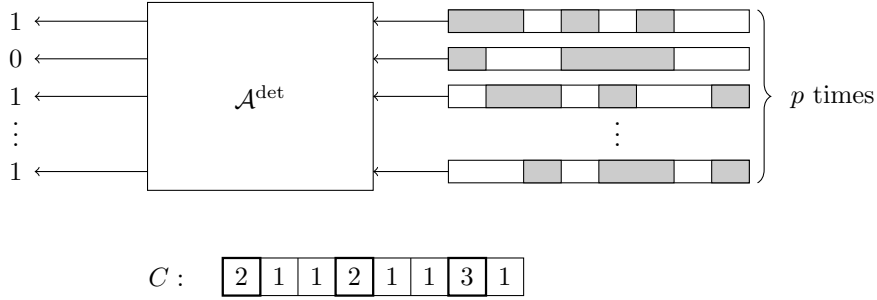

\paragraph{Algorithm $\A^\rec(X)$}
\begin{enumerate}
    \item Let $C \gets 0^r$ be a list of counters.
    \item Repeat $p$ times:
    \begin{enumerate}
      \item[2.1.] Let $Y \gets R(X)$, and let $S$ be the subset chosen in this execution of $R$.
        \item[2.2.] If $\A^\det(Y) = 1$, increment $C_i$ for every $i \not\in S$.
    \end{enumerate}
    \item Output $K\subseteq[r]$ with $|K|=k$ that maximizes $\sum_{i \in K} C_i$.
\end{enumerate}
Sampling an element from $G$ takes time $\mathcal{O}(k \log(r) / \Delta)$, and so $R(X)$ can be computed in time $\Otilde(r / \Delta)$. Also note that the last step can be done in time $\mathcal{O}(r)$ (or faster if using a secondary data structure such as a heap), and hence the total time complexity of the algorithm is $\mathcal{O}((T+r/\Delta)\,p)$. This means we are done if we can show that $\Pr\left[\sum_{i\in K} x_i = 0\right] \geq \gamma$ for some $p = \polylog(r)$.

We now explain our proof strategy at a high level. Intuitively, the above procedure will assign higher counts for the indices belonging to a solution. Indeed, we will give a concentration bound on the value of each counter, conditioned on it being a solution or not. We will then bound the probability of these values belonging to two disjoint intervals, such that the desired solution is output. Finally, we will union bound over all values to achieve the desired result. 

\begin{lemma}\label{lemma:A_rec new solution}
  Let $K \subseteq [r]$ denote the planted set. Then for sufficiently large $r$ it holds that in any given iteration of $\A^\rec$, the probability that there is a set $S \neq K$ such that $\sum_{i \in S}Y_i = 0$ is at most $4 r^{k\,[1-\frac1\Delta]}$.
\end{lemma}
\begin{proof}
  Observe that the $R(X)$ used in a given iteration is effectively sampled from $D_1$ if all the elements in $K$ were preserved, and it is sampled from $D_0$ otherwise. In either case, we can upper bound the probability of a solution distinct from $K$ existing by $\Pr_{Y\sim D_1}[c(Y)>1]$.

  Note that the standard deviation of $c(Y)$ when $Y\sim D_1$ is less than $\sqrt{2\binom{r}{k}/|G|}$ for large enough $r$ (\cref{lemma:c-X stats}, \cref{eq:D-1 variance}). Furthermore, \cref{eq:delta relates to size of G-admissibility} implies $|G|\geq r^{k/\Delta}/2$. Therefore by using \cref{lemma:c-X stats}, \cref{eq:D-1 expectation}, 
  $$
    \underset{Y\sim D_1}{\E}[c(Y)]<1+\frac{\binom{r}{k}}{|G|}\leq 1+\frac{2\binom{r}{k}}{r^{k/\Delta}}<1+r^{k\,(1-\frac{1}{\Delta})}.
  $$

  \noindent We can now apply Chebyshev's inequality (\cref{lemma:chebyshev}) to get,
  \begin{align*}
        \Pr_{Y\sim D_1}[c(Y)>1]
        = &\Pr_{Y\sim D_1}[c(Y)\geq2]\\
        \leq &\Pr_{Y\sim D_1}\left[\left|c(Y)-\exp{c(Y)}\right|>1-r^{k\,(1-\frac{1}{\Delta})}\right]\\
        \leq &\Pr_{Y\sim D_1}\left[\left|c(Y)-\exp{c(Y)}\right|>\frac{1-r^{k\,(1-\frac{1}{\Delta})}}{\sqrt{2\binom{r}{k}/|G|}}\,\Std(c(Y))\right]\\
        \leq &\frac{2\binom{r}{k}}{|G|\left(1-r^{k\,(1-\frac{1}{\Delta})}\right)^2}
        \intertext{Substituting $\binom{r}{k}<(r/k)^k<r^k/2$ and $|G|\geq r^{k/\Delta}/2$, we get,}
        < &\frac{2r^{k\,(1-\frac{1}{\Delta})}}{\left(1-r^{k\,(1-\frac{1}{\Delta})}\right)^2}
        \intertext{The denominator is a monotonically increasing function in $r$ with limit 1. Therefore, for sufficiently large $r$, it will exceed $1/2$, and hence we get,}
        <&4r^{k\,(1-\frac{1}{\Delta})} &&\qedhere
  \end{align*}

\end{proof}

In the following, we will bound the values of the counters. Denote by $\Det(X)$ the correct answer for the instance $X$, i.e.,
$$
    \Det(X) = 
    \begin{cases}
        1 & \text{if $X$ has a solution},\\
        0 & \text{if $X$ does not have a solution}.
    \end{cases}
$$

\begin{lemma}\label{lemma:delta < 1 => Det(X) = b for Db}
    At any constant density $\Delta<1$, if $X\sim D_b$ then $\Det(X) = b$, except with probability $o(1)$.
\end{lemma}
\begin{proof}
  We need to show that with probability $1-o(1)$, an instance $X \sim D_b$ has a solution iff $b=1$. The case of $b=1$ is true by definition, while for $b=0$ we need to upper bound the probability that an instance $X \sim D_0$ has a solution. 
    \begin{align*}
        \underset{X \sim D_0}{\Pr}[\Det(X) \neq 0] &= \underset{X \sim D_0}{\Pr}[c(X) \geq 1] 
        \intertext{Applying Markov's inequality (\cref{lemma:markov}), we get}
        &\leq \exp{c(X)} =\binom{r}{k}/|G| &&\text{\cref{lemma:c-X stats}, \cref{eq:D-0 expectation}}\\
        &< r^k/|G| \leq 2r^{k\,(1-\frac{1}{\Delta})} &&\text{\cref{eq:delta relates to size of G-admissibility}}
    \end{align*}
    which as remarked is subconstant.
\end{proof}

\noindent Say an instance $X$ is \emph{bad} if $\Pr[\A^\det(R(X)) \neq \Det(R(X))] > \frac{1}{2^{2k}}$, with randomness taken over $R$ and $\A^\det$.
\begin{lemma}\label{lemma:most_X_are_correct}
     If $\A^\det$ solves the decision problem with success probability $1-o(1)$, then for sufficiently large $r$, an instance $X \sim D_1$ is bad with probability at most $\frac{1}{2^{8k}}$.
\end{lemma}
\begin{proof}
  At a high level, the result essentially follows using a Markov bound at a sufficiently high value of $r$. Let $X \sim D_1$ be an instance from the planted distribution. First note that the distribution of $R(X)$ is a convex combination of $D_0$ and $D_1$, determined by whether or not the set chosen by $R(\cdot)$ intersects with the solution. Hence by correctness $\A^\det$ has to mostly agree with $\Det$. Let $S$ be the subset chosen by $R$. Clearly, if $S$ is disjoint from the solution, the resulting instance is distributed as $D_1$. Otherwise, we do not preserve the solution and the instance is distributed as $D_0$. Note that as $\Delta<1$ is constant, it follows from \cref{lemma:delta < 1 => Det(X) = b for Db} that except with probability $o(1)$, $X\sim D_{\Det(X)}$. By convexity, in the former case, $\A$ has to output $1$ except with probability $o(1)$, while in the latter case, it has to output $0$ except with probability $o(1)$. Let $Z(X) = \underset{Y \gets R(X)}{\Pr} \left[ \A(Y) \neq \Det(Y)\right]$, and note that an instance $X$ is bad if $Z(X) \geq \frac{1}{2^{2k}}$.  For large enough $r$, we thus get a bound of 
  $\E[Z(X)]  \leq \frac{1}{2^{10k}}$ with the expectation taken over $X$ and $R(\cdot)$. We can now bound the probability that $X$ is bad.
    \begin{align*}
      \underset{X \sim D_1}{\Pr}[\text{$X$ is bad}] = \underset{X \sim D_1}{\Pr}\left[Z(X) \geq \frac{1}{2^{2k}}\right]
        &\leq \underset{X \sim D_1}{\Pr}\left[ Z(X) \geq 2^{8k}\,\E[Z(X)] \right] \leq \frac{1}{2^{8k}},
    \end{align*}
    where the latter follows from Markov's inequality (\cref{lemma:markov}). 
\end{proof}

\begin{proofof}{\cref{thm:search_to_decision}} At a high level, we will give a concentration bound on each counter using a Chernoff bound, and conclude that, with high probability, the range of the counters for the indices in the solution is disjoint from the range of counters outside the solution by employing a union bound on all the counters.

  Fix an input $X$. Now, in the reduction, for each choice of $R(X)$, the counter will be incremented by some vector which is either zero if the solution was destroyed, or a balanced vector if the solution is preserved. Let $E_{ij}$ be the event that the $i^\text{th}$ counter was incremented in the $j^\textrm{th}$ iteration. Note that $C_i = \sum_{j=1}^p E_{ij}$. Now consider an index $i$ belonging to the planted solution, and suppose that $\A^\det$ has no errors and that $R(\cdot)$ did not introduce any new solutions. Then $C_i$ is incremented if all of the indices belonging to the solution were not replaced, and thus $\E[E_{ij}\mid \text{$i$ solution}] = \frac{1}{2^k}$. Analogously, for an index not belonging to a solution, it will be incremented if the solution were preserved and also this index was preserved, and so the error-free expectation will be $\E[E_{ij} \mid \text{$i$ not solution}] = \frac{1}{2^{k+1}}$. By \cref{lemma:most_X_are_correct}, even if $\A^\det$ has a $o(1)$ probability of error, we know that for each counter the error is $\epsilon<\frac{1}{2^{2k}}<\frac{1}{2^{k+3}}$ (since $k\geq 3$) with probability at least $1-\frac{1}{2^{8k}}$, where the randomness is taken over the instance. In addition, even if we destroyed the solution, $R(\cdot)$ might inadvertently create a new solution which happens with probability $< 4r^{k(1-\frac1\Delta)}$.  To account for the errors, we assume, as a worst-case precaution using a union bound, that the expectations change by at most $\epsilon$, such that by linearity of expectation, 
\begin{align*}
    \E[C_i] &\geq
         p\left(\frac1{2^k} - \epsilon\right), && \text{when $i$ is solution.}\\
    \E[C_i] &\leq
         p\left(\frac1{2^{k+1}} + \epsilon + 4r^{k\,(1-\frac1\Delta)}\right), && \text{when $i$ is not solution.}
\end{align*}
We now wish to say that the range of values of indices belonging to the solution is disjoint from the range of those not belonging to the solution. We say a counter is \emph{bad} if it deviates from its expectation by more than $\frac{p}{2^{k+4}}$. This ensures that when no counters are bad, for sufficiently large $r$, the range of counts for the indices belonging to a solution is disjoint from those not belonging to a solution. To see this, we compute the distance $\Delta$ to the midpoint of the expectations, i.e.,
\begin{align*}
    \Delta = \frac{ p\left(\frac1{2^k} - \epsilon\right) - p\left(\frac1{2^{k+1}} + \epsilon + 4r^{k\,\left(1-\frac1\Delta\right)}\right)}{2} &> \frac{p \left(\frac1{2^k} - \frac 1{2^{k+3}} - \frac 1{2^{k+1}} - \frac 1{2^{k+3}} - 4r^{k\,\left(1-\frac1\Delta\right)}\right)}{2}\\ &= p\left(\frac{1}{2^{k+3}} - 2r^{k\,\left(1-\frac1\Delta\right)}\right)\\
    &> \frac{p}{2^{k+4}},
\end{align*}
where the first inequality follows as $\epsilon < \frac{1}{2^{2k}} < \frac{1}{2^{k+3}}$ is true for any $k\geq 3$, and second inequality follows since $\Delta<1$ is constant and $2r^{k\left(1-\frac1\Delta\right)}<\frac{1}{2^{k+4}}$ for sufficiently large $r$. Note that when $X$ is fixed, each $E_{ij}$ and $E_{ik}$ are independent for $j\neq k$ and are supported on $\{0,1\}$. We may thus we may bound the probability of a bad counter as function of $p$ using a Chernoff bound (\cref{lemma:chernoff}). Suppose that $i$ is an index belonging to the solution, then we get the following bound, (\cref{lemma:chernoff}).
\begin{align*}
    \Pr\left[\text{$i^\mathrm{th}$ counter is bad} \mid \text{$i$ solution}\right] 
    &= \Pr\left[C_i < \E [C_i] - \frac{p}{2^{k+4}}\right]
    = \Pr\left[\frac1p \,C_i < \frac1p \,\E [C_i] - \frac{1}{2^{k+4}}\right]
    < \mathrm{e}^{-\frac{p}{2^{2k+7}}}.
\end{align*}
Now let $0<\gamma \leq 1$ be any constant and let $p = 2^{2k+7}\,\ln \frac{r}{\gamma}$. Then we get an upper bound of $\frac{\gamma}{r}$ for a solution counter going bad. We get the same bound for the indices not belonging to a solution. By a union bound on all the counters, we bound the total error rate by $\gamma$. 
\end{proofof}
\subsection{Reduction from $k$-SUM to Subset Sum at Very Low Densities}\label{sec:subset-sum-reduction}

In this section, we reduce the planted $k$-SUM problem on integers to the subset sum problem. We will show both a worst-case as well as an average-case reduction. We then use existing algorithms for low-density subset sum to get non-trivial algorithms for planted $k$-SUM at low densities. Surprisingly, the two constructions are quite different and can not be combined.


\begin{definition}[Worst Case Algorithm for Subset Sum]
  A subset sum problem instance comprises a vector $\mathbf{A}$ containing $r$ integers, and a target value $t$. An worst-case algorithm $\mathrm{Alg}$ solves the problem in time $T(r)$ if and only if it can find a Boolean vector $\mathbf{x}$ of length $r$ such that $\mathbf{A\cdot x} = t$ whenever such a vector $\mathbf{x}$ exists.
\end{definition}
\noindent Note that the above problem is known to be \textsf{NP}-complete. \nikolaj{Elaborate on why this matters.}\shweta{I don't think we need to.}

\begin{definition}[Average-Case Algorithm for Subset Sum]
    The average case problem is parametrized by two integers $r$ and $N$, where $N$ must be a prime power. To sample an instance, we choose a vector $\mathbf{A}$ from $\mathbb{Z}_N^r$ and a vector $\mathbf{x}$ from $\{0, 1\}^r$ uniformly at random. An average case algorithm returns a vector $\mathbf{x'}$ given $\mathbf{A}$ and $\mathbf{A\cdot x} \mod N$ such that $\mathbf{A\cdot x} = \mathbf{A\cdot x'} \mod N$ with constant probability. Note that this probability is taken over the randomness of the input as well as the randomness used inside the algorithm.
\end{definition}
\begin{remark}
  \label{rem:subset-sum}
  The quantity $\Delta=\frac{r}{\log{N}}$ is called the density of the subset sum instance. The runtime of an average case algorithm is expressed as a function of $r$ for some fixed $\Delta$. It is known that the problem is hardest when $\Delta=1$, and there exist polynomial time algorithms when $\Delta\leq \frac{1}{r}$~\cite{LO85,Bennett22}. 
\end{remark}

\begin{remark}
  \label{rem:subset-sum-2}
  In fact, the specific algorithms described in~\cite{LO85,Bennett22} can be applied directly to solve $k$-SUM at the densities in \cref{cor:k-MSUM low density alg,cor:k-SUM low density alg} and would give better success probability than that stated there. Nevertheless, we present our results as corollaries of our reduction to subset sum, as this reduction works for a wider range of densities, and would transfer improvements in algorithms for subset sum immediately to $k$-SUM.
\end{remark}

It is possible to define the average-case version of the problem without using modular arithmetic such that the two versions correspond to each other more obviously. The modular version of the problem can be solved by using an oracle for non-modular subset sum by calling it $r$ times with multiples of $N$ added to the target. On the other hand, the non-modular version can be solved by simply calling the oracle for modular subset sum once; there exists an unique solution with high probability in low densities. We chose this definition because it is cleaner and still equivalent to the more intuitive translation of the subset sum problem to the average-case setting.

\subsubsection{Worst-Case Reduction to Subset Sum}

In this section, we show the following worst-case reduction from the $k$-SUM search problem to subset sum.

\begin{theorem}[Worst-Case Reduction to Subset Sum]\label{thm:worst case reduction to subset sum}
    Suppose there exists a worst-case algorithm $A$ with time complexity $T(r) = \Omega(r)$ that solves subset sum. Then, there is an algorithm $B$ with the same time complexity $T(r)$ that solves the worst-case search $k$-SUM problem.
\end{theorem}

\begin{proof}
    We start by describing the new algorithm:

    \paragraph{Algorithm $B(X)$}
    \begin{enumerate}
        \item Construct a new array $Y$ of $r$ integers.
        \item Let $M := \max(|X_i|)+1$. 
        \item Set $Y[i] := (k+1)M+X[i]$ for all $1\leq i\leq r$.
        \item Run $A$ on the set $Y$ with target $t=k\,(k+1)M$.
        \item Return the set of indices returned by the previous call.
    \end{enumerate}

    The above algorithm clearly has the same runtime complexity as $A$ does. To show correctness, recall that the input array $X$ must have a $k$-SUM solution. So there is a set $\kappa$ of size $k$ such that $\sum_{i\in\kappa}X_i = 0$. By our construction, $\sum_{i\in\kappa}Y_i = k\,(k+1)M+\sum_{i\in\kappa}X_i=t$. Therefore, we ensure that there is at least one solution to the subset sum problem instance we create. Hence, $A$ must return a valid subset sum solution.

    Suppose that $A$ returns a set $S$. By construction, $$|X_i|<M\Rightarrow M+X_i>0 \Rightarrow Y_i = (k+1)M+X_i=kM+(M+X_i)> kM$$ Therefore, if $|S|\geq k+1$, the value of $\sum_{i\in S}Y_i > (k+1)\cdot kM = t$. This is a contradiction since the aforementioned sum is supposed to equal $t$, and we can thus conclude that $|S|\leq k$. Similarly, observe that $$|X_i|<M\Rightarrow X_i<M \Rightarrow Y_i = (k+1)M+X_i< (k+2)M$$ Therefore, if $|S|\leq k-1$, the value of $\sum_{i\in S}Y_i < (k-1)\cdot (k+2)M <k\,(k+1)M= t$. This is also a contradiction since the aforementioned sum is supposed to equal $t$, and we can thus conclude that $|S|\geq k$. Together, these two constraints imply that $|S|=k$. Now we can easily calculate $$\sum_{i\in S}X_i = \sum_{i\in S}\left(Y_i - (k+1)M\right) = \sum_{i\in S}Y_i - k\,(k+1)M = 0$$ This concludes our proof that $B$ returns a set of exactly $k$ indices which correspond to a $k$-SUM solution.

    Note that the above reduction works unchanged in the case where $A$ is a randomized algorithm with some constant success probability; $B$ will also then have the same success probability in that case.
\end{proof}

\subsubsection{Average-Case Reduction to Subset Sum}
The above reduction unfortunately does not generalize to the average case setting. An average-case oracle could potentially be biased against any $Y$ of the above form (which can be constructed from an array $X$ with a $k$-SUM solution). To get around this, we reduce from the modular $k$-SUM problem with prime moduli. In the average case, this is known to be equivalent to the integer $k$-SUM problem.

\begin{theorem}[Average-Case Reduction to Subset Sum]\label{theorem:subset-sum AC reduction}
  Suppose there exists an algorithm $A$ of time complexity $T(r)=\Omega(r)$ that solves the average-case subset sum problem with constant success probability at some constant density $\Delta<1$. Then, for any constant $k$, there is an algorithm $B$ with the same time complexity that can solve the planted search $k$-SUM problem on groups of the form $\mathbb{Z}_p$ (where $p$ is a prime larger than $k$) at density $\frac{k\Delta\log r}{r}$ with constant success probability.
\end{theorem}
\begin{proof}
  We start by describing the new algorithm:

    \paragraph{Algorithm $B(X)$}
    \begin{enumerate}
        \item Construct a new array $Y$ of $r$ integers.
        \item Choose $\alpha$ uniformly at random from $\mathbb{Z}_p$.
        \item Set $Y[i] := \alpha+X[i]$ (mod $p$) for all $1\leq i\leq r$.
        \item Choose a random $S\subseteq[r]$.
        \item Run $A$ on the set $Y$ with target $t=k\alpha+\sum_{i\in S}Y_i$ (mod $p$).
        \item If $A$ succeeds and returns $S'$ such that $S\subset S'$ and $|S'\setminus S|=k$, return $S'-S$.
    \end{enumerate}
    Since both the problems here involve members of $\mathbb{Z}_p$, we implicitly treat all numbers in the following analysis modulo $p$.

    We will first prove that $Y$ \emph{looks random}. Specifically, we will demonstrate that 
    $$\Pr\left[Y[i] = x \,\,\Bigg|\, \bigcup_{j\neq i}Y_j \right] = \frac{1}{p} \hspace{1 in} \forall i\in[r] \,\forall x\in\mathbb{Z}_p$$
    Note that $X$ comes from a planted $k$-SUM instance, and therefore, was sampled from the planted distribution. Let us consider the effect of combining that sampling process with the first 3 steps of our algorithm $B$, which is how we obtain the array $Y$.
    
    \paragraph{Sampling Algorithm for $Y$}
    \begin{enumerate}
        \item Sample $r$ elements $Y_1, Y_2, \ldots, Y_r$ i.i.d. uniformly at random from $\mathbb{Z}_p$.
        \item Choose a random set $T \subseteq [r]$ with $|T|=k$
        \item Let $t \in T$ be the smallest index and let $Y_t \gets - \sum_{\underset{j \neq t}{j \in T}} Y_j$
        \item Choose $\alpha$ uniformly at random from $\mathbb{Z}_p$.
        \item Increment $Y_i$ by $\alpha$ for all $1\leq i\leq r$.
    \end{enumerate}

    Steps 2 and 4 are just independent uniformly random choices; we can obviously move both steps to the very beginning. Observe that we can move step 3 to the end if we simply modify the assignment to $Y_t \gets k\alpha - \sum_{\underset{j \neq t}{j \in T}} Y_j$. This is because step 3 does not affect the other $r-1$ indices of $Y$ and the modified assignment represents the net effect of steps 3 and 5 on $Y_i$. Therefore the sampling algorithm is equivalent to the following:

    \paragraph{Equivalent Sampling Algorithm for $Y$}
    \begin{enumerate}
        \item Sample $r$ elements $Y_1, Y_2, \ldots, Y_r$ i.i.d. uniformly at random from $\mathbb{Z}_p$.
        \item Choose a random set $T \subseteq [r]$ with $|T|=k$. Let $t\in T$ be its smallest index.
        \item Choose $\alpha$ uniformly at random from $\mathbb{Z}_p$.
        \item Increment $Y_i$ by $\alpha$ for all $1\leq i\leq r$.
        \item Let $Y_t\gets k\alpha-\sum_{\underset{j \neq t}{j \in T}} Y_j$
    \end{enumerate}

    In the new algorithm, it is easy to see that $Y$ is uniformly distributed in $\mathbb{Z}_p^r$ at the end of step 3. That remains true after step 4, since $\mathbb{Z}_p$ is a cyclic group and adding $\alpha$ is merely applying an invertible translation. Finally, note that, for any setting of the $Y_i$'s, $k\alpha-\sum_{\underset{j \neq t}{j \in T}} Y_j$ is a uniformly random member of $\mathbb{Z}_p$, since so is $\alpha$ and thus $k\alpha$ as $k$ is co-prime to $p$. Thus, the last step simply replaces $Y_t$ with a freshly chosen uniformly random element of $\mathbb{Z}_p$. Therefore, the result of this entire sampling procedure $Y$ is distributed uniformly in $\mathbb{Z}_p^r$. Further, due to symmetry, the set $T$ is also independent of $Y$ and distributed uniformly over all subsets of $[r]$ of size $k$.

    \pnote{Commented out another sampling algorithm that doesn't seem to be needed.}



    Observe that since the density of the provided $k$-SUM is $\frac{k\Delta\log r}{r}$, we have $\frac{k\log r}{\log p} = \frac{k\Delta\log r}{r}$, and so $\frac{r}{\log p} = \Delta$. Recall that in the planted $k$-SUM problem, the existence of a $k$-tuple $T\subset[r]$ such that $\sum_{i\in T}X_i = 0$ is guaranteed. Therefore, $\sum_{i\in T}Y_i=k\alpha$. With probability $2^{-k}$, the set $S$ is disjoint from $T$. In that event, $S'=S\bigcup T$ is a valid solution to the subset sum instance in line 5 (of algorithm $B$), and it also satisfies the conditions in line 6. We can now upper bound the probability of there being a different subset of $Y$ having the same sum as follows:
    \begin{align*}
    &\Pr\left[\exists \veca\in\bitset^r \text{ such that } \langle Y,\veca\rangle = \langle Y,\vecb\rangle \text{ and } \veca\neq\vecb\right], &&\text{$\vecb$ is the indicator bit vector for $S'$}\\
        &\quad\leq \sum_{\substack{\veca\in\bitset^r\\ \veca\neq \vecb}}\pr{\langle Y,\veca\rangle = \langle Y,\vecb\rangle}, &&\text{by the union bound}\\
        &\quad=\sum_{\substack{\veca\in\bitset^r\\ \veca\neq \vecb}}\frac{1}{p} &&\text{Since $Y\sample\Z_p^r$}\\
        &\quad< \frac{2^{r}}{p} = 2^{r-\log{p}} \\
        &\quad=2^{r-\frac{r}{\Delta}} = 2^{r(1-\frac{1}{\Delta})} \\
        &\quad= \negl(r) &&\text{Since $\Delta$ is a constant $<1$}
    \end{align*}
    
    \noindent So if $A$ succeeds on the problem instance $(Y, t)$ constructed in line 5 with some constant probability $q$, we can conclude that $B$ successfully solves the $k$-SUM instance with probability at least $\frac{q}{2^{k+1}}$.
    \begin{align*}
        &\pr{B \text{ succeeds}}\geq \pr{B\text{ succeeds}\mid S\cap T=\phi}\cdot\pr{S\cap T = \phi}\\
        =&\frac{1}{2^k}\pr{B\text{ succeeds}\mid S\cap T=\phi} \\
        =&\frac{1}{2^k}\pr{A\text{ succeeds and } A(Y,t) \text{ passes the checks in line 6}\mid S\cap T=\phi} 
        \intertext{If we now condition on $Y$ not having a second solution, $A$ succeeding would imply it finds the set $S\cup T$ which passes the checks in line 6}
        \geq&\frac{1}{2^k}(1-\negl(r))\pr{A\text{ succeeds}\mid S\cap T=\phi\text{ and } (Y,t) \text{ has a unique subset sum solution}}
        \intertext{Note that for any 3 events $P,Q,R$, we have $\pr{P\mid Q\And R}\pr{R}+\pr{P\mid Q\And \texttt{not} R}\pr{\texttt{not}R}=\pr{P\mid Q}$ which implies $\pr{P\mid Q\And R}\geq\pr{P\mid Q}-\pr{\texttt{not}R}$}
        \geq&\frac{1-\negl(r)}{2^k}\paren{\pr{A\text{ succeeds}\mid S\cap T=\phi} - \pr{Y\text{ has a second subset sum solution}}} \\
        =&\frac{1-\negl(r)}{2^k}\paren{q - \negl(r)}
        \geq\frac{q}{2^{k+1}} \hspace{25pt}\text{for large }r
    \end{align*}
    where $q$ is the probability that $A$ succeeds conditioned on $S\cap T = \phi$.
 
    A subset sum problem instance is characterized by a vector and a subset. We have already concluded that the vector input in our constructed problem instance is distributed uniformly in $\mathbb{Z}_p^r$. The density of our constructed subset sum problem is exactly $\Delta$, and so $A$ would succeed with constant probability if the subset was chosen uniformly. Note that the subset ($S'$) is the union of $S$ and $T$ where $S$ is chosen uniformly at random in line 4, $T$ is a random subset of size $k$, and we only condition on $S\bigcap T=\emptyset$. Since $T$ was independent from $Y$, so is $S'$ even after the above conditioning. Therefore, since $k$ is a constant, use \cref{lemma:subset_sum set distribution} to conclude that the probability that $A$ solves the subset sum problem we construct conditioned on $S\cap T = \phi$ (that is, $q$) is at least a constant. Since $k$ is a constant, this coupled with the inequality above implies that $B$ has constant success probability.
\end{proof}
\begin{corollary}[$k$-SUM in $\bbZ_p$ is Easy at Very Low Density]\label{cor:k-MSUM low density alg}
  There exists an average-case algorithm for $k$-SUM over $\mathbb{Z}_p$ groups at density at most $\frac{k\log r}{r^2}$ whose runtime complexity is polynomial in $r$ and does not depend on $k$.
\end{corollary}
\begin{proof}
  The subset sum problem can be solved in the average-case in polynomial time for density at most $1/r$ with constant success probability (see \cref{rem:subset-sum}). The statement then follows from \cref{theorem:subset-sum AC reduction}.
\end{proof}

\noindent We can use a reduction by \cite{dinur_keller_klein} to lift this result to $k$-SUM over the integers.
\begin{corollary}[$k$-SUM Over Integers is Easy at Very Low Density]
\label{cor:k-SUM low density alg}
    There exists an average-case algorithm for $k$-SUM over integers (where the integers are chosen at random from $[-N,N]$ such that $\Delta=\frac{k\log r}{\log N}$) at density at most $\frac{k\log r}{r^2}$ whose runtime complexity is polynomial in $r$ and does not depend on $k$.
\end{corollary}
\begin{proof}
    Let $p$ be a prime between $N$ and $2N$. \cref{cor:k-MSUM low density alg} gives us a polynomial algorithm $A$ that solves $k$-SUM in $\mathbb{Z}_p$ for density approximately $\Delta$. We can now use Theorem 4.5 in \cite{dinur_keller_klein} to obtain the desired algorithm.
\end{proof}

\begin{lemma}\label{lemma:subset_sum set distribution}
    Assume that $k$ is some fixed constant. Let $D_0$ be the uniform distribution on subsets of $[r]$. Let $D_1$ be the distribution induced by the following sampling procedure:
    \begin{enumerate}
        \item Choose $k$ distinct integers uniformly at random from $[r]$; call this set $T$.
        \item Choose a random subset $S$ of $[r]$
        \item If $S\bigcap T=\emptyset$, return $S\bigcup T$. Otherwise, repeat the process.
    \end{enumerate}
    If an average-case algorithm $A$ has success probabilities $\gamma_0=\Omega_k(1)$ on inputs from $D_0$ and $\gamma_1$ on inputs from $D_1$, then $\gamma_1\geq \gamma_0 \,2^{-(2k+1)}$.
\end{lemma}
\begin{proof}
    Let $D_2$ be the uniform distribution on subsets of $[r]$ with size at least $\frac{r}{4}$, and let $\gamma_2$ be the success probability of $A$ on inputs from $D_2$. Our proof will bound the total variation distance between $D_0$ and $D_2$, and the Rényi divergence of order $\infty$ between $D_1$ and $D_2$. We will thus derive bounds on $|\gamma_0-\gamma_2|$ as well as $\frac{\gamma_1}{\gamma_2}$, which will together imply the given statement.

    Note that $\supp(D_2)\subset\supp(D_0)$. Since both of these distributions are uniform, we have,
    \begin{align*}
        \Delta&(D_0, D_2) \\
        &= \sum_{S \in \supp(D_2)} \left(D_2(S)-D_0(S)\right)\\
        &= \sum_{S \in \supp(D_0)-\supp(D_2)} D_0(S)\\
        &= \sum_{S\subset [r], |S|<\frac{r}{4}} D_0(S)\\
        &= \Pr_{S\sim D_0}\left[|S|<\frac{r}{4}\right]\\
        &\leq \mathrm{exp}\left(-2r\left(\frac{1}{2}-\frac{1}{4}\right)^2\right) \intertext{Using Hoeffding's inequality (\cref{lemma:hoeffding}), we get,}\\
        &= \mathrm{exp}\left(-\frac{r}{8}\right)\\
        &= o(1).
    \end{align*}
    This implies that $\gamma_2\geq \gamma_0-o(1)>\frac{\gamma_0}{2}$. Let us now calculate an explicit expression for $D_1(Q)$ for a subset $Q$ of size at least $k$. To get an output of $Q$, we need two events to happen simultaneously. $T$ must be a subset of $Q$, which happens with probability $\binom{|Q|}{k}/\binom{r}{k}$. Given any such choice of $T$, $S$ must be exactly $Q-T$. There are $2^{r-k}$ choices of $S$ which do not intersect with $T$, so the probability of this particular set being chosen is $2^{k-r}$ (note that the sampling process for $D_1$ implicitly conditions on $S\cap T=\emptyset$). Since the above two events are independent, we have $$D_1(Q) = \frac{\binom{|Q|}{k}}{\binom{r}{k}}\cdot\frac{2^k}{2^r}$$
    Note that all sets of size at least $k$ are in the support of $D_1$. Therefore, as long as $r>4k$, we have $\supp(D_2)\subseteq\supp(D_1)$. We can now calculate the following Rényi divergence using \cref{lemma:renyi} as follows.
    \begin{align*}
        R&(D_2 \lVert D_1)\\
        &= \underset{Q \in \supp(D_2)}{\max} \frac{D_2(Q)}{D_1(Q)}\\
        &= \underset{Q\subseteq[r], |Q|\geq\frac{r}{4}}{\max} \left(D_2(Q)\cdot\frac{\binom{r}{k}}{\binom{|Q|}{k}}\cdot2^{r-k}\right)\\
        &= \underset{Q\subseteq[r], |Q|\geq\frac{r}{4}}{\max} \left(\frac{1}{\left\lvert S \subseteq [r] \mid |S| \geq r/4\right\rvert}\cdot\frac{\binom{r}{k}}{\binom{|Q|}{k}}\cdot2^{r-k}\right) \intertext{We may bound this as follows.}\\
        &\leq \left(\frac{1}{2^{r-1}}\cdot\frac{\binom{r}{k}}{\binom{r/4}{k}}\cdot2^{r-k}\right).\intertext{This can be written as follows.}\\
        &=\frac{1}{2^{k-1}}\cdot\frac{r\,(r-1)\cdots(r-k+1)}{(r/4)(r/4-1)\cdots(r/4-k+1)}\\
        &<\frac{2}{2^k}\cdot\left(\frac{r-k+1}{r/4-k+1}\right)^k\\
        &<\frac{1}{2^k}\cdot\left(\frac{2r-8k+8}{r/4-k+1}\right)^k =4^k.
    \end{align*}
    
    \noindent Where the last inequality holds for large enough values of $r$. Applying the Rényi divergence lemma (\cref{lemma:renyi}) on the event that algorithm $A$ succeeds, we get that, $$\gamma_1\geq \gamma_2/R(D_2\lVert D_1)>\gamma_2/4^k.$$

    \noindent Combining the two above relations, we get $\gamma_1\geq 2^{-(2k+1)}\gamma_0$.
\end{proof}

\subsection{Algorithm for $k$-XOR at Low Densities}
\label{sec:algo}

In this section, we will show that the $k$-XOR problem becomes easy to solve at densities $\Delta \leq \frac{1}{r^{0.5+\epsilon}}$. To start with, observe that the $k$-XOR problem becomes very easy to solve when the input matrix is square (i.e. $m=r$). This is because a random boolean square matrix is full rank with constant probability, in which case the planted $k$-XOR solution is the only linear dependence, and we can find it using Gaussian elimination in $\O(r^3)$ time. This algorithm also works when $m>r$; we can simply ignore the bottom $m-r$ rows. Below, we will provide an algorithm for $k$-XOR that works by repeatedly trying to reduce a more general instance to this case.

\begin{theorem}[Algorithm for $k$-XOR]\label{thm:faster-algo-kxor}
    For any $k\in\bbN$, there is an algorithm that runs in time $\Otilde\left(r^{k}m^{3-k}\right)$ and solves the planted search $k$-XOR problem at any density with success probability $1-o(1)$.
\end{theorem}

\begin{proof}
Consider the following algorithm that is given input $A\in\bbF_2^{m\times r}$. It becomes trivial when $r < m/2$, but in that case the Gaussian elimination part can be applied directly to the instance to get the same results.
\paragraph{Algorithm $\cP(A)$}
\begin{enumerate}
\item Repeat the following steps $\left(\dfrac{4r}{m}\right)^k\log{r}$ times:
  \begin{enumerate}
      \item[1.1.] Randomly choose $\frac{m}{2}$ columns of $A$ to create a new $m\times \frac{m}{2}$ matrix $B$.
      \item[1.2.] Run Gaussian elimination on $B$ to find any linear relationships between its columns.
      \item[1.3.] If the previous step returns a linear dependence of size exactly $k$, return it.
  \end{enumerate}
\item Return $\bot$.
\end{enumerate}

\noindent Observe that the planted solution is preserved in $B$ in any given repetition with probability 
$$
    \dfrac{\binom{r-k}{m/2-k}}{\binom{r}{m/2}}=\dfrac{(r-k)!\,(m/2)!}{(m/2-k)!\,r!}=\dfrac{(m/2)(m/2-1)\cdots(m/2-k+1)}{r\,(r-1)\cdots(r-k+1)}>\left(\frac{m/2-k+1}{r-k+1}\right)^k>\left(\frac{m}{4r}\right)^k
$$
\noindent Since we run $(\frac{4r}{m})^k\log{r}$ repetitions, the probability that $B$ contains the original $k$-XOR solution at least once is greater than $$1-\left(1-\left(\frac{m}{4r}\right)^k\right)^{\left(\frac{4r}{m}\right)^k\log{r}}>1-\frac{1}{r}$$

If we condition on $B$ having the $k$-XOR solution preserved, it is easy to see that the induced distribution on $B$ is the planted distribution on $\mathbb{F}_{2}^{m\times\frac{m}{2}}$. Since the number of rows is much larger than the number of columns, the planted solution is the only linear dependence in such a matrix with probability $1-\negl(m)$. Therefore, Gaussian elimination returns a $k$-XOR solution and we return the correct answer 1 with probability greater than $1-o(1)$.

The runtime of Gaussian elimination is $\O(m^3)$. The total runtime of the above algorithm is therefore $\O\left(\left(\frac{r}{m}\right)^{k}m^3\log{r}\right)=\Otilde(r^{k}\,m^{3-k})$. 
\end{proof}
\begin{corollary}\label{corollary:k-xor algorithm at low density}
    For any $k\in\bbN$ and constant $\epsilon>\dfrac{3}{k-3}$, there is an algorithm for the planted search $k$-XOR problem that runs in time $r^{\ceil{k/2}-\Omega(1)}$.
\end{corollary}
\begin{proof}
    Since $m=\dfrac{k\log{r}}{\Delta}$, the runtime of $\cP$ is $\Otilde\left(\dfrac{r^k\,\Delta^{k-3}}{k^{k-3}}\right)=\Otilde_k(r^{(k-3)/2})=o(r^{\ceil{k/2}-\Omega(1)})$. 
\end{proof}
\section{Hardness Amplification}
\label{sec:success amplification}
In this section, we propose a success amplification procedure for search $\kSUM$ outside the dense regime. The procedure amplifies the success probability of any algorithm that solves planted search $k$-SUM with probability $\Omega(1/\polylog(r))$ to $1-o(1/\log{r})$, at the cost of increasing its runtime by a polylogarithmic factor. In \cref{sec:hardness amplification for general group}, we prove this result for any abelian group up to density $(1-\eps)$ for some $\eps = o(1)$. In \cref{sec: success amp upto 1 for vksum,sec: success amp upto 1 for modular ksum}, we show how to extend our result up to density 1 for the special case of $\kSUM$ over integers modulo a power of $2$, and \vkSUM (of which $\kXOR$ is a special case). Note that while $\textsf{G}$ is the group ensemble, and $G$ is the group specific to the instance size $r$, for simplicity of notation, we will use $\textsf{G}$ to denote also the specific group and trust that this will not cause confusion.

\subsection{Hardness Amplification for General Groups} \label{sec:hardness amplification for general group}

We will consider the $\kSUM$ problem over some arbitrary abelian group ensemble $\G = \set{G^{(r)}}$. Our input will be an array of length $r$ containing elements of $G^{(r)}$. At density $\Delta$, we have the relation $\Delta \log{|G^{(r)}|}=k\log{r}$. As long as $1> \Delta=\Omega\left(\frac{1}{\polylog(r)}\right)$ and $k = O(1)$, this implies that $\log{|G^{(r)}|} = \Theta(\polylog(r))$. Below, for any instance $M$ and set $S\subseteq [r]$, we will denote by $M[S]$ the sub-array of $M$ indexed by the elements in $S$.

Recall the distributions $D_0$ and $D_1$ defined in \cref{sec:planted_k_sum} for any $r$ (specified by context) as follows:
\begin{enumerate} 
  \item To sample from $D_0$, we simply choose each element of the array independently and uniformly at random from $G^{(r)}$. 
  \item To sample from $D_1$, we do the same thing and then replace a random entry with the negated sum of $k-1$ randomly chosen other entries\footnote{This is equivalent to replacing the smallest entry with the negated sum of the other entries, as defined in \cref{sec:planted_k_sum}.}.
\end{enumerate}
Given an arbitrary instance $M$ sampled from $D_1$, the planted search $\kSUM$ problem asks for a $k$-tuple of indices $S$ such that: 
\begin{equation}\label{eq:k-XOR solution}
  \sum_{i\in S}M_i=0
\end{equation}

\begin{theorem}[Hardness Amplification for General Groups]\label{thm: recovery algorithm success amplification}
  Consider any $k\geq 3$, $\alpha \geq 1$, and a group ensemble $\G$ with density $\Delta_k\,(\G) = \Delta$ such that:
  $$\Omega\pfrac{1}{\polylog(r)}\leq\Delta\leq \frac{k\log{r}}{k\log{r}+(\alpha+1)\log\log{r}}<1.$$
  Suppose there exists an algorithm $\A^\rec$ that runs in time $T(r) = \Omega(r)$, and solves the planted search $\kSUM$ problem over $\G$ with success probability $\gamma = \Omega\pfrac{1}{\log^\alpha{r}}$. Then, there is an algorithm that runs in time $\Otilde\left(T\cdot (k/\gamma^2)^k\right)$, and solves the planted search $\kSUM$ problem over $\G$ with success probability $\left(1-o\left(\frac{1}{\log{r}}\right)\right)$.
\end{theorem}

\noindent The proof of this statement is postponed to first build an appropriate framework. At a high level, we will simply invoke the recovery oracle multiple times on inputs related to the original input. To deal with the oracle potentially being malicious, we employ an obfuscation procedure to hide the original solution. We will raise the success probability in several steps. For the sake of simplicity, we will assume that $\gamma=\Theta\left(\frac{1}{\log^\alpha{r}}\right)$ in the rest of the proof. Of course, if our starting algorithm $\A^\rec$ has a higher success probability, we can always make it less reliable by randomly failing on an appropriate fraction of inputs.

Fix a group ensemble $\G$ satisfying the conditions in the theorem. Fix some $r$, and denote the corresponding group $G^{(r)} \in \G$ simply by $G$. For our construction, we will restrict our attention to arrays which have exactly one $\kSUM$ solution.

\begin{definition}[Permissible Array]\label{def:permissible}
  We call an array that having exactly one $\ksum$ solution \emph{permissible}. Formally, an array $M\in\G^r$ is permissible if and only if there exists a unique subset $\kappa\subset[r]$ such that $|\kappa|=k$ and $\sum_{i\in\kappa}M[i]=0$. Define $D_{perm}$ to be the distribution $D_1$ conditioned on the sampled array being permissible.
\end{definition}
\begin{remark}
  Since it is nontrivial to check if a given array is permissible, there is no obvious efficient sampling process for $D_{perm}$. We only use this distribution to make theoretical arguments.
\end{remark}
\begin{lemma}\label{lemma:permissible inputs}
  The following is true.
  \begin{enumerate}
    \item $\displaystyle\Pr_{M\sim D_1}[\text{$M$ is permissible}]=1-o(\gamma)$.
    \item $D_{perm}$ is uniform on its support.
  \end{enumerate}
\end{lemma}
\begin{proof}
  Since $M$ is sampled from $D_1$, there is always the solution we explicitly planted -- call this set $S$. The probability of another $\kSUM$ solution among the $r-1$ original entries is at most ${\binom{r-1}{k}}/{|\G|}<r^{k\,(1-1/\Delta)}$. The new entry creates an extraneous solution exactly when there are some other $k-1$ entries which have the same sum as the unchanged entries in $S$; this happens with probability at most $\frac{\binom{r-1}{k-1}}{|\G|}<r^{k\,(1-1/\Delta)-1}<\frac{1}{r}=o(\gamma)$. Therefore, we conclude that $M$ is permissible except with probability $\O(r^{k\,(1-1/\Delta)})$. The following calculation shows that our upper bound on $\Delta$ implies that the probability of multiple solutions is $o(\gamma)$.
  \begin{align*}
    r^{k\,(1-\frac{1}{\Delta})}
    \leq r^{k\,\left(1-\frac{k\log{r}+(\alpha+1)\log{\log{r}}}{k\log{r}}\right)}=2^{k\log{r}\left(\frac{-(\alpha+1)\log{\log{r}}}{k\log{r}}\right)}
    =\left(\frac{1}{\log{r}}\right)^{\alpha+1} = o(\gamma).
  \end{align*}
  
  \noindent The second statement follows from \cref{lemma:explicit_pmf} and the fact that each permissible array has exactly one $\kSUM$ solution.
\end{proof}

\subsubsection{Solution Obfuscation}\label{sec: Building A_obf}

Our first step is creating an intermediate algorithm $\A^\obf$ that runs in $\Otilde(T)$ time and solves the same problem somewhat more reliably. Let $\A^\rec$ be the recovery algorithm with success probability $\gamma$ promised in the statement of \cref{thm: recovery algorithm success amplification}; it takes as input an array sampled from $D_1$ and returns an array of size $k$ which contains the indices corresponding to a solution. Observe that $\A^\rec$ can be adversarially biased in at least two ways. It can only look for solutions in some particular positions (e.g. it only looks at the first $r/2$ entries and always fails in the $1-1/2^k$ fraction of the inputs where the solution involves the other entries). To deal with this issue, we \emph{anonymize} the solution by applying an arbitrary permutation to $M$.

The original oracle can also only look for solutions where the entries have some particular property (e.g. it can only find solutions where all the elements in $S$ belong to some subgroup of $G$ and always fails for the rest of the input space). To fix this issue, we \emph{randomize} the solution by adding to it $k$ random elements from $G$ that sum to 0. This transforms the solution into a random set of $k$ elements adding up to 0. We cannot do this directly because we do not know where the solutions are. Instead, we pick a set of $k$ random elements that sum to $0$, and add a random element from this set to each element in $M$. With some small constant probability, we will then end up with the desired result of a different one of these added to each element of the solution.

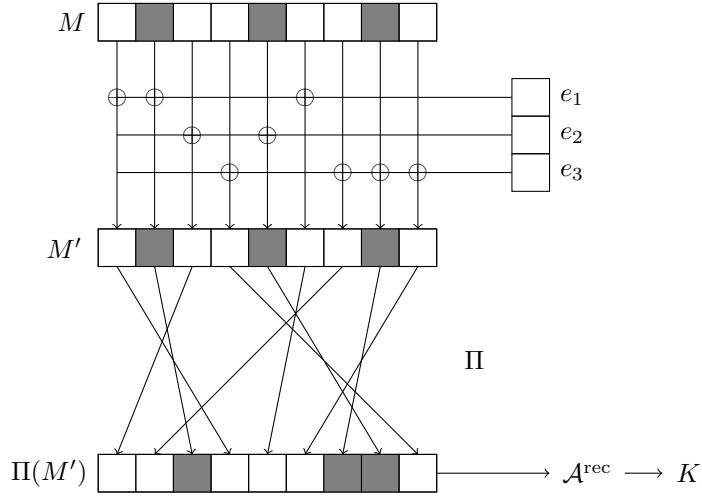
\begin{figure}
    \centering
    \begin{tikzpicture}
        \node at (-0.35,6.25) {$M$};
        \draw (0,6) rectangle (4.5,6.5);
        \draw (0,6) rectangle (0.5,6.5);
        \draw[fill=gray] (0.5,6) rectangle (1,6.5);
        \draw (1,6) rectangle (1.5,6.5);
        \draw (1.5,6) rectangle (2,6.5);
        \draw[fill=gray] (2,6) rectangle (2.5,6.5);
        \draw (2.5,6) rectangle (3,6.5);
        \draw (3,6) rectangle (3.5,6.5);
        \draw[fill=gray] (3.5,6) rectangle (4,6.5);
        \draw (4,6) rectangle (4.5,6.5);

        \draw[->] (0.25,6) -- (0.25,3.5);
        \draw[->] (0.75,6) -- (0.75,3.5);
        \draw[->] (1.25,6) -- (1.25,3.5);
        \draw[->] (1.75,6) -- (1.75,3.5);
        \draw[->] (2.25,6) -- (2.25,3.5);
        \draw[->] (2.75,6) -- (2.75,3.5);
        \draw[->] (3.25,6) -- (3.25,3.5);
        \draw[->] (3.75,6) -- (3.75,3.5);
        \draw[->] (4.25,6) -- (4.25,3.5);

        \node at (6.3,5.25) {$e_1$};
        \node at (6.3,4.75) {$e_2$};
        \node at (6.3,4.25) {$e_3$};
        \draw (5.5,5.25) -- (0.25,5.25);
        \draw (5.5,4.75) -- (0.25,4.75);
        \draw (5.5,4.25) -- (0.25,4.25);

        \node at (0.25,5.25) {$\oplus$};
        \node at (0.75,5.25) {$\oplus$};
        \node at (1.25,4.75) {$\oplus$};
        \node at (1.75,4.25) {$\oplus$};
        \node at (2.25,4.75) {$\oplus$};
        \node at (2.75,5.25) {$\oplus$};
        \node at (3.25,4.25) {$\oplus$};
        \node at (3.75,4.25) {$\oplus$};
        \node at (4.25,4.25) {$\oplus$};

        \draw (6,5.5) rectangle (5.5,5);
        \draw (6,5) rectangle (5.5,4.5);
        \draw (6,4.5) rectangle (5.5,4);
        
        \node at (-0.45,3.25) {$M'$};
        \draw (0,3) rectangle (4.5,3.5);
        \draw (0,3) rectangle (0.5,3.5);
        \draw[fill=gray] (0.5,3) rectangle (1,3.5);
        \draw (1,3) rectangle (1.5,3.5);
        \draw (1.5,3) rectangle (2,3.5);
        \draw[fill=gray] (2,3) rectangle (2.5,3.5);
        \draw (2.5,3) rectangle (3,3.5);
        \draw (3,3) rectangle (3.5,3.5);
        \draw[fill=gray] (3.5,3) rectangle (4,3.5);
        \draw (4,3) rectangle (4.5,3.5);

        \draw[->] (0.25,3) -- (1.75,0.5);
        \draw[->] (0.75,3) -- (1.25,0.5);
        \draw[->] (1.25,3) -- (0.25,0.5);
        \draw[->] (1.75,3) -- (4.25,0.5);
        \draw[->] (2.25,3) -- (3.75,0.5);
        \draw[->] (2.75,3) -- (2.25,0.5);
        \draw[->] (3.25,3) -- (0.75,0.5);
        \draw[->] (3.75,3) -- (3.25,0.5);
        \draw[->] (4.25,3) -- (2.75,0.5);
        \node at (5,1.75) {$\Pi$};
        
        \node at (-0.65,0.25) {$\Pi(M')$};
        \draw (0,0) rectangle (4.5,0.5);
        \draw (0,0) rectangle (0.5,0.5);
        \draw (0.5,0) rectangle (1,0.5);
        \draw[fill=gray] (1,0) rectangle (1.5,0.5);
        \draw (1.5,0) rectangle (2,0.5);
        \draw (2,0) rectangle (2.5,0.5);
        \draw (2.5,0) rectangle (3,0.5);
        \draw[fill=gray] (3,0) rectangle (3.5,0.5);
        \draw[fill=gray] (3.5,0) rectangle (4,0.5);
        \draw (4,0) rectangle (4.5,0.5);

        \draw[->] (4.5,0.25) -- (6,0.25);
        \node at (6.5,0.25) {$\A^\rec$};
        \draw[->] (7,0.25) -- (7.5,0.25);
        \node at (7.85,0.25) {$K$};
    \end{tikzpicture}
    \caption{Illustration of the obfuscation procedure $\A^\obf$. The shaded cells indicate the position of the solution. We first add a random $e_i$ to each element (represented by $\oplus$) of the input, and permute it before giving it to $\A^\rec$.}
    \label{fig:enter-label}
\end{figure}

\paragraph{Algorithm $\A^\obf(M)$}\label{alg:O_obf(M)}
\begin{enumerate}
  \item Sample a random permutation $\Pi$ of $[r]$.
  \item Sample $e_i\sample G$ uniformly at random for $1\leq i<k$.
  \item Set $e_k\gets-\sum_{i=1}^{k-1}e_i$.
  \item Repeat $k^k\cdot\log{r}$ times : \label{line: outer loop of A_obf}
  \begin{enumerate}[label*=\arabic*.]
    \item Create a copy of the array $M'\gets M$.
    \item For each index $j\in\{1,2,\ldots,r\}$ :
    \begin{enumerate}[label*=\arabic*.]
      \item \label{line: sampling i} Sample $i\sample[k]$ uniformly at random
      \item Set $M'[j]\gets M[j]+e_i$.
    \end{enumerate}
    \item Let $K\gets\A^\rec(\Pi(M'))$ and $K'\gets \Pi^{-1}(K)$
    \item If $\sum_{i\in K'}M[i]=0$, return $K'$.
  \end{enumerate}
  \item If a solution has not been found yet, return $\bot$.
\end{enumerate}

Assuming we can sample elements of the underlying group and perform group operations in $\O(\log{|\G|})$ time, the runtime of $\A^\obf$ is $\O(k^k\log{r}(r\log{|\G|}+T)) = \Otilde(k^k\cdot T)$, since $\log{|\G|}=\Theta(\polylog(r))$ and $T = \Omega(r)$.

\begin{lemma}\label{oprime-guarantee}
    Let $S$ be any set of indices in $[r]$ of size $k$. Let $V$ be any array of elements of $G$ of length $k$ such that $\sum_{j=1}^kV[j]=0$. Recall that $\gamma$ is the success probability of $\A^\rec$ Then, for large enough $r$, 
    \begin{align*}
        \Pr_{M\sim D_{perm}}\left[\A^\obf(M)\,\mathrm{ succeeds}\mid M[S]=V\right] \geq \frac{\gamma}{2}
    \end{align*}
\end{lemma}

\begin{proof}
  Note that the distribution $D_1$ may be equivalently sampled as follows:
  \paragraph{Equivalent Sampling Algorithm for $D_1$}
  \begin{enumerate}[itemsep=0pt,topsep=5pt]
    \item Sample $M\gets D_0$
    \item Sample uniformly at random a set $S'$ of indices in $[r]$ of size $k$, and an array $V'$ of elements of $G$ of size $k$ such that $\sum_{j=1}^kV'[j]=0$
    \item Replace the entries of $M$ at the indices given by $S'$ with the respective elements of $V'$
    \item Output $M$
  \end{enumerate}
  For any $S$ and $V$, denote by $D_1^{S,V}$ the distribution that follows from fixing these to be the $S'$ and $V'$, respectively, in the sampling process above (rather than sampling them at random). Note that $D_1^{S,V}$ may not be the same as sampling an $M$ from $D_1$ conditioned on $M[S] = V$. Nevertheless, we have the following.

  \begin{claim}
    \label{claim:oprime-1}
    For any $S$ and $V$ as in the statement of the lemma, the following two distributions on $M$ are identical:
    \begin{enumerate}
      \item $M\sim D_1$ conditioned on $\left((M \text{ is permissible}) \wedge (M[S] = V)\right)$.
      \item $M\sim D_1^{S,V}$ conditioned on $(M \text{ is permissible})$.
    \end{enumerate}
  \end{claim}
  \begin{proofof}{\cref{claim:oprime-1}}
    $M$ being permissible implies that $M$ has exactly one $\kSUM$ solution. By definition of $S$ and $V$, the set $S$ is a $\kSUM$ solution. Further, in the above sampling procedure for $D_1$, the set $S'$ chosen there is also always a solution. So conditioning $D_1$ on having exactly one solution and also satisfying $M[S] = V$ implies that $S' = S$ and $V' = V$ in the sampling process. So the sampling process for this distribution is to set $S'=S$ and $V'=V$, and then condition on the resulting $M$ having exactly one solution. But this is also exactly the procedure of sampling $D_1^{S,V}$ conditioned on it having exactly one solution. This proves the claim.
  \end{proofof}

  Using \cref{claim:oprime-1} and the definition of $D_{perm}$, we can write the quantity we want to bound for the lemma as follows:
  \begin{align}
    \label{eq:oprime-1}
    \Pr_{M\sim D_{perm}}\left[\A^\obf(M)\,\mathrm{ succeeds}\mid M[S]=V\right] &= \Pr_{M\sim D_1}\left[ \A^\obf(M)\,\mathrm{ succeeds} \mid M[S]=V \wedge M \text{ is permissible} \right]\nonumber\\
    &= \Pr_{M\sim D_1^{S,V}}\left[\A^\obf(M)\,\mathrm{ succeeds}\mid M \text{ is permissible}\right]
  \end{align}

  Next we show that the probability that $M$ sampled from $D_1^{S,V}$ is not permissible is very small. So it is sufficient to bound the above probability without the conditioning on $M$'s permissibility.
  
  \begin{claim}
    \label{claim:oprime-2}
    For any $S$ and $V$ as in the statement of the lemma,
    \begin{align*}
      \prob{M\sim D_1^{S,V}}{M \text{ is not permissible}} \leq o(\gamma).
    \end{align*}
  \end{claim}
  \begin{proofof}{\cref{claim:oprime-2}}
    $M$ is not permissible iff there is an additional $\kSUM$ solution apart from $S$. The probability that there exists another solution is at most $\binom{r}{k} \cdot \size{G}^{-1} = o(\gamma)$, as argued for \cref{lemma:permissible inputs}. 
  \end{proofof}

  \begin{claim}
    \label{claim:oprime-3}
    For any $S$ and $V$ as in the statement of the lemma,
    \begin{align*}
      \Pr_{M\sim D_1^{S,V}}\left[\A^\obf(M)\,\mathrm{ succeeds}\right] \geq \gamma - o(\gamma).
    \end{align*}
  \end{claim}
  \begin{proofof}{\cref{claim:oprime-3}}
    Recall that $A^\obf$ proceeds by selecting a random permutation $\Pi$ and a random array $E$ of $k$ elements $(e_1,\dots,e_k)$ of $G$ that sum to $0$. Denote by $\Pi(S)$ the set that results from applying $\Pi$ to each index contained in $S$. $A^\obf(M)$ definitely succeeds if both of the following events happen in at least one of the iterations of step 4 of $A^\obf(M)$:
    \begin{enumerate}[itemsep=0pt,topsep=5pt]
      \item The element $e_i$ of $E$ added to $M[j]$ for each $j\in S$ is distinct.
      \item $\A_{\rec}(\Pi(M'))$ returns $\Pi(S)$
    \end{enumerate}
    In any iteration, the first event above happens with probability $(k!)/k^k > 1/k^k$, independently of $M$, $S$, and $V$. So the probability that it never happens in all $k^k\cdot\log(r)$ iterations is $\O\left(\textsf{exp}\left({-\log(r)}\right)\right) = o(\gamma)$.

    Suppose the first event does happen. Denote by $\widehat{M} = \Pi(M')$ the input provided to $\A_{rec}$ in the first iteration in which this event happens. For simplicity, suppose that the array $E$ gets added elementwise to $V$ in $M$ -- that is, for the least $j\in S$, the element added to $M[j]$ is $e_1$, and so on. Observe that the overall sampling procedure for $\widehat{M}$ can be equivalently described as follows:
    \begin{enumerate}[itemsep=0pt,topsep=5pt]
      \item Sample $M\gets D_1^{S,V}$. $M$ is a uniformly random array subject to the condition $M[S] = V$
      \item Add random elements of $E$ to elements of $M$ to get $M'$, subject to the above condition. $M'$ is now a uniformly random array subject to the condition $M[S] = V+E$
      \item Apply $\Pi$ to $M'$ to get $\widehat{M}$. $\widehat{M}$ is a uniformly random array subject to the condition $\widehat{M}[\Pi(S)] = \Pi_S(V+E)$, where $\Pi_S$ permutes the elements of the array $(V+E)$ according to match the change in the relative ordering of the indices in $S$ following the application of $\Pi$.
    \end{enumerate}
    As $E$ was chosen as a uniformly random array of elements that sum to $0$, and $V$ is an array of elements that sum to $0$, the sum $(V+E)$ is also a uniformly random array of elements that sum to $0$. As $\Pi$ is a uniformly random permutation, $\Pi(S)$ is a uniformly random set of indices of size $k$. Thus, if $M$ is sample from $D_1^{S,V}$ for any $S$ and $V$ as in the statement of the lemma, $\widehat{M}$ is distributed according to $D_1$.

    By our hypothesis, given a sample $\widehat{M}$ from $D_1$, the algorithm $\A_{rec}(\widehat{M})$ outputs a $\kSUM$ solution with probability $\gamma$. $\Pi(S)$ is always a solution of $\widehat{M}$ (under our current conditioning). The probability that there exists another solution is at most $o(\gamma)$, as argued earlier. Thus, the probability, when $M$ is sampled from $D_1^{S,V}$ and conditioning on the first event above happening, that $\A_{rec}(\widehat{M})$ outputs $\Pi(S)$ is at least $\gamma - o(\gamma)$.

    By the union bound, both the events above together happen with probability at least $(\gamma - o(\gamma))$, which proves the claim.
  \end{proofof}

  \noindent Now, putting together \cref{claim:oprime-2,claim:oprime-3} and (\ref{eq:oprime-1}) gives us the following statement, 
  \begin{align*}
    \Pr_{M\sim D_{perm}}\left[\A^\obf(M)\,\mathrm{ succeeds}\mid M[S]=V\right] \geq \gamma - o(\gamma) &&& \qedhere
  \end{align*}
\end{proof}

\subsubsection{Success Amplification}\label{sec: low density hardness amplification}
\pnote{$\G$ is the entire group ensemble, and $G$ is the group specific to instance size $r$. So it should be $G$ in most of the occurrences below. But maybe that makes the notation for the graphs ambiguous. Might be better to change all the $G$ to  $\G$ earlier in the section, then, after pointing out this notation.}
Our next step is amplifying the success probability to almost 1. The main idea is replacing many of the entries of $M$ with random elements from $\G$, and running $\A^\obf$ on the result. If none of the entries in the solution get replaced, we end up with an almost random array sampled from $D_1^{S, V}$. Repeating this procedure enough times and utilizing the guarantees provided by $\A^\obf$, we can amplify the success probability significantly. The algorithm on input $M$ is formally described below.

\paragraph{Algorithm $\A^\amp(M)$}\label{alg:A_det(M)}
\begin{enumerate}
  \item Repeat $\frac{64\log{r}}{\gamma^{2k+2}}$ times:
  \begin{enumerate}
    \item[1.1.] Make a new copy of $M$ and call it $M'$
    \item[1.2.] Repeat $r\log{\frac{1}{\gamma}}$ times:
    \begin{enumerate}
      \item[1.2.1.] Pick $i$ uniformly at random from $[r]$
      \item[1.2.2.] Replace $M'[i]$ with a random element from $\G$ different from $M'[i]$ 
    \end{enumerate}
    \item[1.3.] Call $\A^\obf(M')$
    \item[1.4.] If the previous call succeeds and returns a $k$-tuple that was unchanged from $M$ to $M'$, return that.
  \end{enumerate}
\end{enumerate}
It is easy to see that the algorithm never returns a wrong answer. Note that the inner loop takes $\O(\log{|\G|+\log{r}})$ time if we assume that sampling from a set takes time logarithmic in the size of that set. The outer loop takes $\Otilde(k^k\cdot T)+\Otilde\left(r\log{|\G|}\log{\frac{1}{\gamma}}+r\log{|\G|}\right)$ time since the runtime of $\A^\obf$ is $\Otilde(k^k\cdot T)$ (recall that $T$ was defined as the runtime of $\A^\rec$). Since $\log{|\G|}=\mathcal{O}(\polylog(r))$, $\frac{1}{\gamma}=\mathcal{O}(\polylog(r))$ and $T=\Omega(r)$, the expression simplifies to $\Otilde(k^kT)$. The runtime of the whole algorithm is therefore $\Otilde\left(\frac{\log{r}}{\gamma^{2k+2}}\right)\cdot\Otilde(k^kT) = \Otilde\left(T\cdot \left(k/\gamma^2\right)^k\right)$, as desired. To analyze the success probability of $\A^\amp$, we start with a couple of definitions.

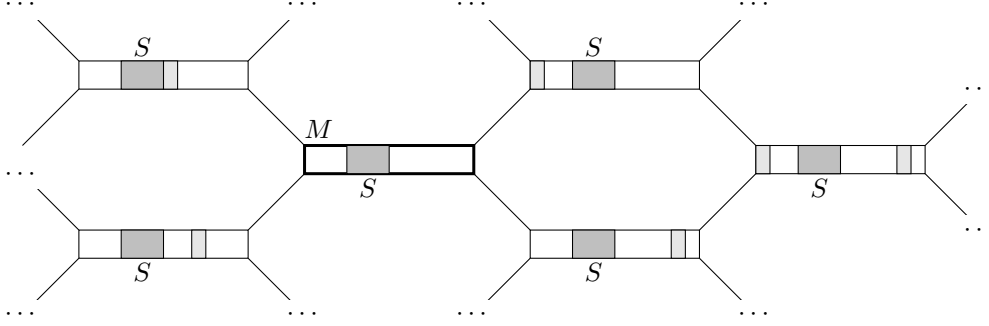
\begin{figure}
    \centering
    \begin{tikzpicture}[scale=0.75]
        \draw (0,0) -- (-1,-1);
        \draw (0,0.5) -- (-1,1.5);
        \draw (3,0) -- (4,-1);
        \draw (3,0.5) -- (4,1.5);
        \draw (7,1.5) -- (8,0.5);
        \draw (7,-1) -- (8,0);

        \draw (-4,2) -- (-5,3) node[fill=white] {$\cdots$};
        \draw (-1,2) -- (0,3) node[fill=white] {$\cdots$};
        \draw (4,2) -- (3,3) node[fill=white] {$\cdots$};
        \draw (7,2) -- (8,3) node[fill=white] {$\cdots$};
        
        \draw (-4,-1.5) -- (-5,-2.5) node[fill=white] {$\cdots$};
        \draw (-1,-1.5) -- (0,-2.5) node[fill=white] {$\cdots$};
        \draw (4,-1.5) -- (3,-2.5) node[fill=white] {$\cdots$};
        \draw (7,-1.5) -- (8,-2.5) node[fill=white] {$\cdots$};

        \draw (-4,1.5) -- (-5,0.5);
        \draw (-4,-1) -- (-5,0) node[fill=white] {$\cdots$};
        \draw (11,0.5) -- (12,1.5) node[fill=white] {$\cdots$};
        \draw (11,0) -- (12,-1) node[fill=white] {$\cdots$};

        \node at (0.25,0.8) {$M$};
        \node at (1.125,-0.25) {$S$};
        \draw[line width=0.4mm] (0,0) rectangle (3,0.5);
        \draw[fill=lightgray] (0.75,0) rectangle (1.5,0.5);

        \node at (9.125,-0.25) {$S$};
        \draw (8,0) rectangle (11,0.5);
        \draw[fill=white!80!gray] (8,0) rectangle (8.25,0.5);
        \draw[fill=lightgray] (8.75,0) rectangle (9.5,0.5);
        \draw[fill=white!80!gray] (10.5,0) rectangle (10.75,0.5);

        \node at (-4+1.125,-1.75) {$S$};
        \draw (-4,-1.5) rectangle (-1,-1);
        \draw[fill=lightgray] (-4+0.75,-1.5) rectangle (-4+1.5,-1);
        \draw[fill=white!80!gray] (-4+2,-1.5) rectangle (-4+2.25,-1);

        \node at (4+1.125,-1.75) {$S$};
        \draw (4,-1.5) rectangle (7,-1);
        \draw[fill=lightgray] (4+0.75,-1.5) rectangle (4+1.5,-1);
        \draw[fill=white!80!gray] (4+2.5,-1.5) rectangle (4+2.75,-1);

        \node at (-4+1.125,2.25) {$S$};
        \draw (-4,1.5) rectangle (-1,2);
        \draw[fill=lightgray] (-4+0.75,1.5) rectangle (-4+1.5,2);
        \draw[fill=white!80!gray] (-4+1.5,1.5) rectangle (-4+1.75,2);

        \node at (4+1.125,2.25) {$S$};
        \draw (4,1.5) rectangle (7,2);
        \draw[fill=white!80!gray] (4,1.5) rectangle (4+0.25,2);
        \draw[fill=lightgray] (4+0.75,1.5) rectangle (4+1.5,2);

    \end{tikzpicture}
    \caption{The Hamming graph representing the inner loop of $\A^\amp(M)$. Each iteration resamples a random entry of $M$ and corresponds to taking a random step in the graph, assuming this does not destroy the solution $S$. Intuitively, the oracle cannot consistently fail on all the vertices we visit, and thus the process must eventually return the set $S$. The proof of this statement amounts to bounding the edge expansion of this graph.}
    \label{fig:hamming_graph}
\end{figure} 

\begin{definition}[Correspondence Graph]\label{def:correspondence graph}
  For a permissible array $M$ with the $\kSUM$ solution $S\subset [r]$, we define its \emph{correspondence graph} $G_M$ as follows: create a vertex for every array $M'\in\G^r$ such that $M[S]=M'[S]$. Two vertices are connected by an edge if and only if the corresponding arrays only differ in one index.
\end{definition}
\begin{lemma}\label{lemma: G is hamming}
  $G_M$ is isomorphic to the Hamming graph $H(r-k, |\G|)$ (see \cref{def:hamming graph}). 
\end{lemma}
\begin{proof}
  We are allowed to change $r-k$ entries of $M$, and there are $|\G|$ possibilities for each of those indices. The isomorphism follows from definition.
\end{proof}
\begin{lemma}\label{lemma: almost all vertices permissible}
  The fraction of vertices corresponding to arrays which are not permissible is $o(\gamma)$ in any correspondence graph $G_M$.
\end{lemma}
\begin{proof}
  We will show that for any permissible $M$, if we choose the $r-k$ entries outside its $\kSUM$ solution uniformly at random from $\G$, we end up with a permissible array with probability $1-o(\gamma)$. This is equivalent to the lemma statement.

  The proof of \cref{lemma:permissible inputs} establishes this result almost directly. With very high probability, there will be no new solutions other than the $k$ fixed entries. We can bound the probability of a new solution by 
$$\frac{1}{|\G|}\left(\binom{r-k}{k}\binom{k}{0}+\binom{r-k}{k-1}\binom{k}{1}+\cdots\binom{r-k}{1}\binom{k}{k-1}\right)$$
  Since $k$ is a constant, this sum is $\mathcal{O}\left(r^{k\,(1-1/\Delta)}\right) = o(\gamma)$ assuming $\Delta \leq {k\log{r}}/({k\log{r}+(\alpha+1)\log{\log{r}}})$.
\end{proof}

\noindent For a pictorial representation of $G_M$, see \cref{fig:hamming_graph}. Note that if we are lucky enough to not destroy the solution while running the inner loop, we end up taking $r\log{\frac{1}{\gamma}}$ steps in this graph.

\begin{definition}\label{def:correspondence power graph}
  For a permissible array $M$, its \emph{correspondence power graph} $G^*_M$ is a multigraph which has the same vertices as $G_M$, and has $t$ edges between $M_1$ and $M_2$ where $t$ is the number of paths of length $r\log{\frac{1}{\gamma}}$ from $M_1$ to $M_2$ in $G_M$. In other words, we can obtain the adjacency matrix of $G^*_M$ by raising the adjacency matrix of $G_M$ to the power $r\log{\frac{1}{\gamma}}$.
\end{definition}
\begin{lemma}\label{lemma: G* expansion}
  The algebraic expansion of $G^*_M$ is $\left((r-k)(|\G|-1)-|\G|\right)^{r\log{\frac{1}{\gamma}}}$.
\end{lemma}
\begin{proof}
  By \cref{lemma:hamming,lemma: G is hamming}, we know that for $G_M$, the highest eigenvalue is $(r-k)(|\G|-1)$, the second highest eigenvalue is $(r-k)(|\G|-1)-|\G|$, and the lowest eigenvalue is $k-r$. Note that raising a matrix to a certain power also raises its eigenvalues to the same power, and the algebraic expansion of a graph with eigenvalues $\lambda_1\geq\lambda_2\geq\cdots\geq\lambda_n$ is defined as $\max(|\lambda_2|, |\lambda_n|)$. Since $(r-k)(|\G|-1)-|\G|>r-k$ for any $r>k$, the lemma follows.
\end{proof}
\begin{lemma}\label{lemma: G* is d-regular}
  Each vertex of $G^*_M$ has degree $\left((r-k)(|\G|-1)\right)^{r\log{\frac{1}{\gamma}}}$
\end{lemma}
\begin{proof}
  \cref{lemma: G is hamming} implies that $G_M$ is a regular graph of degree $d = (r-k)(|\G|-1)$. Therefore, starting from any vertex $v$, there are exactly $d^l$ paths of length $l$. By definition, this implies $G^*_M$ is $\left((r-k)(|\G|-1)\right)^{r\log{\frac{1}{\gamma}}}$-regular.
\end{proof}

\noindent Note that taking $r\log{\frac{1}{\gamma}}$ steps in $G_M$ is equivalent to taking 1 step in $G^*_M$. 

\begin{definition}[Good Vertices]\label{def:good vertex}
  We call an array $X$ \emph{good} if it is permissible and $\A^\obf$ has at least a $\frac{\gamma}{4}$ success probability on it.
  $$\Pr[\A^\obf(X)\,\mathrm{succeeds}]\geq \frac{\gamma}{4}$$
  We call a vertex \emph{good} if it corresponds to a good array.
\end{definition}

\begin{lemma}\label{lemma: at least gamma/8 fraction are good}
  Let $M$ be a large enough permissible array. At least a $\frac{\gamma}{8}$ fraction of the vertices of $G_M$ are good.
\end{lemma}
\begin{proof}
  Note that only a $o(\gamma)$ fraction of the vertices in $G_M$ are not permissible. We choose $r$ large enough so that this fraction does not exceed $\frac{\gamma}{8}$. The rest of the vertices exactly correspond to the support of $D_{perm}$ which has the same $\kSUM$ solution in the same position as $M$. Since $D_{perm}$ is uniform on its support (\cref{lemma:permissible inputs}), conditioning on a random vertex of $G_M$ being permissible is the same as sampling from $D_{perm}$ and conditioning on it having the same $\kSUM$ solution in the same place as $M$. \cref{oprime-guarantee} now yields
  \begin{equation}\label{eq:o obf success guarantee on G_M}
    \underset{M'\sim G_M}{\E}[\Pr[\A^\obf(M')\,\mathrm{succeeds}]\mid M' \text{ is permissible}] \geq \frac{\gamma}{2}
  \end{equation}
  We proceed via a proof by contradiction. Let us assume that the lemma is false. This implies that at least a $1-\frac{\gamma}{8}$ fraction of the vertices in $G$ are either not permissible or have success probability less than $\frac{\gamma}{4}$. Since the fraction of vertices not permissible is at most $\frac{\gamma}{8}$, at least a $1-\frac{\gamma}{4}$ fraction of the vertices are permissible but not good. Then, we have:
  \begin{align*}
    \underset{M'\sim G_M}{\E}&[\Pr[\A^\obf(M')\,\mathrm{succeeds}]\mid M' \text{ is permissible}]
                               \intertext{This expected value is maximized when all the good vertices have success probability $1$ and all the other permissible arrays in $G_M$ have success probability $\frac{\gamma}{4}$.}
    &\leq \frac{\frac{\gamma}{8}}{1-\frac{\gamma}{8}}+\frac{1-\frac{\gamma}{4}}{1-\frac{\gamma}{8}}\cdot\frac{\gamma}{4}
      \intertext{This simplifies as follows.}
    &= \frac{\gamma}{8-\gamma} + \frac{(4-\gamma)\gamma}{2(8-\gamma)}
      =\frac{6\gamma-\gamma^2}{2(8-\gamma)} < \frac{8\gamma-\gamma^2}{2(8-\gamma)} = \frac{\gamma}{2}
  \end{align*}
  This inequality directly contradicts \cref{eq:o obf success guarantee on G_M}, and hence we are done.
\end{proof}

\begin{definition}[Bad Vertices]\label{def:bad vertex}
  Let $M$ be a permissible array. We call a vertex $X$ of $G^*_M$ \emph{bad} if it is permissible and less than $\frac{\gamma}{16}$ of its outgoing edges connect to a good vertex.
\end{definition}
\begin{lemma}\label{lemma: few bad vertices in expander}
  Let $M$ be a large enough permissible array. The fraction of bad vertices in $G^*_M$ is $o\left(\frac{1}{\log{r}}\right)$.
\end{lemma}
\begin{proof}
  Let $S$ be the set of bad vertices and $T$ be the set of good vertices. We denote the set of all vertices of $G^*_M$ by $V$. We know from \cref{lemma: G* is d-regular} that $G^*_M$ is a regular graph; we shall call its degree $d$. We will also denote the algebraic expansion of $G^*_M$ by $\lambda$. The expander mixing lemma (\cref{lemma:expander_mixing}) now implies
  \begin{align*}
    &\left|E(S,T)-\frac{d\cdot|S|\cdot|T|}{|V|}\right| \leq \lambda\sqrt{|S|\cdot|T|} \quad\Longrightarrow\quad E(S,T)\geq \frac{d\cdot|S|\cdot|T|}{|V|} - \lambda\sqrt{|S|\cdot|T|}.
  \end{align*}
  The number of edges between $S$ and $T$ is at most $\frac{d\gamma|S|}{16}$ since each bad vertex, by \cref{def:bad vertex}, has at most $\frac{d\gamma}{16}$ edges connecting to a good vertex. Therefore,
  \begin{align*}
    \frac{d\gamma|S|}{16}\geq \frac{d\cdot|S|\cdot|T|}{|V|} - \lambda\sqrt{|S|\cdot|T|} \quad &\Longrightarrow \quad 
                                                                                                \frac{d\gamma}{16}\cdot\frac{|S|}{|V|} \geq d\cdot\frac{|S|}{|V|}\cdot\frac{|T|}{|V|} - \lambda\sqrt{\frac{|S|}{|V|}\cdot\frac{|T|}{|V|}}\\
                                                                                              &\Longrightarrow \quad \frac{|S|}{|V|}\left(\frac{d|T|}{|V|}-\frac{d\gamma}{16}\right)\leq \lambda\sqrt{\frac{|S|}{|V|}\cdot\frac{|T|}{|V|}}\\
                                                                                              &\Longrightarrow \quad \sqrt{\frac{|S|}{|V|}} \leq \frac{\lambda\sqrt{\frac{|T|}{|V|}}}{\frac{d|T|}{|V|}-\frac{d\gamma}{16}}\\
                                                                                              &\Longrightarrow \quad \frac{|S|}{|V|} \leq \left(\frac{\lambda}{d}\right)^2\cdot\frac{\frac{|T|}{|V|}}{\left(\frac{|T|}{|V|}-\frac{\gamma}{16}\right)^2}.
  \end{align*}
  Since $T$ is the set of good vertices, \cref{lemma: at least gamma/8 fraction are good} implies $\frac{|T|}{|V|}\geq\frac{\gamma}{8}$. Therefore,
  \begin{align*}
    &\frac{|S|}{|V|} \leq \left(\frac{\lambda}{d}\right)^2\cdot\frac{\frac{\gamma}{8}}{\left(\frac{\gamma}{8}-\frac{\gamma}{16}\right)^2} = \left(\frac{\lambda}{d}\right)^2\cdot\frac{32}{\gamma}
      \intertext{We now plug in the values of $\lambda$ and $d$ from \cref{lemma: G* expansion,lemma: G* is d-regular} respectively to get}
    &\frac{|S|}{|V|} \leq \left(1-\frac{|\G|}{(r-k)(|\G|-1)}\right)^{2r\log{\frac{1}{\gamma}}}\cdot\frac{32}{\gamma} < \left(1-\frac{1}{r}\right)^{2r\log{\frac{1}{\gamma}}}\cdot\frac{32}{\gamma}
      \intertext{We can further bound the above expression by using the inequality $\left(1-\frac{1}{r}\right)^r \leq \frac{1}{e}$}
    &\frac{|S|}{|V|} < \frac{32\textsf{exp}\left({-2\log{\frac{1}{\gamma}}}\right)}{\gamma}=32\gamma=\Theta\left(\frac{1}{\log^{\alpha}{r}}\right)=o\left(\frac{1}{\log{r}}\right)
  \end{align*}
  We have thus shown that the fraction of bad vertices in the correspondence power graph of any permissible array is $o(1/\log{r})$. 
\end{proof}

We are now ready to conclude our second step.

\begin{lemma}\label{lemma: low density amplification}
  The average-case success probability of $\A^\amp$ is $1-o\left(\frac{1}{\log{r}}\right)$. This algorithm works at any $\Omega\left(\frac{1}{\polylog(r)}\right)$ density $\Delta \leq \frac{k\log{r}}{k\log{r}+(\alpha+1)\log{\log{r}}}$ for large enough $r$.
\end{lemma}
\begin{proof}
  With probability $1-o(\gamma)=1-o(1/\log{r})$, the input array $M$ is permissible (see \cref{lemma:permissible inputs}). In that case, $M$ has a unique solution. \cref{lemma: few bad vertices in expander} tells us that among all arrays that contain this exact solution at this exact position, only a $o(1/\log{r})$ fraction can be \emph{bad}. So we can assume with probability $1-o(1/\log{r})$ that $M$ has at least a $\frac{\gamma}{16}$ fraction of its edges leading to good vertices.

  Each iteration of the outer loop of $\A^\amp$ makes $r\log{\frac{1}{\gamma}}$ replacements before calling $\A^\obf$. The original solution is preserved if we choose one of the $r-k$ entries not in the solution at every step. This happens with probability $(1-\frac{k}{r})^{r\log{\frac{1}{\gamma}}}>\textsf{exp}\left({-2k\log{\frac{1}{\gamma}}}\right) = \gamma^{2k}$, since $(1-\frac{1}{t})^t>\frac{1}{e^2}$ for $t>2$. Note that when we preserve the solution, the inner loop takes one random step in $G^*_M$. With probability at least $\frac{\gamma}{16}$, this lands us into a good vertex where $\A^\obf$ succeeds with probability at least $\frac{\gamma}{4}$. Thus, each iteration of the outer loop has a probability at least $\frac{\gamma^{2k+2}}{64}$ of recovering the original solution for any input which is permissible and not a bad vertex. Observe that different iterations of the outer loop are independent once we condition on the input array. Since the outer loop runs $\frac{64\log{r}}{\gamma^{2k+2}}$ times, the probability of recovering the solution in at least one of the iterations is 
  \begin{align*}
    1-\left(1-\frac{\gamma^{2k+2}}{64}\right)^{({64\log{r}})\big/{\gamma^{2k+2}}}>1-\frac{1}{\textsf{exp}\left({\log(r)}\right)}=1-\frac{1}{r}=1-o\left(\frac{1}{\log{r}}\right)&\qedhere
  \end{align*}
\end{proof}

\subsection{Hardness Amplification for Vector \texorpdfstring{$k$-SUM}{} at Density \texorpdfstring{$\Delta \leq 1$}{}}
\label{sec: success amp upto 1 for vksum}

The statement of \cref{thm: recovery algorithm success amplification} can be further strengthened if the underlying group ensemble has some extra structure. In some special cases, we can show an analogous results for all densities $\Delta\leq 1$. In this subsection, we will handle the special case where $\G$ is $\Z_q^m$ for some fixed integer $q$, and addition is defined as pointwise addition modulo $q$. Note that any element of $\Z_q^m$ can be represented by a vector of size $m$, hence we will represent the input as an $m\times r$ matrix whose columns represent the elements. Each entry of this matrix will be in $\Z_q$. At density $\Delta$, we have the relation $m\Delta\log{q}=k\log{r}$. As long as $1\geq \Delta=\Omega\left(\frac{1}{\polylog(r)}\right)$, this implies that $m = \frac{k\log{r}}{\Delta\log{q}}=\Theta(\polylog(r))$. The $k$-XOR problem is the special case of this with $q=2$.

\begin{theorem}[Hardness Amplification for Vector $k$-SUM]\label{thm: vksum success amplification}
    For $k\geq 3$ and density $\Omega\left(\frac{1}{\polylog(r)}\right) \leq \Delta \leq 1$, suppose there exists an algorithm $\A^\rec$ that runs in time $T(r) = \Omega(r)$, and solves the planted search \vkSUM problem at density $\Delta$ ($\kSUM$ over $\G_{\textsc{vector-}(q,k)\textsc{-SUM}}^{(\Delta)}$) with success probability $\gamma = \Omega\left(\frac{1}{\polylog(r)}\right)$. Then, there is an algorithm that runs in time $\Otilde\left(T\cdot (k/\gamma^2)^k\right)$, and solves the same problem with success probability $\left(1-o\left(\frac{1}{\log{r}}\right)\right)$.    
\end{theorem}

\begin{proof}
  As before, we will assume w.l.o.g. that $\gamma=\Theta\left(\frac{1}{\log^\alpha{r}}\right)$ for some constant $\alpha>1$.  Let us define, $$\Delta_0 := \frac{k\log{r}}{k\log{r}+(\alpha+1)\log{\log{r}}}$$
  We have already proven this result for densities $\Delta\leq \Delta_0$ in \cref{sec:hardness amplification for general group}. To complete the proof, we will now consider the other case and assume that $\Delta\in(\Delta_0,1]$.

  Note that this implies that permissible matrices may no longer take up a $1-o(\gamma)$ fraction of the input space, and we can no longer pretend the input is permissible without a significant loss in success probability anymore.
  
  First, we use $\A^\rec$ to construct an algorithm $\A_0$ for solving the same problem at density $\Delta_0$. On input $M_0$ (which is sampled from $D_1^{\Delta_0}$), this algorithm $\A_0$ simply throws away a randomly chosen $1-\frac{\Delta_0}{\Delta}$ fraction of the rows to get a density $\Delta$ instance $M$. It then checks if $\A^\rec(M)$ returns a $k$-tuple which is also a solution for the original input $M_0$, and if so, returns it. Observe that since the rows of $M_0$ are independently chosen, $M$ has the same distribution as $D_1^\Delta$. Therefore, $\A^\rec$ returns a solution with probability $\gamma$. Any $k$-tuple that sum to 0 in $M_0$ obviously also sum to 0 in $M$. If $M$ has $c$ solutions of vector $k$-SUM, we can argue by symmetry that they are all equally likely to be a solution for $M_0$. Now, note that since $M$ has density $\leq1$, the expected value of $c$ is less than 3. If the planted solution is $S$,
  $$\E[c]=\sum_{\substack{|\kappa|=k \\ \kappa\subset [r]}}\Pr[\kappa \text{ is a solution}]=1+ \sum_{\substack{|\kappa|=k \\ \kappa\subset [r]\\ \kappa\neq S}}\Pr[\kappa \text{ is a solution}]=1+\frac{\binom{r}{k}-1}{q^m}<3.$$
  We can therefore apply Markov's inequality (\cref{lemma:markov}) to conclude that with probability at least $1/2$, there are less than 6 \vkSUM solutions for $M$. So with probability at least $\gamma\cdot\frac{1}{2}\cdot\frac{1}{5}$, we will call $\A^\rec$ on a matrix with at most 5 solutions, $\A^\rec$ will succeed, and the solution it returns will also be a solution for our original input $M_0$. The success probability of $\A_0$ is therefore at least $\frac{\gamma}{10}=\Omega\left(\frac{1}{\polylog(r)}\right)$. The runtime of $\A_0$ is clearly $\mathcal{O}(mr)+T=\Otilde(T)$.

  We now use our success amplification procedure described in \cref{sec:hardness amplification for general group} to obtain an average-case recovery algorithm $\A_0^\amp$ with runtime $\Otilde(T)$ that solves density $\Delta_0$ instances with probability $1-o(1/\log{r})$.

  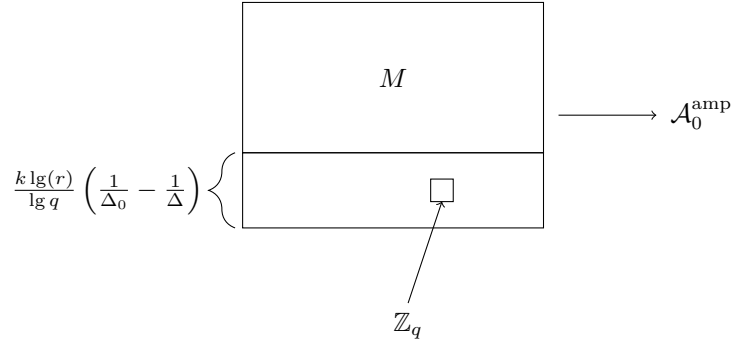
\begin{figure}
      \centering
      \begin{tikzpicture}
          \draw (0,0) rectangle node {$M$} (4,2);
          \draw (0,-1) rectangle (4,0);
          \draw [decorate,decoration={brace,amplitude=10pt}] (-0.1,-1) -- (-0.1,0);
          \node at (-1.8,-0.5) {$\frac{k\lg(r)}{\lg{q}}\left(\frac{1}{\Delta_0}-\frac{1}{\Delta}\right)$};
          \draw (2.5,-0.35) rectangle (2.8,-0.65);
          \draw[<-] (2.65,-0.65) -- (2.2,-2);
          \node at (2.2,-2.3) {$\bbZ_q$};
          \draw[->] (4.2,0.5) -- (5.5,0.5);
          \node at (6.1,0.5) {$\A_0^\amp$};
      \end{tikzpicture}
      \caption{Illustration of how to lift the hardness amplification to all densities $\Delta \leq 1$ for the case of vector $k$-SUM, with $k$-XOR being a special case of $q=2$. Here, a total of $\frac{k\lg(r)}{\lg{q}}\left(\frac{1}{\Delta_0}-\frac{1}{\Delta}\right)$ rows are added to the solution, where each entry is sampled i.i.d. from $\bbZ_q$, and the result is given to $\A_0^\amp$.}
      \label{fig:vector ksum reduction}
  \end{figure}
  
  Now we are ready to construct the final algorithm $\B$. Given a density $\Delta$ instance $M$ with dimensions $m\times r$, our algorithm does the following (see also \cref{fig:vector ksum reduction} for an illustration).
  \paragraph{Algorithm $\B(M)$}
  \begin{enumerate}
    \item Define $t:=(\log r)^{\alpha+2}$
    \item Repeat $t$ times:
    \begin{enumerate}
      \item[2.1.] Add $\frac{k\log(r)}{\log{q}}\left(\frac{1}{\Delta_0}-\frac{1}{\Delta}\right)$ rows to $M$. Each new entry is chosen uniformly at random from $\Z_q$.
      \item[2.2.] Call $\A_0^\amp$ on the new matrix $M_0$.
      \item[2.3.] If a solution was returned in the last step, check if that is also a solution for $M$. If so, return it.
    \end{enumerate}
  \end{enumerate}

  Observe that $M$ is guaranteed to have at least one solution $S$ since it is sampled from $D_1^\Delta$. The probability of $S$ being a $k$-XOR solution for $M_0$ is exactly $q^{\frac{k\log(r)}{\log{q}}\left(\frac{1}{\Delta}-\frac{1}{\Delta_0}\right)}$ since the probability of the original solution being preserved drops by a factor of $q$ for each additional row added. We now bound the probability $p$ that at least one iteration of the inner loop ran $\A^\amp_0$ on a matrix where $S$ was a solution. 
  \begin{align*}
    p &\geq 1-\left(1-q^{\frac{k\log(r)}{\log{q}}\left(\frac{1}{\Delta}-\frac{1}{\Delta_0}\right)}\right)^t = 1-\left(1-r^{k\,\left(\frac{1}{\Delta}-\frac{1}{\Delta_0}\right)}\right)^t
        \intertext{Now, since $\left(1-\frac{1}{n}\right)^n\leq \frac{1}{e}$ and $\Delta \leq 1$, it follows that,}\\
      &\geq 1- \textsf{exp}\left(-t\cdot r^{k\,\left(1-\frac{1}{\Delta_0}\right)}\right)
        \intertext{We now substitute in $\Delta_0 = \frac{k\log{r}}{k\log{r}+(\alpha+1)\log\log{r}}$ to get,}\\
      &= 1- \textsf{exp} \left( -t\cdot r^{\frac{-(\alpha+1)\log\log{r}}{\log{r}}} \right)
        \intertext{We now let $t=(\log r)^{\alpha+2}$ and obtain,} \\
      &= 1-\textsf{exp}\left(-\left(\log{r}\right)^{\alpha+2}\cdot\left(\log{r}\right)^{-(\alpha+1)}\right) \\
      &= 1-\textsf{exp}\left(-\log(r)\right) =1-o\left(\frac{1}{\log{r}}\right)
  \end{align*}
  When $M_0$ does have a solution, it is easy to see that it has the same distribution as $D_1^{(\Delta_0)}$ and hence, $\A^\amp_0$ returns a solution with probability $1-o(1/\log{r})$. Taking a union bound, we conclude that with probability $1-o(1/\log{r})$, we call $\A_0^\amp$ on a matrix with a \vkSUM solution at least once \emph{and} it returns that solution. Any solution to $M_0$ is always a solution to $M$, and our algorithm therefore returns it. We thus have the required success probability.

  The runtime of $\B$ is clearly $t$ times the runtime of $\A_0^\amp$. Since $\A_0^\amp$ runs in $\Otilde(T)$ and $t=\mathcal{O}(\polylog(r))$, our algorithm also runs in $\Otilde(T)$ time.
\end{proof}

\begin{corollary}
\label{cor:kxor-amp-1}
    If the planted search $\kXOR$ problem at density $\Omega\pfrac{1}{\polylog(r)}\leq\Delta\leq1$ is hard to solve with probability $1-o\pfrac{1}{\log{r}}$ in time $\Otilde(T)$, it is also hard to solve with probability $\gamma = \Omega\pfrac{1}{\polylog(r)}$ in time $\Otilde(T \cdot (\gamma^2/k)^k)$.
\end{corollary}
\begin{proof}
    This follows from setting $q=2$ in \cref{thm: vksum success amplification} and taking the contrapositive.
\end{proof}

\begin{corollary}[Strong Hardness of $k$-XOR]
    \label{cor:kxor-amp}
    For any constant $k\geq 3$, the average-case $k$-XOR conjecture (\cref{conj:kxor}) implies that any algorithm that, given $r$ uniformly random vectors from $\bbF_2^m$ for $m = k\log{r}$, can find a set of $k$ that sum to $0$ with probability $\Omega({1}/{\polylog(r)})$ takes time at least $r^{\ceil{k/2}-o(1)}$.
\end{corollary}

\begin{proof}
    This follows from setting $\Delta=1$ and $k = O(1)$ in \cref{cor:kxor-amp-1}.
\end{proof}

\subsection{Hardness Amplification for Modular \texorpdfstring{$\kSUM$}{} at Density \texorpdfstring{$\Delta \leq 1$}{}}
\label{sec: success amp upto 1 for modular ksum}

In this subsection, we will prove an analogue of \cref{thm: vksum success amplification} for the special case where $\G$ is $\Z_{2^m}$ (see \cref{eq: k-MSUM def}). At density $\Delta$, we have the relation $m\Delta=k\log{r}$. Note that the input will now be an array $M$ of size $r$ where each entry is between $0$ and $2^m-1$.

\begin{theorem}[Hardness Amplification for Modular $k$-SUM]\label{thm: k-Msum success amplification}
    For $k\geq 3$ and density $\Omega\left(\frac{1}{\polylog(r)}\right) \leq \Delta \leq 1$, suppose there exists an algorithm $\A^\rec$ that runs in time $T(r) = \Omega(r)$, and solves the planted search $\kSUM$ problem over $\G_{\text{$k$-SUM}}^{(\Delta)}$ with success probability $\gamma = \Omega\left(\frac{1}{\polylog(r)}\right)$. Then, there is an algorithm that runs in time $\Otilde\left(T\cdot (k/\gamma^2)^k\right)$, and solves the same problem with success probability $\left(1-o\left(\frac{1}{\log{r}}\right)\right)$.    
\end{theorem}
\begin{proof}
    We follow the proof of \cref{thm: vksum success amplification} very closely. We will again assume w.l.o.g. that $\gamma=\Theta\left(\frac{1}{\log^\alpha{r}}\right)$ for some constant $\alpha>1$.  We define, $$\Delta_0 := \frac{k\log{r}}{k\log{r}+(\alpha+1)\log{\log{r}}}$$
    Since we already proved this result for densities $\Delta\leq \Delta_0$ in \cref{sec:hardness amplification for general group}, we only consider the case $\Delta\in(\Delta_0,1]$.
  
    First, we use $\A^\rec$ to construct an algorithm $\A_0$ for solving the same problem at density $\Delta_0$. On input $M_0$ (which is sampled from $D_1^{\Delta_0}$), this algorithm $\A_0$ simply reduces each entry modulo $2^{k\log{r}/\Delta}$ to get a density $\Delta$ instance $M$. It then checks if $\A^\rec(M)$ returns a $k$-tuple which is also a solution for the original input $M_0$, and if so, returns it. Observe that since the entries of $M_0$ are independently sampled uniformly from $\Z_{2^{k\log{r}/\Delta_0}}$ and the latter modulus is a multiple of the new modulus, $M$ has the same distribution as $D_1^\Delta$. Therefore, $\A^\rec$ returns a solution with probability $\gamma$. Any $k$-tuple that sum to 0 in $M_0$ obviously also sum to 0 in $M$, since (in the following $a \mid b$ means that ``$a$ divides $b$''),
    \begin{align*}
        \sum_{i\in\kappa}M_0[i]&=0 \\
        \implies &2^{k\log{r}/\Delta_0}\mid \sum_{i\in\kappa}M_0[i] &&\text{By definition of }+\text{ in }M_0 \\
        \implies &2^{k\log{r}/\Delta}\mid \sum_{i\in\kappa}M_0[i] &&\text{Since }\Delta_0<\Delta \\
        \implies &2^{k\log{r}/\Delta}\mid \sum_{i\in\kappa}(M_0[i] \mod{2^{k\log{r}/\Delta}}) \\
        \implies  &2^{k\log{r}/\Delta}\mid \sum_{i\in\kappa}M[i] &&\text{By design of }\A_0 \\
        \implies &\sum_{i\in\kappa}M[i]=0 &&\text{By definition of }+\text{ in }M
    \end{align*}
    If $M$ has $c$ solutions of $\kSUM$, we can argue by symmetry that they are all equally likely to be a solution for $M_0$. Since $M$ has density $\leq1$, the expected value of $c$ is less than 3. We can therefore apply Markov's inequality (\cref{lemma:markov}) to conclude that with probability at least $1/2$, there are less than 6 $\kSUM$ solutions for $M$. As before, this implies that the success probability of $\A_0$ is at least $\frac{\gamma}{10}=\Omega\left(\frac{1}{\polylog(r)}\right)$. The runtime of $\A_0$ is clearly $\O(mr)+T=\Otilde(T)$.

  We now use our success amplification procedure described in \cref{sec:hardness amplification for general group} to obtain an average-case recovery algorithm $\A_0^\amp$ with runtime $\Otilde(T)$ that solves density $\Delta_0$ instances with probability $1-o(1/\log{r})$.

  \begin{figure}
      \centering
      \begin{tikzpicture}
        \node at (-0.35,6.25) {$M$};
        \draw (0,6) rectangle (4.5,6.5);
        \draw (0,6) rectangle (0.5,6.5);
        \draw[fill=gray] (0.5,6) rectangle (1,6.5);
        \draw (1,6) rectangle (1.5,6.5);
        \draw (1.5,6) rectangle (2,6.5);
        \draw[fill=gray] (2,6) rectangle (2.5,6.5);
        \draw (2.5,6) rectangle (3,6.5);
        \draw (3,6) rectangle (3.5,6.5);
        \draw[fill=gray] (3.5,6) rectangle (4,6.5);
        \draw (4,6) rectangle (4.5,6.5);

        \draw[->] (0.25,6) -- (0.25,0.5);
        \draw[->] (0.75,6) -- (0.75,0.5);
        \draw[->] (1.25,6) -- (1.25,0.5);
        \draw[->] (1.75,6) -- (1.75,0.5);
        \draw[->] (2.25,6) -- (2.25,0.5);
        \draw[->] (2.75,6) -- (2.75,0.5);
        \draw[->] (3.25,6) -- (3.25,0.5);
        \draw[->] (3.75,6) -- (3.75,0.5);
        \draw[->] (4.25,6) -- (4.25,0.5);

        \node at (6.3,5.25) {$\beta_1$};
        \node at (6.3,4.75) {$\beta_2$};
        \node at (6.3,4.1) {$\vdots$};
        \node at (7.9,3.25) {$\beta_i \sample \left[ 2^{k \lg r \left(1/\Delta_0 - 1/\Delta\right)} \right]$};
        \node at (6.3,2.25) {$\vdots$};
        \node at (6.3,1.25) {$\beta_r$};
        \draw (5.5,5.25) -- (0.25,5.25);
        \draw (5.5,4.75) -- (0.25,4.75);
        \draw (5.5,4.25) -- (0.25,4.25);
        \draw (5.5,3.75) -- (0.25,3.75);
        \draw (5.5,3.25) -- (0.25,3.25);
        \draw (5.5,2.75) -- (0.25,2.75);
        \draw (5.5,2.25) -- (0.25,2.25);
        \draw (5.5,1.75) -- (0.25,1.75);
        \draw (5.5,1.25) -- (0.25,1.25);

        \node at (0.25,5.25) {$\oplus$};
        \node at (0.75,4.75) {$\oplus$};
        \node at (1.25,4.25) {$\oplus$};
        \node at (1.75,3.75) {$\oplus$};
        \node at (2.25,3.25) {$\oplus$};
        \node at (2.75,2.75) {$\oplus$};
        \node at (3.25,2.25) {$\oplus$};
        \node at (3.75,1.75) {$\oplus$};
        \node at (4.25,1.25) {$\oplus$};

        \draw (6,5.5) rectangle (5.5,5);
        \draw (6,5) rectangle (5.5,4.5);
        \draw (6,4.5) rectangle (5.5,4);
        \draw (6,4) rectangle (5.5,3.5);
        \draw (6,3.5) rectangle (5.5,3);
        \draw (6,3) rectangle (5.5,2.5);
        \draw (6,2.5) rectangle (5.5,2);
        \draw (6,2) rectangle (5.5,1.5);
        \draw (6,1.5) rectangle (5.5,1);
        
        \node at (-0.45,0.25) {$M_0$};
        \draw (0,0) rectangle (4.5,0.5);
        \draw (0,0) rectangle (0.5,0.5);
        \draw[fill=gray] (0.5,0) rectangle (1,0.5);
        \draw (1,0) rectangle (1.5,0.5);
        \draw (1.5,0) rectangle (2,0.5);
        \draw[fill=gray] (2,0) rectangle (2.5,0.5);
        \draw (2.5,0) rectangle (3,0.5);
        \draw (3,0) rectangle (3.5,0.5);
        \draw[fill=gray] (3.5,0) rectangle (4,0.5);
        \draw (4,0) rectangle (4.5,0.5);

        \draw[->] (4.5,0.25) -- (6.5,0.25);
        \node at (7,0.25) {$\A_0^\amp$};
      \end{tikzpicture}
      \caption{Illustration of how to lift the hardness amplification to all densitites $\Delta \leq 1$ for the case of modular $k$-SUM. We sample random $\beta_i$ from an appropriate range and add to each entry $M[i]$ of the original array, and give the result is given to $\A_0^\amp$.}
      \label{fig:modular ksum}
  \end{figure}
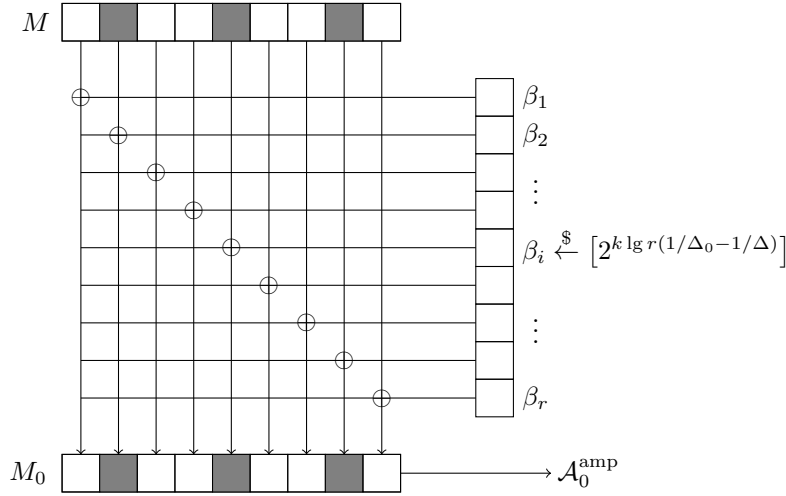
  
  Now we are ready to construct the final algorithm $\B$. Given a density $\Delta$ instance $M$ with dimensions $m\times r$, our algorithm does the following (see also \cref{fig:modular ksum} for an illustration).
  \paragraph{Algorithm $\B(M)$}
  \begin{enumerate}
    \item Define $t:=(\log r)^{\alpha+2}$
    \item Repeat $t$ times:
    \begin{enumerate}[label*=\arabic*.]
      \item Initialize a new empty array $M_0$ of size $r$.
      \item For each $1\leq i\leq r$:
      \begin{enumerate}[label*=\arabic*.]
          \item Sample $\beta_i$ uniformly from $[2^{k\log{r}(1/\Delta_0-1/\Delta)}]$
          \item Set $M_0[i] := M[i]+\beta_i2^{k\log{r}/\Delta}$
      \end{enumerate}
      \item Call $\A_0^\amp$ on the new matrix $M_0$.
      \item If a solution was returned in the last step, check if that is also a solution for $M$. If so, return it.
    \end{enumerate}
  \end{enumerate}

  Observe that $M$ is guaranteed to have at least one solution $S$ since it is sampled from $D_1^\Delta$. The probability of $S$ being a $\kSUM$ solution for $M_0$ is 
  \begin{align*}
      \pr{\sum_{i\in S}M_0[i] = 0} = &\pr{2^{k\log{r}/\Delta_0}\big|\sum_{i\in S}M_0[i]} \\
      =&\pr{2^{k\log{r}/\Delta_0}\big|\sum_{i\in S}\paren{M[i]+\beta_i2^{k\log{r}/\Delta}}}\\
      =&\pr{2^{k\log{r}/\Delta_0}\big|\sum_{i\in S}M[i]+2^{k\log{r}/\Delta}\sum_{i\in S}\beta_i}
      \intertext{Since $S$ is a $\kSUM$ solution in $M$, the first term in the above sum is divisible by $2^{k\log{r}/\Delta}$. We can therefore cancel out $2^{k\log{r}/\Delta}$ on both sides to get}
      =&\pr{2^{k\log{r}(1/\Delta_0-1/\Delta)}\bigg|\paren{\sum_{i\in S}M[i]}\bigg/2^{k\log{r}/\Delta}+\sum_{i\in S}\beta_i} \\
      =&\pr{\sum_{i\in S}\beta_i\equiv-\paren{\sum_{i\in S}M[i]}\bigg/2^{k\log{r}/\Delta} \mod{2^{k\log{r}(1/\Delta_0-1/\Delta)}}}
      \intertext{Since each $\beta_i$ is uniformly random, so is their sum $\beta:=\sum_{i\in S}\beta_i$. Denoting the right hand side of the above congruence by $t$ and letting $x:=2^{k\log{r}(1/\Delta_0-1/\Delta)}$, we have}
      =&\prob{\beta\sample[x]}{\beta\equiv t\mod{x}} = \frac{1}{x} = 2^{k\log{r}(1/\Delta-1/\Delta_0)}.
  \end{align*}
  \noindent As in the proof of \cref{thm: vksum success amplification}, we can bound the probability $p$ that at least one iteration of the inner loop ran $\A^\amp_0$ on an array where $S$ was a solution by 
  \begin{align*}
    p &\geq 1-\left(1-2^{\frac{k\log(r)}{\log{q}}\left(\frac{1}{\Delta}-\frac{1}{\Delta_0}\right)}\right)^t = 1-\left(1-r^{k\,\left(\frac{1}{\Delta}-\frac{1}{\Delta_0}\right)}\right)^t =1-o\left(\frac{1}{\log{r}}\right)
  \end{align*}
  When $M_0$ does have a solution, it is easy to see that it has the same distribution as $D_1^{(\Delta_0)}$ and hence, $\A^\amp_0$ returns a solution with probability $1-o(1/\log{r})$. Taking a union bound, we conclude that with probability $1-o(1/\log{r})$, we call $\A_0^\amp$ on an array with a $\kSUM$ solution at least once \emph{and} it returns that solution. Any solution to $M_0$ is always a solution to $M$, and our algorithm therefore returns it. We thus have the required success probability.

  The runtime of $\B$ is clearly $t$ times the runtime of $\A_0^\amp$. Since $\A_0^\amp$ runs in $\Otilde(T)$ and $t=\mathcal{O}(\polylog(r))$, our algorithm also runs in $\Otilde(T)$ time.
\end{proof}


\pnote{Replaced the corollary with the following}
\begin{corollary}[Strong Hardness of $k$-SUM]
    \label{cor:ksum-amp}
    For any constant $k \geq 3$, the average-case $k$-SUM conjecture (\cref{conj:ksum-intro}) implies that any algorithm that, given $r$ uniformly random integers from $[-r^k,r^k]$ can find a set of $k$ that sum to $0$ with probability $\Omega({1}/{\polylog(r)})$ takes time at least $r^{\ceil{k/2}-o(1)}$.
\end{corollary}

\begin{proof}
    This follows from setting $\Delta=1$ and $k = O(1)$ in \cref{thm: k-Msum success amplification}, using the reduction of $\ksum$ to modular $\ksum$ presented in Theorem 4.5 in \cite{dinur_keller_klein}, and taking the contrapositive.
\end{proof}






\section{Implications to Public-Key Encryption } \label{sec:pke}
In this section, we propose a class of public-key bit encryption schemes based on the planted search $k$-XOR problem and the learning parity with noise (LPN) problem. By instantiating the class appropriately, we strike various trade-offs between the hardness required for LPN and the densities at which we assume planted search $k$-XOR is hard.

\subsection{Preliminaries}
\begin{definition}[Public Key Encryption] \label{def:pke}
    A \emph{Public Key Encryption (PKE)} scheme for a message space $\cM$ consists of PPT algorithms $\pke = (\keygen, \enc, \dec)$ with the following syntax:

\begin{itemize}
    \item $\keygen(1^{\secp}) \xrightarrow{} (pk, sk)$: on input the unary representation of the security parameter $\secp$, generates a public key $pk$ and a secret key $sk$.
    \item $\enc(pk,m) \xrightarrow{} ct$: on input a public key $pk$ and a message $m \in \cM$, outputs a ciphertext $ct$.
    \item $\dec(sk,ct) \xrightarrow{} m$: on input a secret key $sk$ and a ciphertext $ct$, outputs a message $m \in \cM$.
\end{itemize}
The scheme should satisfy the following properties:
    \begin{description}
        \item \textbf{Correctness}. A scheme $\pke$ is \emph{correct} if there exists a negligible function $\negl(.)$ such that for every security parameter $\secp$ and message $\mu \in \cM$ :
        \begin{equation*}
            \Pr[\dec(sk, \enc(pk, \mu)) = \mu] \geq 1 - \negl(\secp) 
        \end{equation*}
        where $(pk, sk) \gets \keygen(1^{\secp})$.
    The scheme is \emph{weakly} correct if the probability of correct decryption is bounded by $1 - o(1)$ instead of $1 - \negl(\secp)$ above.
        
     \item \textbf{CPA Security}. A scheme $\pke$ is \emph{IND-CPA secure} if for any PPT adversary $\sfA$ there exists a negligible function $\negl$ such that: 
        \begin{equation*}
        \label{eq:indcpa}
        \big|\Pr[\pkeGame_{\sfA(\secp)}^{0} = 1] - \Pr[\pkeGame_{\sfA(\secp)}^{1} = 1]\big| \leq  \negl(\secp)
        \end{equation*}
        where $\pkeGame_{\sfA(\secp)}^{b}$ is a game between an adversary $\sfA$ and a challenger $\sfC$ with a challenge bit $b$ defined as follows:
        \begin{itemize}
            \item $\sfC$ samples $(pk, sk) \gets \keygen(1^{\secp})$, and sends $pk$ to $\sfA$.
            \item $\sfA$ chooses $\mu_{0}, \mu_{1} \in \cM$ and sends them to $\cC$.
            \item $\sfC$ computes $ct \gets \enc(pk, \mu_{b})$, and sends $ct$ to $\sfA$.
            \item The adversary $\sfA$ outputs a bit $b'$ which we define as the output of the game.
        \end{itemize}
        The scheme is said to be \emph{weakly} IND-CPA secure if we replace $\negl(\secp)$ with $o(1)$ in the two conditions above. 
    \end{description}    
 \end{definition}  
  \begin{remark}
    Technically speaking, we will obtain a \emph{weak} public-key encryption scheme in the sense that the advantage of the adversary is not negligible in the security parameter but only vanishing. This may be amplified using error-correction with an appropriate hardcore lemma (such as \cite{pke_amplification,holenstein_hardcore}) to get a full-fledged PKE scheme. 
\end{remark}

\begin{definition}[Search LPN Problem]
    An algorithm $\A$ is said to solve the \emph{search Learning Parity with Noise (LPN) problem with noise rate $\eta$} with probability $\epsilon$ if, given $\langle X, Xs+e\rangle$ where $X\gets\F_2^{r\times m}$, $s\gets\F_2^m$, and $e\gets \Ber_\eta^r$, it outputs $s$ with probability at least $\epsilon$, where the randomness is taken over the instance and the random coins used by the algorithm. 
\end{definition}

The best known algorithm for search LPN when $\eta = \Theta(1)$ is by Blum, Kalai and Wasserman~\cite{bkw} that for an $m$-dimensional secret runs in subexponential time $2^{\mathcal{O}(m / \log(m))}$. There is another algorithm by Esser, Kübler and May~\cite{EKM17} that runs in time $2^{\mathcal{O}(\eta m)}$ and thus outperforms the BKW algorithm when the noise-rate is small. We also consider the decision LPN problem that we formally define as follows.
\begin{definition}[Decision LPN Problem]
    The \emph{Decision Learning Parity with Noise (LPN) problem with noise rate $\eta\in(0,1)$} is to distinguish between the following distributions:
\begin{enumerate}
    \item $\langle X,Xs+e \rangle$, where $X\gets\F_2^{r\times m}$, $s\gets\F_2^m$, and $e\gets \Ber_\eta^r$.
    \item $\langle X,y \rangle$, where $X\gets\F_2^{r\times m}$, and $y\gets\F_2^r$.
\end{enumerate}
We say that an algorithm $\A$ solves the decision LPN problem with advantage $\epsilon$ if,
$$
    \Big \lvert \Pr[\A(X,Xs+e) = 1] - \Pr[\A(X,y) = 1] \Big \rvert \geq \epsilon,
$$
where the randomness is taken over the instance and the random coins used by $\A$.
\end{definition}
The search and decision LPN problems are known to be polynomially equivalent, as showed in \cite{lpn_search_decision}. We restate their result as follows.
\begin{lemma}[Search-to-Decision LPN  \cite{lpn_search_decision}]\label{lemma:lpn search decision equiv}
    If there is an algorithm that runs in time $T$ and solves the decision LPN problem with advantage $\epsilon$, then there is an algorithm that runs in time $\mathcal{O}\left(\frac{T\,m \log m}{\epsilon^2}\right)$ and solves the search LPN problem with probability $\epsilon/4$.
\end{lemma}

\begin{definition}[Hardness of LPN]
\label{def:lpn-search-hard}
    We say that \emph{$\eta$-noise search LPN is $T(m)$-hard} if any algorithm $\A$ that runs in time at most $T(m)$ has success probability at most $o(1/T(m))$ in solving the search LPN problem for secrets of size $m$ with noise rate $\Omega(\eta)$.
\end{definition}
We can leverage the equivalence of \cref{lemma:lpn search decision equiv} to translate hardness of LPN into a bound on the advantage of an algorithm that solves decision LPN.

\begin{corollary}\label{cor:decision lpn advantage}
    If $\eta$-noise Search LPN is $(T(m)^3\cdot m\log{m})$-hard, then there is no algorithm that runs in time $\mathcal{O}(T(m))$ and solves the decision LPN problem with advantage $\Omega(1/T(m))$.
\end{corollary}
\begin{definition}[Hardness of Planted Search $k$-XOR]\label{def:search-k-xor hardness}
    We say that planted search \emph{$k$-XOR is $T(r)$-hard at density $\Delta$} if any algorithm that runs in time at most $T(r)$
    has a success probability at most $(1-\Omega(1/\log{r}))$ in solving the planted search $k$-XOR problem at density $\Delta$. 
\end{definition}

Similarly, we put together \cref{thm:search-decision-IN,cor:kxor-amp-1} to translate mild hardness of planted search $k$-XOR (\cref{def:search-k-xor hardness}) into a bound on the advantage of an algorithm that solves decision $k$-XOR with any constant advantage.
\begin{corollary}\label{cor:decision kxor advantage}
    For any $k = o(\log{r}/\log\log{r})$, if planted search $k$-XOR is $T(r)$-hard at density $\Delta$, then any algorithm that runs in time $\left(T(r)/r^{2}\right)$
    solves the decision-$k$-XOR problem at density $\Delta$ with advantage at most $o(1)$.
\end{corollary}

\subsection{Construction}
In this section, we propose a class of cryptosystems $PKE_{\eta,k,\ell,m}$ that is parameterized by four functions $\eta : \Nat \rightarrow (0,1)$ and $k,\ell,m : \Nat \rightarrow \Nat$, where for security parameter $r$, $\eta(r)$ is the noise-rate, $k(r)$ is the size of the planted $k$-XOR solution, $m(r)$ is the dimension of the vectors, and $\ell(r)$ is the number of repetitions in the cryptosystem we describe (this is necessary to satisfy correctness). We will show that $PKE_{\eta,k,\ell,m}$ can be instantiated in various ways to obtain public-key encryption, by striking a trade-off between the assumed hardness of LPN and the densities at which planted search $k$-XOR is assumed hard. We can also alternatively trade off the density for the size of the public key. This allows us to obtain public keys of size $r^{1+o(1)}$ at noise rate $\eta=\Theta(1)$ from $2^{m^{0.5}}$-hardness, while the previous best-known construction of PKE from the same assumptions used a public key of size $r^2$. For given $\eta,k,\ell,m$, the cryptosystem is defined as follows.
\paragraph{Key Generation $\KeyGen(1^r)$.}
\begin{enumerate}
    \item  Let $pk \sample \F_2^{m(r) \times r}$.
    \item Choose a random set $sk \subseteq [r]$ with $|sk|=k(r)$.
    \item Let $i \in sk$ be the smallest index and let $pk_{i} \gets - \displaystyle\sum_{\underset{j \neq i}{j \in sk}} pk_{j}$.
    \item Return $\langle pk, sk \rangle$.
\end{enumerate}

\noindent Here, we will interpret each $sk \in \F_2^r$ as the characteristic vector for the set $sk \subseteq [r]$.

\paragraph{Encryption $\Enc(pk, b)$.}
\begin{enumerate}
    \item If $b=0$.
    \begin{enumerate}
        \item[1.1.] Return $C \sample \F_2^{\ell(r) \times r}$.
    \end{enumerate}
    \item If $b=1$.
    \begin{enumerate}
        \item[2.1.] Let $S \sample \F_2^{\ell(r) \times m(r)}$.
        \item[2.2.] Let $E \gets \Ber_{\eta(r)/2}^{\ell(r) \times r}$.
        \item[2.3.] Return $C \gets S\, pk + E$.
     \end{enumerate}
\end{enumerate}

\paragraph{Decryption $\Dec(sk, C)$.}
\begin{enumerate}
  \item Return $0$ if $\lVert C \, sk \rVert_0 > \ell(r) \left(\frac12 - \frac{{(1-\eta(r))}^{k(r)}}{4}\right)$; else return 1.
\end{enumerate}

Moving forward, we will implicitly fix the value of $r$ and thus treat $\eta,k,\ell,m$ as constants such that we may drop their dependence on $r$.

\begin{lemma}\label{lemma:epsilon-correctness}
    Decryption succeeds with probability at least $(1-\eps)$ when $\ell \geq 32 \left(1-\eta\right)^{-2 k}\ln(1/\varepsilon)$.
\end{lemma}
\begin{proof}
  It will be convenient for our proof to think of the error of $\Ber_{\eta/2}$ added to each bit as $\Ber_\eta \cdot \Ber_{1/2}$ -- for each location, with probability $\eta$, XOR it with a uniformly random bit. 

  Suppose for $\ell=1$ that we encrypt $b=0$ such that $c$ is a uniformly random vector from $\{0,1\}^r$. It is clear in this case, as $sk \neq 0^r$, that $sk^\top c$ a uniformly random bit, i.e. if $\langle pk, sk \rangle \gets \KeyGen(r)$, then we get that, $
    \Pr[sk^\top c = 1] = \frac12$. If instead we encrypt $b=1$ and there are no errors in the $k$ positions corresponding to the planted set, then $sk^\top c = 0$. This happens with probability $(1-\eta)^k$. If instead, some error occurred in the planted set, $sk^\top c$ will be uniformly distributed. As such, if $c$ is obtained by encrypting $1$,
        $$
        \Pr[sk^\top c=0] = (1-\eta)^{k} +  \frac{1-(1-\eta)^k}2 = \frac12 + \frac{(1-\eta)^k}{2}.
        $$
        This does not provide useful correctness by itself so we amplify the difference by repeating the process $\ell$ times and take the majority vote. Now suppose we have $\ell$ iterations and consider the case of $b=0$. Let $E_i = c_i^\top sk$, for $i=1\ldots \ell$. Then all $E_i$ are i.id. and satisfy $\E[E_i] = \frac12$. It follows by a Chernoff bound (\cref{lemma:chernoff}) that,
    \begin{align*}
      \Pr[\Dec(\Enc(pk,0),sk)=1] 
        &= \Pr\left[\sum_{i=1}^\ell E_i \leq \ell \left(\frac12 - \frac{(1-\eta)^{k}}{4}\right)\right]
        \leq \textsf{exp}\left({-\left(\frac14\left(1-\eta\right)^{k}\right)^2 \ell / 2}\right).
    \end{align*}
    This gives a decryption error of at most $\varepsilon$ whenever $\ell \geq 32 \left(1-\eta\right)^{-2 k}\ln(1/\varepsilon)$. The case for $b=1$ is identical.
\end{proof}

We aim to show that no adversary can efficiently distinguish between an encryption of zero and an encryption of one, as per the usual definition of semantic security \cite{semantic_security}.

\begin{lemma}[Indistinguishability of Cryptosystem]\label{lemma:pke}
    Suppose the following conditions are satisfied for some choice of $\eta$, $k$, $\ell$, $m$ (as functions of $r$), and $T$ (as a function of $m$):
    \begin{enumerate}
        \item[(1)] $k = o(\log(r)/\log\log(r))$ and $\ell = o(T(m))$.
        \item[(2)] $\eta$-noise Search LPN is $(T(m)^3\cdot m\log{m})$-hard (\cref{def:lpn-search-hard}) for secrets of size $m$.
        \item[(3)] Planted search $k$-XOR is $r^{k/2-o(1)}$-hard (\cref{def:search-k-xor hardness}) at density $\frac{k \log r}{m}$. 
    \end{enumerate} 
    Then any adversary that runs in time $\min\left(T(m),r^{k/2-2}\right)$ has an advantage $o(1)$ in distinguishing an encryption of 0 from an encryption of 1 in $PKE_{\eta,k,\ell,m}$.
\end{lemma}
\begin{proof}
    We proceed using a hybrid argument (see \cref{fig:pke hybrid}) and define a class of distributions $H_b^{(i)}$, each of which generates a public key $pk$ using the distribution $D_b$, generates $i$ vectors of the form $s^\top pk + e^\top$ and $\ell-i$ random vectors from $\F_2^r$. Note that $H_1^{(0)}$ is an encryption of 0 and $H_1^{(\ell)}$ is an encryption of 1. 
    
    By our assumption about the hardness of planted search $k$-XOR and \cref{cor:decision kxor advantage}, any adversary that runs in time $r^{k/2 - 2}$ has advantage $o(1)$ in distinguishing $H_1^{(0)}$ from $H_0^{(0)}$ (respectively, $H_1^{(\ell)}$ from $H_0^{(\ell)}$). 
    
    Next, note that $H_0^{(i)}$ differs from $H_0^{(i+1)}$ only by having replaced one random vector with a vector of the form $s^\top pk + e^\top$ -- this is exactly an instance of decision LPN. Thus, by our assumption about the hardness of LPN and \cref{cor:decision lpn advantage} we have that $H_0^{(i)}$ is $o(1/T(m))$-indistinguishable from $H_0^{(i+1)}$ for each $0 \leq i < \ell$, which combined with condition (1) implies that $H_0^{(0)}$ is $o(1)$-indistinguishable from $H_0^{(\ell)}$. 
    
    We then conclude that $H_1^{(0)}$ and $H_1^{(\ell)}$ are $o(1)$-indistinguishable to any adversary running in time $\min\left(T(m), r^{k/2-2}\right)$, which proves the lemma.
\end{proof}

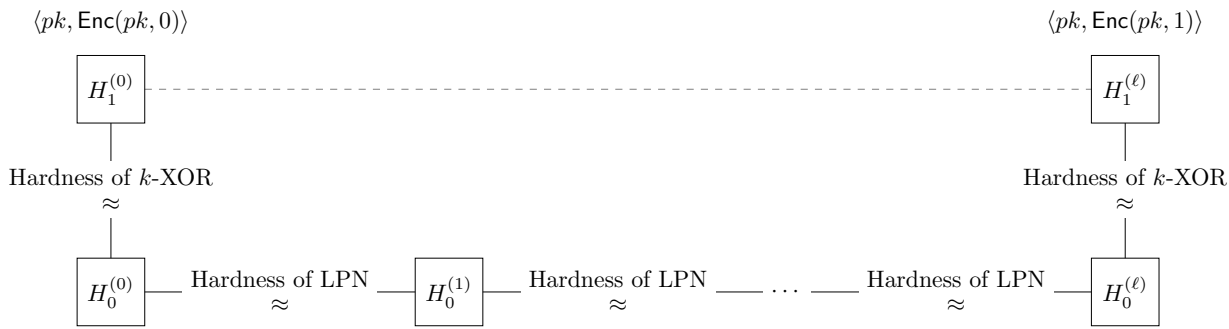
\begin{figure}
    \centering
    \resizebox{\textwidth}{!}{%
    \begin{tikzpicture}
        \node at (0,1) {$\langle pk, \Enc(pk,0)\rangle$};
        
        \draw[dashed,color=gray] (0.5,0) -- (14.5,0);
        
        \draw (-0.5,-0.5) rectangle (0.5,0.5);
        \node at (0,0) {$H_1^{(0)}$};

        \draw (0,-0.5) -- (0,-2.5);

        \draw (-0.5,-2.5) rectangle (0.5,-3.5);
        \node at (0,-3) {$H_0^{(0)}$};
        
        \draw (0.5,-3) -- (4.5,-3);
        
        \draw (4.5,-2.5) rectangle (5.5,-3.5);
        \node at (5,-3) {$H_0^{(1)}$};
        
        \draw (5.5,-3) -- (14.5,-3);
        
        \node[fill=white] at (10,-3) {$\cdots$};
        
        \draw (14.5,-2.5) rectangle (15.5,-3.5);
        \node at (15,-3) {$H_0^{(\ell)}$};

        \draw (15,-0.5) -- (15,-2.5);
        
        \node at (15,1) {$\langle pk, \Enc(pk,1) \rangle$};
        \draw (14.5,-0.5) rectangle (15.5,0.5);
        \node at (15,0) {$H_1^{(\ell)}$};
        
        \node[fill=white] at (0,-1.3) {Hardness of $k$-XOR};
        \node[fill=white] at (0,-1.7) {$\approx$};
        
        \node[fill=white] at (15,-1.3) {Hardness of $k$-XOR};
        \node[fill=white] at (15,-1.7) {$\approx$};
        
        \node[fill=white] at (2.5,-2.8) {Hardness of LPN};
        \node[fill=white] at (2.5,-3.2) {$\approx$};
        
        \node[fill=white] at (7.5,-2.8) {Hardness of LPN};
        \node[fill=white] at (7.5,-3.2) {$\approx$};
        
        \node[fill=white] at (12.5,-2.8) {Hardness of LPN};
        \node[fill=white] at (12.5,-3.2) {$\approx$};
    \end{tikzpicture}
    }
    \caption{Structure of the hybrid argument. The distribution $H_b^{(i)}$ generates a public key $pk$ using the distribution $D_b$, generates $i$ vectors of the form $s^\top pk + e^\top$ and $\ell-i$ random vectors from $\F_2^r$. The distribution $H_1^{(0)}$ thus corresponds to an encryption of 0, while $H_1^{(\ell)}$ is an encryption of 1.}
    \label{fig:pke hybrid}
\end{figure}
Using \cref{lemma:pke}, we can identify various tradeoffs between the hardness assumed for LPN and the densities at which planted search $k$-XOR is assumed to be hard.

\begin{theorem}[PKE from $k$-XOR and LPN]\label{thm:pke_from_lpn}
  There are values of $\eta=\Theta(1)$, $k=\mathcal{O}(\log r)$, $m=r^{o(1)}$, and $\ell = \poly(r)$ such that the cryptosystem $PKE_{\eta,k,\ell,m}$ is a (weak) public-key bit encryption scheme, with all operations running in time $r^{1+o(1)}$, if any of the following conditions are satisfied:
    \begin{enumerate}
        \item[(1)] Constant-noise LPN is $2^{m^{0.5}}$-hard; or,
        \item[(2)] Constant-noise LPN is $2^{m^c}$-hard (for any constant $c>0$) and planted search $k$-XOR is $r^{k/2-o(1)}$-hard at every density $\frac1{\polylog(r)}$; or,
        \item[(3)] Constant-noise LPN is $\poly(m)$-hard and planted search $k$-XOR is $r^{k/2-o(1)}$-hard at every density $\frac{1}{r^{o(1)}}$.
    \end{enumerate}
\end{theorem}
\begin{proof}
Suppose we write the parameters to the cryptosystem as follows (where $\alpha,\beta,\epsilon > 0$ are parameters we will define later). 
\begin{align*}
    \eta = \frac13 &&
    k = \epsilon\,(\log r)^\alpha &&
    m = (\log r)^\beta &&
    \ell = 32\cdot1.5^{2k}\,(\log r)^2
\end{align*}
Here, we will always set $\alpha > 0$ so that $k$ is super-constant. The constraint $\ell = \textsf{exp}(k)\, \omega(\log r)$ is needed by \cref{lemma:epsilon-correctness} to ensure a vanishing decryption error which, moving forward is guaranteed by our choice of parameters. Note that in order for the ciphertext to be superconstant and of polynomial size, we need to have $1/\log \log r < \alpha \leq 1$. Also note that the density in our case is $\frac{k \log r}{m} = \epsilon\, (\log r)^{\alpha + 1 - \beta}$. 

Now to show conclusion (2), suppose LPN is $2^{m^c}$-hard for some constant $c>0$. Suppose we wish to instantiate the cryptosystem at constant density $>1$. In this case, since $k=\omega(1)$, it follows by \cref{thm:stat_indist} that $H_1^{(0)}$ (respectively, $H_1^{(\ell)}$) is, in fact, statistically indistinguishable from $H_0^{(0)}$ (respectively, $H_0^{(\ell)}$). Note that we have $T(m) \approx 2^{m^c/3} = 2^{(\log r)^{\beta c}/3}$ which must satisfy $\beta c > 1$ for $T(m)$ to be superpolynomial. Note that in order to have density $\epsilon$, we need that $\beta = \alpha + 1$ and thus the cryptosystem is only secure if $c \geq \frac12$. Setting $c = \frac12$, $\alpha=1-\delta$, and $\beta = 2-\delta$ for a small constant $\delta >0$ shows conclusion (1) of the theorem.

Suppose instead we allow the density to be sub-constant such that we may circumvent this lower bound. Again with $T(m) \approx 2^{m^c/3} = 2^{(\log r)^{\beta c}/3}$, set $\beta$ such that $T(m) = r^{k/2}$. That is, $(k)\log(r)/2 \approx m^c/3$. By \cref{lemma:pke}, the cryptosystem satisfies indistinguishability against $r^{k/2-2}$-time adversaries if we assume planted search-$k$-XOR is hard at density $\Delta =  \frac{k \log r}{m} \approx \frac{2 m^{c-1}}{3}$. Using the fact that $m = k \log(r) /\Delta$, this is:
$$
    \Delta = \left(\frac{2}{3 \left(k \log(r)\right)^{1-c}} \right)^{\frac1c} = \frac{1}{\polylog(r)},
$$
for any constant $\alpha$ and $\beta$ chosen to satisfy the above conditions, say $\alpha = 1$ and $\beta \approx 2/c$. Applying \cref{lemma:pke} with these parameters now gives conclusion (2). 

Finally, to show conclusion (3), suppose instead that $T(m) = m^{c}$ for some $c=\omega(1)$ but $c \ll m^{0.01}$. In this case, we pick $m$ such that $m^c = r^{k/2}$, and hence using $m=k \log(r)/\Delta$, we need hardness of $k$-XOR at some $\Delta \approx \frac{k \lg r}{r^{\frac{k}{2c}}}$. Now let $k = \min(\sqrt{c},\sqrt{\log{r}})$, then we can choose $m=r^{o(1)}$ such that $\Delta = \frac{k \lg r}{r^{o(1)}}$, for which again, we know planted search $k$-XOR is hard by assumption. Applying \cref{lemma:pke} now gives the conclusion.

\end{proof}

\noindent Observe that the size of the public key size is $m\cdot r$ bits. The encryption time is $\mathcal{O}(\ell\cdot m\cdot r)$, while decryption takes time $\mathcal{O}(\ell\cdot r)$. In all of the instantiations in \cref{thm:pke_from_lpn}, the above quantities are at most $r^{1+o(1)}$.

We note that \cref{thm:pke_from_lpn}(1) was also previously shown by Yu and Zhang \cite{lpn_pke}, who use a different construction to build public-key encryption from $2^{m^{0.5}}$-hardness of constant-noise LPN. Their cryptosystem has a public key of size of $r^{2}$, while the above cryptosystem has a public key of size $mr = r^{1+o(1)}$, which is substantially smaller for large $r$.

\iftrue
\section*{Acknowledgments}
Nikolaj thanks Ivan Damg{\aa}rd for discussions related to the public-key encryption scheme. We also thank Rachel Lin for encouraging us to think about the case of $k$-SUM with super-constant $k$, and Eldon Chung for helpful discussions.
\fi
\clearpage
\bibliographystyle{alpha}
\bibliography{refs}

\clearpage
\appendix
\section{PKE from LWE and $k$-SUM}
\label{sec:pke-lwe}

In this section, we show how to build a public key encryption scheme from hardness of vector $k$-SUM and the Learning With Errors (LWE) problem with super-constant modulus to noise ratio.
\subsection{Reduction from $k$-SUM to Vector $k$-SUM}
We first show via a reduction that the hardness of decision vector $k$-SUM may be based on the hardness of decision $k$-SUM.

\begin{restatable}{lemma}{vectorksum}
\label{reduction-vkSUM}
    If there is  an algorithm that solves decision {\vkSUM} at density $\Delta$ in time $T(r)$ with advantage $\eps$, then there is an algorithm that solves decision {\kSUM} at density $\Delta$ in time $O(T(r)\cdot k^m)$ with advantage $\eps$.
\end{restatable}


\begin{proof}
    Assume that we have an oracle access to an adversary $\cA$ that solves decision {\vkSUM}. It follows from definition of decision {\vkSUM} that $\cA$ takes as input $r$ vectors of length $m$ each of whose elements are members of $\Z_q$.

    For sake of simplicity, we assume that we are trying to solve decision {\kSUM} modulo $q^m$, and that $q$ is a prime. It is easy to verify that a density $\Delta$ instance consists of $r$ elements.

    Our first step involves expressing each element in base $q$. This transforms an element of $\Z_{q^m}$ into a vector of size $m$ with each element belonging to $\Z_q$. Observe that if $k$ elements added up to zero, each element of the sum of their image vectors must be within $[-(k-1), 0] \mod{q}$ since the carry at each position must be less than the number of summands. This suggests the following algorithm.

    \paragraph{Algorithm $\cB_{1}(x_1,x_2,\ldots,x_r)$}
    \begin{enumerate}
        \item Create matrix $Y$ of size $m\times r$ where $Y[i][j]$ is the $i$th digit of the base $q$ representation of $x_j$.
        \item For each vector $v\in[0, k-1]^m \mod{q}$:
        \begin{enumerate}
            \item[2.1.] Create a new matrix $Y_v$ where $Y_v[i][j] = Y[i][j] + v_i/k$
            \item[2.2.] Call $\cA$ on $Y_v$. Say it returns $b'$
            \item[2.3.] If $b' = 1$ return $1$
        \end{enumerate}
        \item Return $0$
    \end{enumerate}
    Let us assume that the input $X$ is a planted instance of decision {\kSUM} at density $\Delta$. Let $S$ be the set of planted indices. Observe that there exists some particular $v\in[0, k-1]^m \mod{q}$ such that $\sum_{j\in S}Y[i][j] + v[i] = 0 \mod{q}$ for all $i\in[m]$ ($v$ can be thought of as the carry vector). Since $|S|=k$, the above equation implies $\sum_{j\in S}Y_v[i][j]= 0 \mod{q} \hspace{4 pt}\forall i\in[m]$. Furthermore, it is easy to see that $Y_v$ is sampled from the planted distribution on $\Z_q^m$. Let us further assume that the input $X$ is a non-planted instance of decision {\kSUM} at density $\Delta$. This means that all the elements   $\{x_{i}\}_{i \in r]}$ are sampled uniformly at random $\implies Y_{v}$ is identical to being sampled uniformly. 

    By the above, we have shown that the distribution from which $Y_{v}$ is sampled is identical to the distribution from which an instance of decision {\kSUM} is sampled in both planted and non-planted cases respectively. Hence, if $\cA$ solves decision {\vkSUM} with advantage $\epsilon$ then $\cB_{1}$ solves decision {\kSUM} with success probability $\epsilon$.

    The runtime of this algorithm is $k^m$ times the runtime of $\cA$ which simplifies to $O\left(k^{m}\,T\right)$.
\end{proof}
\noindent
\textbf{Note}: At low densities ($\Delta \leq \frac{1}{r^{0.5+\epsilon}\cdot \log{q}}$) where $\epsilon \geq \frac{3}{k-3}$, decision {\vkSUM} can be solved in time $o(r^{\ceil{\frac{k}{2}}})$ by the algorithm described in Lemma \ref{thm:faster-algo-kxor}.

\begin{corollary}\label{cor:vk-SUM-hardness}
    If planted search {\kSUM} problem is hard at density $\Delta$, then there is no algorithm that runs in time $\mathcal{O}\left(r^{k/2\,(1-\Omega(1))} \, k^{-m}\right)$ and solves the decision {\vkSUM} problem at density $\Delta$ with advantage $o(1)$.
\end{corollary}
\noindent
Note that this follows from conjectured hardness of the planted $k$-SUM problem, \cref{thm:search-decision-IN}, and the above reduction.
\subsection{Reduction from Vector $k$-SUM to Targeted Vector $k$-SUM}
\label{sec:targeted-ksum}
\begin{definition} [Targeted Vector $k$-SUM]
  \label{def:targeted-kSUM}
  For any $r \in \N$, $k\in\Nat$, $q \in \N$, $\Delta \in \R$, an algorithm $\A$ is said to solve the \emph{{\tvkSUM}} with success probability $\eps$ if for both $b\in\{0,1\}$,
\begin{align*}
  \underset{X \sim {\sf{D}}_b^{(r)}}{\Pr}\left[\A(X) = b\right] \geq \eps.
\end{align*}
If $\eps = 1-o(1)$, we simply say that $\A$ solves the {\tvkSUM}.
\end{definition}
\noindent
where the distributions are defined below:
\paragraph{Distribution ${\sf{D}}_0^{(r)}$}
\begin{enumerate}
  \item Sample $r$ group elements $X_1, X_2, \ldots, X_{r}$ i.i.d. uniformly at random from   $\G_{\textsc{vector-$(q,k)$-SUM}}^{(\Delta)}$
  \item Return $(X_{1},X)$
\end{enumerate}
\paragraph{Distribution ${\sf{D}}_1^{(r)}$}
\begin{enumerate}
  \item Sample $r-1$ group elements $X_2, \ldots, X_{r}$ i.i.d. uniformly at random from $\G_{\textsc{vector-$(q,k)$-SUM}}^{(\Delta)}$
  \item Choose a random set $S \subseteq [2,r]$ with $|S|=k-1$.
  \item Compute $X_{1} := \sum_{{j \in S}} X_j$
  \item Return $(X_{1},~X)$
\end{enumerate}
\noindent
Note that {\tvkSUM} as defined above is a decision problem. We do not prepend the word 'decision' to it because we only use the decision version of  {\tvkSUM} in this paper and it's use is restricted to the appendix.
\begin{lemma}\label{reduction-tkSUM}
    If  there is an algorithm that solves the targeted vector $k$-SUM problem at density $\Delta$ in time $T$ with advantage $\epsilon$, then there exists an algorithm that solves decision vector {\kSUM} at density $\Delta$ in  expected time $r\Theta(T)$ with advantage $\epsilon$.
\end{lemma}
\begin{proof}
Below, we give a reduction (algorithm $\cB_{2}$) from {\tvkSUM} to {\vkSUM}.
\paragraph{Algorithm $\cB_{2}(x_1,x_2,\ldots,x_r)$}
\begin{enumerate}
   \item Repeat atmost $r$ times:
        \begin{enumerate}
            \item Pick a permutation $\pi:[r] \xrightarrow{} [r]$ uniformly. Let $x' := [x_{\pi(1)},\ldots,x_{\pi(r)}]$.
        \item Call $\cA$ on the vector $[x'_{2},\ldots,x'_{r}]$ and the target $x'_{1}$. Say it returns value $b$.
        \item If $b = 1$ then return $1$ (Indicating a planted set)
        \end{enumerate}
        \item Return $0$
    \end{enumerate}
    We can re-interpret the decision {\vkSUM} problem defined at density $\Delta$ as follows:
\begin{enumerate}
    \item Pick the (planted) index $i \sampleU [r-k+1]$.
    \item Pick a set $S$ uniformly such that from $S \subseteq [r]\setminus [i]$ and $|S| = k-1$.
    \item Pick $a_{1},\ldots a_{r}$ for all indices uniformly at random except $i$.
    \item Let $x_{i} = -\sum_{j \in S}x_{j}$ (or) $x_{i} \sampleU \Z_{q}^{m}$
    \item Output: $[x_{1},\ldots,x_{r}]$.
\end{enumerate}
The {\tvkSUM} problem at density $\Delta$ can be re-interpreted as follows:
\begin{enumerate}
    \item Pick a set $S$ uniformly such that from $S \subseteq [2,r]$ and $|S| = k-1$.
    \item Pick $x_{2},\ldots x_{r}$ for all indices uniformly.
    \item Let $x_{1} = \sum_{j \in S}x_{j}$ (or) $x_{1} \sampleU \Z_{q}^{m}$
    \item Output: $([x_{2},\ldots,x_{n}], x_{1})$.
\end{enumerate}
Notice that the distributions of decision {\vkSUM} and {\tvkSUM} are equivalent at density $\Delta$, conditioned on the fact that (planted) index $i$ sampled in {\vkSUM} is equal to $1$. And by re-permuting the indices after every iteration in our reduction, we ensure that the planted index gets permuted to the location $1$ (this happens in expected no.of iterations = $r$).
\end{proof}
\begin{corollary}\label{cor:tvk-SUM-hardness}
    If planted search {\kSUM} is hard at density $\Delta$, then there is no algorithm that runs in time $\mathcal{O}\left(r^{k/2\,(1-\Omega(1))} \, k^{-m}\right)$ and solves the {\tvkSUM} problem at density $\Delta$ with advantage $o(1)$.
\end{corollary}
\noindent
Note that the \cref{cor:tvk-SUM-hardness} directly follows from \cref{cor:vk-SUM-hardness} and the above reduction.
\subsection{Construction of Public Key Encryption}

\paragraph{Preliminaries.} Below we define some preliminaries needed for our construction.

\begin{lemma}(Leftover Hash Lemma)\label{lhl}(\cite{HILL,goldwasser2010robustness})
Let $(X,Z) \in \mathcal{X} \times \mathcal{Z}$ be any joint random variable over $\Z_{q}^{n}$ with min-entropy $H_{\infty}(X|Z) \geq k$. For any $\epsilon > 0$ and $m \leq \frac{k-2\log(\frac{1}{\epsilon})-O(1)}{\log{q}}$, the joint distribution of ($Z,C,C \cdot s$) where $C \sampleU \Z_{q}^{m \times n}$ is uniformly random and $s \in \mathcal{X}$ is $\epsilon$-close to the uniform distribution over $(Z,\Z_{q}^{m \times n} \times \Z_{q}^{m})$. 
\end{lemma}

Regev \cite{regev05} defined a natural distribution over lattices called the discrete Gaussian distribution, parametrized by a scalar $\alpha > 0$. 
We additionally need a bound on samples drawn from the discrete Gaussian distribution.
 \begin{lemma}\cite{lyubashevsky2012lattice}
 \label{lem:dgsizebd}
Let $\chi_\sigma$ be the discrete Gaussian distribution on $\mathbb{Z}$ with parameter $\sigma$. Then, for any $k>0$, $\Pr[|z|> k\sigma; z\leftarrow \chi_{\sigma}] \leq 2e^{(- k^2/2)}$.
 \end{lemma}

 \begin{definition}[\lwe\ Problem]\label{def:lwe}
	Let $r$ be the security parameter, let $m=m(\secp)$, $n= n(\secp)$ and $q=q(\secp)>2$ be integers,
	and $\chi = \chi(\secp)$ be a distribution over $\Z_q$. The $\LWE(m,n,q,\chi)$ problem over $\Z_{q}$ is to distinguish between the following distributions:
 \begin{itemize}
     \item $( A, s^T A + e^T )$, where $A \gets \Z_q^{m\times n}$, $s \gets \Z_q^{m}$, and $e \gets \chi^n$
     \item $( A, v^T )$ where $A \gets \Z_q^{m\times n}$, and $v \gets \Z_q^n$
 \end{itemize}
We say that an algorithm $\A$ solves the $\lwe$ problem with advantage $\epsilon$ if,
$$
    \Big \lvert \Pr[\A( A, s^T A + e^T ) = 1] - \Pr[\A( A, v^T ) = 1] \Big \rvert \geq \epsilon,
$$
where the randomness is taken over the instance and the random coins used by $\A$.
We say that the $\LWE(m,n,q,\chi)$ problem is \emph{$T$-hard} if no adversary $\sfA$ running in time $T(\secp)$ can distinguish between the above distributions with advantage greater than $1/T(\secp)$.
 
 \end{definition}
 
 

It is known~\cite{regev05,BLPRS13} that if we set $\chi$ to be the discrete Gaussian distribution with parameter $\alpha q$, 
the $\LWE(m,n,q,\chi)$ problem is as hard as solving worst-case lattice problems
such as gapSVP and SIVP with approximation factor $p(m)/\alpha$ for some polynomial $p$.
\begin{theorem}\cite{regev05,Pei09,BLPRS13}\label{lwe-hardness-param}
For any $n = \poly(m)$, any modulus $q \leq 2^{\poly(m)}$, and any (discretized) gaussian
error distribution $\chi$ of parameter $\alpha q \geq 2\sqrt{m}$ where $0 < \alpha < 1$, solving the $\lwe(m,n,q,\chi)$ problem is at least as hard as solving $\sf{gapSVP}_{\gamma}$ on arbitrary $m$-dimensional lattices, for some
$\gamma = \tilde{O}(\frac{m}{\alpha})$
\end{theorem}
Since the best known algorithms for $2^k$-approximation of gapSVP and SIVP run in time
$2^{\tilde{O}(m/k)}$\cite{decade}, \cref{lwe-hardness-param} implies that the best known algorithms that solve $\lwe$ run in time $\Omega(2^{\frac{m}{\log{\frac{m}{\alpha}}}})$.

\paragraph{Construction.}\label{lwe-pke:construction}
 Our public key encryption scheme $\pke = (\keygen, \enc, \dec)$ for message space $\cM=\{0,1\}$, is described as follows: 
\begin{description}
    \item $\pke.\keygen(1^\secp)$: Upon input the unary representation of the security parameter $\secp$, do the following:
    \begin{enumerate}
        \item $A \sampleU \Z_{q}^{m \times n}$
        \item Sample $x \in \{0,1\}^{n}$ uniformly such that  $\norm{x}_{0} = k$
        \item Let $u := Ax$
        \item Output  $\pk := (u,A)$, $\sk := x$
    \end{enumerate}
    \item $\pke.\enc(\pk := (u,A), \mu \in \{0,1\})$: Upon input the public key $\pk$ and the message $\mu$, do the following:
    \begin{enumerate}
        \item Compute $c_{1} := A^{T}s + e $ where $s \sample \Z_{q}^{m},e \gets \chi^{n}$
        \item Compute $c_{2} := u^{T}s+e'+\floor{\frac{q}{2}}\mu$ where $e' \gets \chi$
        \item Output $\ct := (c_{1},c_{2})$
    \end{enumerate}
      \item $\pke.\dec(\pk:= (u, A), \sk := x, \ct := (c_{1},c_{2}))$: Upon input the public key $\pk$, the secret key $\sk$ and the ciphertext $\ct$, do the following:
    \begin{enumerate}
        \item Compute $\mu' := c_{2}-x^{T}c_{1}$.
        \item If $\norm{\mu'}_{2} \leq \frac{q}{4}$ then output $0$ else output $1$.
    \end{enumerate}
\end{description}

\begin{lemma}[Correctness]\label{lwe-pke-correctness}
For any $\secp$, given parameters of the $\pke$ construction are set as mentioned in (\cref{lwe-pke-param-1}). Then, the public key encryption scheme $\pke$ described above is correct.
\end{lemma}
\begin{proof}
Follows from a straight-forward calculation.
\begin{align*}
    \mu' &= c_{2}-x^{T}c_{1}
    \\&= u^{T}s+e'+\floor{\frac{q}{2}}\mu-x^{T}A^{T}s-x^{T}e. 
    \\&= e'+\floor{\frac{q}{2}}\mu-x^{T}e.
    \\&=\floor{\frac{q}{2}}\mu+ (e'-x^{T}e).
\end{align*}
Hence, when $\mu = 0$, the value  is $e'-x^{T}e$, if we show that $\norm{e'-x^{T}e}_{2} \leq \frac{q}{4}$ then that would be sufficient.

Note that $\norm{x^{T}e}_{2} = k (\alpha q)$ with probability $1-o(1)$, since at most $k$ of the entries in the summation are from bounded distribution and the rest of them are zero, we have that $\abs{e'} \leq \alpha \cdot q$ with probability $1-o(1)$(\cref{lem:dgsizebd}), and so, 
$$
    \norm{e'+x^{T}e}_{2} \leq (k+1)\alpha q
$$
Therefore, the correctness holds for the parameters mentioned in \cref{lwe-pke-param-1,lwe-pke-param-2}.
\end{proof}
Next, we prove the security of our scheme under the hardness of LWE and $k$-SUM. 
\begin{lemma}[Security]\label{security:pke}
    Assuming the hardness of planted-search-{\kSUM} at density $\Delta$, and that $\lwe(m,n+1,q,\chi)$(\cref{lwe-hardness-param}) is $2^{\frac{m}{\log{\frac{m}{\alpha}}}}$-hard where the parameters are chosen as described in \cref{sec:pke-params}, the public key encryption scheme $\pke$ satisfies weak $\indcpa$ security (Definition \ref{def:pke}). 
\end{lemma}
\begin{proof}
We prove the theorem via a sequence of hybrids between the challenger and a PPT adversary $\cA$.
\begin{description}
    \item {\sf Hybrid 0}: This is the real world with challenge bit $0$. This is same as $\pkeGame^{0}_{\cA(r)}$
    \item {\sf Hybrid 1}: This world is same as {\sf Hybrid 0} except we sample public key $\pk$ as $(u',A) \sampleU (\Z_{q}^{n \times 1},\Z_{q}^{m\times n})$.
    \item {\sf Hybrid 2}: This world is same as {\sf Hybrid 1} except we sample the cipher-text, i.e, $\ct := (c_{1}',c_{2}')$ as $c_{1}'\sampleU \Z_{q}^{n},c_{2}'\sampleU \Z_{q}$ respectively.
\end{description}
In {\sf Hybrid 2}, the distribution seen by the adversary is independent of the challenge bit $b$ hence the advantage of the adversary in this world is negligible.
\paragraph{Indistinguishability of Hybrids.} We now show that consecutive hybrids are indistinguishable.
\begin{claim}\label{lwe-security-hybrid-1}
   Assume that planted search {\kSUM} is hard at density $\Delta$ (\cref{cor:tvk-SUM-hardness}) for the parameters described in  \cref{sec:pke-params}. Then, {\sf Hybrid 0} and {\sf Hybrid 1} are indistinguishable for any PPT adversary.
\end{claim}
\begin{proof}
Let $\cD$ be a PPT adversary that distinguishes {\sf{Hybrid} 0} from {\sf{Hybrid} 1} with advantage $\epsilon$. Let $\cC$ be the {\tvkSUM} challenger at density $\Delta$. Then we give a reduction $\cB$ that solves {\tvkSUM} with advantage $\epsilon$ as follows:
\begin{enumerate}
    \item $\cC$ outputs $(A,u^{*}) \in (\Z_{q}^{m \times n},\Z_{q}^{m})$.
    \item $\cB$ samples $\mu_{0},\mu_{1} \gets \cM, s \sampleU \Z_{q}^{m}$, $e \gets \chi^{n}$, $e' \gets \chi$
    \item $\cB$ then computes $(\mu_{0},\mu_{1},A,u^{*},c_{1} := A^{T}s + e,c_{2}:=u^{*T}\,s+e'+\floor{\frac{q}{2}}\mu)$
    \item If $\cD$ outputs $b$ then $\cB$ outputs $b$
\end{enumerate}
Notice that when $\cC$ outputs $u' \sampleU \Z_{q}^{m}$ then $\cB$ simulates {\sf Hybrid 1} else, it simulates {\sf Hybrid 0}. Hence, $\cB$ solves {\tvkSUM} at density $\Delta$ with probability equal to the advantage of $\cD$ i.e, $\epsilon$.

From \cref{cor:tvk-SUM-hardness}, $\epsilon = o(1)$. Therefore, by the hardness of {\tvkSUM} at density $\Delta$, we have {\sf Hybrid 0} $\approx_{c}$ {\sf Hybrid 1}
\end{proof}
\begin{claim}\label{lwe-security-hybrid-2}
     Assume that $\lwe(m,n+1,q,\chi)$ is $2^{\frac{m}{\log{\frac{m}{\alpha}}}}$ hard (Theorem \ref{lwe-hardness-param}) for the parameters described in Section \ref{sec:pke-params}. Then, {\sf Hybrid 1} and {\sf Hybrid 2} are indistinguishable for any PPT adversary.
\end{claim}
\begin{proof}
Let $\cD$ be a PPT adversary that distinguishes {\sf{Hybrid} 1} from {\sf{Hybrid} 2} with advantage $\epsilon$. Let $\cC$ be the $\lwe(m,n+1,q,\chi)$ challenger. Then we give a reduction $\cB$ that solves $\lwe(m,n+1,q,\chi)$ problem with advantage $\epsilon$ as follows:
\begin{enumerate}
    \item $\cC$ outputs $\begin{bmatrix}
        c_{1}^{*}\\c_{2}^{*}
    \end{bmatrix} \in \Z_{q}^{n+1}$ such that $c_{1}^{*} \in \Z_{q}^{n}$ and $c_{2}^{*} \in \Z_{q}$
    \item $\cB$ samples $\mu_{0},\mu_{1} \gets \cM, u' \sampleU \Z_{q}^{m}$, $e \gets \chi^{n}$, $e' \gets \chi$
    \item $\cB$ then computes $(\mu_{0},\mu_{1},A,u',c_{1}^{*},c_{2}^{*})$
    \item If $\cD$ outputs $b$ then $\cB$ outputs $b$
\end{enumerate}
Notice that when $\cC$ outputs such that $c_{1}^{*} \sampleU \Z_{q}^{n}$ and $c_{2}^{*} \sampleU \Z_{q}$ then $\cB$ simulates {\sf Hybrid 2} else, it simulates {\sf Hybrid 1}. Hence, $\cB$ solves $\lwe(m,n+1,q,\chi)$ with probability equal to the advantage of $\cD$ i.e, $\epsilon$.

By the hardness of $\lwe(m,n+1,q,\chi)$ (\cref{lwe-hardness-param}) where $\chi$ is parameterised by $\alpha q$, $\epsilon = \negl(\secp)$. 
\end{proof}
\noindent
From the claims (\cref{lwe-security-hybrid-1,lwe-security-hybrid-2}) we have that for any PPT adversary $\cA$: 
\begin{align*}     \big|\Pr[\pkeGame_{\cA(\secp)}^{0} = 1] - \Pr[\pkeGame_{\cA(\secp)}^{1} = 1]\big| \leq  o(1) && \qedhere
\end{align*}
\end{proof}
\paragraph{Parameters.}\label{sec:pke-params}
We now wish to give instantiations of parameters for the PKE construction in \cref{lwe-pke:construction}, and observe the improvements in these parameters due to replacing {\tvkSUM} as a computational analogue of ``LHL'' (Leftover Hash Lemma).  Replacing LHL with {\tvkSUM} is equivalent to assuming mild planted-search-$k$-SUM (\cref{cor:tvk-SUM-hardness}).

First, we enumerate the constraints that need to be satisfied for the correctness and security of PKE construction in \cref{lwe-pke:construction}. We then provide two instantiations of parameters along with the analysis of how assuming hardness of {\tvkSUM} assumption helps us improve the parameters of underlying PKE scheme. For correctness to hold, from Theorem \ref{lwe-pke-correctness} we need $(k+1)\alpha q \leq \frac{q}{4}$ which implies that,
$$
    \alpha \leq \frac{1}{4(k+1)}.
$$
For security wrt .\cref{lwe-security-hybrid-2}, we need that the minimum time taken by an adversary to solve $\lwe(m,n,q,\chi)$ to be super-polynomial in $\secp$, this implies that,
$$
    2^{\frac{m}{\log{\frac{m}{\alpha}}}} \geq r^{\omega(1)}.
$$
In order to use \cref{lwe-hardness-param} we also need to satisfy the condition $q\geq 2\sqrt{m}k$. For security wrt. \cref{lwe-security-hybrid-1}, we need that minimum time taken by an adversary to solve {\tvkSUM} at density $\Delta$ is super-polynomial in $\secp$. Thus, from Lemma \ref{cor:tvk-SUM-hardness}, we require that $n^{\frac{k-2}{2}}\,{k^{-m}} \geq {\secp}^{\omega(1)}$. 
\paragraph{Parameter Analysis.} 
In the construction of PKE in \cref{lwe-pke:construction}, the $|\pk| =n\,m\,\log{q}$, $|\sk| = m\,\log{q}$, encryption time is $m\,k\,\log{q}$, decryption time is $m\,\log{q}$. Below we provide parameters for two settings, one which emphasizes the improvement in the public key size and one which emphasizes the  improvement of approximation factor ($\alpha$) of $\lwe(m,n,q,\chi)$. Note that the first three constraints are to be satisfied irrespective of whether we assume the hardness of {\tvkSUM} (or) not. Therefore, we first set the parameters such that they satisfy the first three constraints then focus only on the constraint that assuming {\tvkSUM} (or) using LHL enforces to analyze the improvement. Note that the constraint enforced by LHL is $m \leq \frac{k\,\log{n}}{\log{q}}$ (\cref{lhl}).

\paragraph{Parameters for Reducing Public Key Size.}
\label{lwe-pke-param-1} We basically fix all the values except $n$ and see how it is affected by assuming {\tvkSUM} hardness in place of LHL. As the discussion is about efficiency in terms of $\pk$-size, $\sk$-size, encryption and decryption time, when  given a choice we picked smaller values of $m$ and $q$.
\begin{itemize}
    \item  $c = \omega(1)$ 
    \item $k = c\,(\log{\log{\secp}}) $
    \item $m = 2c\,\log{\secp}\,(\log{\log{\secp}}) $
    \item $q=8\sqrt{c\,(\log{\secp})\,(\log{\log{\secp}})}\,\,(k+1)$
    \item $\frac{1}{\alpha} =4\,(k+1)$
\end{itemize}
If we invoke LHL i.e, try to satisfy the constraint $m \leq \frac{k\,\log{n}}{\log{q}}$ to argue security then $n \geq \secp^{2}$. On the other hand if we assume hardness of {\tvkSUM} at density $\Delta$ then we  have to satisfy the constraint that $\Delta = \frac{k \log{n}}{m\,\log{q}}$ instead and, assuming $\Delta = \frac{1}{\polylog{(\secp)}}$ means  $n \geq (\secp)^{1/\log{\log{\secp}}}$. This means that the construction using {\tvkSUM} instead of LHL has considerable gain improvement in terms of $\pk$-size, i.e.,
\begin{table}[H]
    \centering
    \begin{tabular}{|c|c|}
    \hline
         Targeted vector $k$-SUM construction & LHL-based construction 
         \\
         \hline
          $|\pk| = \secp^{1/\log{\log{\secp}}}\,(2c\, \log{\secp} \, \log{\log{\secp}})$ & $|\pk| = \secp^{2}\,(2c \, \log{\secp} \, \log{\log{\secp}})$
          \\
          \hline
    \end{tabular}
\end{table}
\noindent
Similarly, we now try to reduce $k$ value. 

\paragraph{Parameters to Improve the Approximation Factor ($\alpha$) of $\lwe(m,n,q,\chi)$.}\label{lwe-pke-param-2} Setting $r=n$, would ease our analysis. As all the other parameters except $k$ are same, this assumption would not affect our qualitative analysis of how $k$ is affected by the computational LHL. Therefore, we retain the values of $q$ and $m$ from the above analysis and, the parameters are as follows:
\begin{itemize}
    \item  $c = \omega(1)$
    \item $m = 2c\,\log{\secp}\,(\log{\log{\secp}}) $
    \item $q=8\sqrt{c\,\log{\secp}\,(\log{\log{\secp}})}\,\,(k+1)$
    \item $\frac{1}{\alpha} =4\,(k+1)$
    \item $n=r$
\end{itemize}
Note that in LHL, one has to satisfy the condition of $m \leq \frac{k\log{n}}{ \log{q}} \implies k \geq 2\,c(\log{\log{r}})\,\log{q}$ whereas when we instantiate our scheme using {\tvkSUM} assumption, the condition to satisfy is $\Delta = \frac{k\, \log{r}}{m \log{q}}$. Notice that, in this case we can set $k$ as any super-constant because $\Delta = \frac{1}{\polylog(\secp)}$ which is fine. Thereby providing us with a small improvement in terms of the value of $k$. Note that, by trading off $n$ and $k$ value, one can improve both approximation factor, $|\pk|$ and encryption time.

\end{document}